\pdfoutput=1
\documentclass[a4paper,11pt,english]{article}

\usepackage[bindingoffset=0.3cm,textheight=22.5cm,hdivide={2.7cm,*,2.7cm}, vdivide={*,22cm,*}]{geometry}
\usepackage{cite}
\usepackage{youngtab}
\usepackage{amsmath,amsfonts,amssymb,babel,slashed,color,empheq,tensor,amsthm}
\usepackage{hyperref}

\makeatletter

\newtheorem{thm}{Theorem}[section]
\newtheorem{lem}[thm]{Lemma}
\newtheorem{corollary}[thm]{Corollary}
\newtheorem{prop}[thm]{Proposition}

\newcommand{\dd}{\mathrm{d}}

\newcommand{\del}{\partial}

\def\nn{\nonumber} 
\def\obar{\overline}
\numberwithin{equation}{section}

\def\a{\alpha}  \def\b{\beta}
 \def\g{\gamma} 
 \def\d{\delta} 

\def\l{\lambda} \def\L{\Lambda}  
    \def\r{\rho}
  \def\t{\tau}

\def\cA{{\cal A}}  \def\cC{{\cal C}} 
\def\cD{{\cal D}}  \def\cF{{\cal F}} 
\def\cG{{\cal G}} \def\cH{{\cal H}} \def\cI{{\cal I}} 
 \def\cK{{\cal K}}  
\def\cM{{\cal M}} \def\cN{{\cal N}} \def\cO{{\cal O}} 
 \def\cQ{{\cal Q}}  
\def\cS{{\cal S}}

\def\R{{\mathbb R}} \def\C{{\mathbb C}} \def\N{{\mathbb N}}

 \def\one{\mbox{1 \kern-.59em {\rm l}}}

\def\msu{\mathfrak{su}}
\def\mso{\mathfrak{so}}

\newcommand{\Tr}{\mathrm{Tr}}

\newcommand{\End}{\mathrm{End}}
\def\hs{\mathfrak{hs}}

\def\Tr{\mbox{Tr}}

\def\und{\underline}

\newcommand{\im}{\mathrm{i}}

\newcommand{\diag}{\rm diag}

\newcommand{\eq}[1]{(\ref{#1})}

 \allowdisplaybreaks[3]

\begin{document}


\parindent=0.4cm
\parskip=0.1cm

\renewcommand{\title}[1]{\vspace{10mm}\noindent{\Large{\bf#1}}\vspace{8mm}} 
\newcommand{\authors}[1]{\noindent{\large #1}\vspace{5mm}}
\newcommand{\address}[1]{{\itshape #1\vspace{2mm}}}


\begin{titlepage}
\begin{flushright}
 UWThPh-2019-28
\end{flushright}
\begin{center}
\title{ {\Large Higher-spin kinematics \& no ghosts on quantum space-time\\[1ex] 
in Yang-Mills matrix models}  }

\vskip 3mm

\authors{Harold C.\ Steinacker${}^\ddagger$}

\vskip 3mm

 \address{ 
${}^\ddagger${\it Faculty of Physics, University of Vienna\\
Boltzmanngasse 5, A-1090 Vienna, Austria  }  \\
Email: {\tt harold.steinacker@univie.ac.at}  
  }

\bigskip

\vskip 1.4cm

\textbf{Abstract}
\vskip 3mm

\begin{minipage}{14cm}%

A classification of bosonic on- and off-shell modes on a cosmological quantum space-time solution of the 
IIB matrix model is given, which leads to a higher-spin gauge theory. 
In particular, the no-ghost-theorem is established.
The physical on-shell modes consist of 2 towers of higher-spin modes, which are
effectively massless but include would-be massive degrees of freedom.
The off-shell modes consist of 4 towers of higher-spin modes, one of which was missing previously.
The noncommutativity leads to a cutoff in spin, which disappears in the semi-classical limit. 
An explicit  basis allows to obtain the full propagator, 
which is governed by a universal effective metric.
The physical metric fluctuations arise from would-be massive spin 2 modes,
which were previously shown to include the linearized Schwarzschild solution.
Due to the maximal supersymmetry of the IIB model, this 
is expected to define a consistent quantum theory in 3+1 dimensions, which includes gravity.

\end{minipage}

\end{center}

\end{titlepage}

\tableofcontents
%
%
\section{Introduction} 

The starting point of this paper is a 
recent solution of  IKKT-type matrix models with mass term \cite{Sperling:2019xar}, 
which is naturally interpreted as 3+1-dimensional cosmological FLRW quantum space-time.
It was  shown   that the fluctuation modes around this background 
include spin-2 metric fluctuations, as well as a truncated tower of higher-spin modes which are 
organized in a higher-spin gauge theory. The standard Ricci-flat massless 
graviton modes were found, as well as some additional 
vector-like and scalar metric modes. The latter was
shown to provide the linearized Schwarzschild solution in \cite{Steinacker:2019dii}.
However, the fluctuation analysis was not complete. In particular, although general arguments suggest 
that the model should be free of ghosts, this has not been established up to now.

The present paper provides a complete analysis and classification of all bosonic
fluctuation modes which arise on this background in the matrix model. It turns out that in addition to 
the three towers of (off-shell) higher spin modes found in \cite{Sperling:2019xar}, there is a fourth 
tower, which is obtained  in a coherent way. This provides a full and explicit diagonalization 
of the gauge-fixed quadratic action for the bosonic matrix fluctuations. 
Moreover, we classify and find the physical modes (i.e. the gauge-fixed 
on-shell modes modulo pure gauge modes) 
and show that the invariant inner product is positive, so that
they define a Hilbert space. 
Since the quadratic action is defined by the same inner product, this amounts to the statement that 
there are no ghosts, i.e. no physical modes with negative norm.  
We also compute the inner products for all off-shell modes, which is found to have the same 
Minkowski structure as in flat space. This allows in principle to write down the full propagator, 
and should be very useful in a future analysis of perturbative quantization. 

Along the way, many useful and surprisingly nice properties of the spacetime and its modes are uncovered, 
including simple on-shell relations which show that the time evolution  behaves 
very much like on commutative space, even in the presence of space-time noncommutativity.
Quite generally speaking, even though the organization is rather involved due to the higher-spin structure, 
the results are remarkably nice and simple.

The origin of higher-spin modes can be understood as follows.
The mathematical structure underlying the background under consideration is quantized twistor 
 space $\C P^{1,2}_n$, 
which is a quantized 6-dimensional coadjoint orbit of $SU(2,2)$ or $SO(4,2)$.
Semi-classically, this is an $S^2$ bundle over the 4-hyperboloid $H^4$, or over the space-time $\cM^{3,1}$.
The latter is a projection of $H^4$ with Minkowski signature, describing a 
FLRW cosmological space-time with a Big Bounce. This $S^2$ fiber is quantized and therefore 
admits only finitely many harmonics, which transmute into higher spin modes on $\cM^{3,1}$
due to the twisted bundle structure. All this is automatic on the 
matrix background under consideration.

For reasons of 
transparency and simplicity the analysis is performed in the semi-classical Poisson limit,
where spacetime is described by a classical manifold carrying extra structure which is underlying 
the noncommutativity. This case is already very interesting in its own right,
and since most computations are based on the Lie-algebraic 
structures, 
most steps would go through  in the noncommutative case with  minor modifications.
The classification of modes is literally the same due to the $SO(4,2)$-covariant 
quantization map $\cQ$ \eq{quantization-map}, 
and the no-ghost result is expected to hold also in the 
 non-commutative case up to the cutoff.

However, there is one complication. Due to the FLRW geometry, the isometry group $SO(3,1)$
of the background comprises  space-like translations and rotations, but no boosts.
This means that local Lorentz invariance is only partially manifest. 
The usual 3+1-dimensional tensor fields accordingly decompose into several $SO(3,1)$ sub-sectors. 
This sub-structure is addressed in section \ref{sec:higher-spin}
which  leads to an organization reminiscent of but distinct from primary and secondary fields in CFT.
In any case, the underlying $SO(4,2)$ 
structure group is powerful enough to  control the kinematics.
There is in fact one advantage, since the absence of ghost is quite transparent
as the fields are naturally organized in space-like or radiation gauge.
In the end, local Lorentz invariance seems to be effectively respected and
all modes propagate in the exact same way, governed by a universal effective metric.
This is expected due to the manifest higher spin gauge symmetry,
which includes an analog of (modified) volume-preserving diffeos.
Nevertheless, the issue of local Lorentz invariance should  be clarified further.

The appearance of a higher-spin gauge theory is of course very reminiscent of Vasiliev's higher spin 
theory \cite{Vasiliev:1990en,Didenko:2014dwa}. 
Indeed as elaborated in previous papers \cite{Sperling:2018xrm,Sperling:2017gmy}, the present higher-spin kinematics 
is clearly related to the higher spin algebras of Vasiliev theory, although further clarification would
be desirable. There may also be a close relation with the Yang-Mills higher spin models 
considered in \cite{Bonora:2018ggh}. 
However there are clearly significant differences. In particular, the present model is defined by an 
action and features two scales, and IR scale given by the cosmic curvature and a UV scale 
where the noncommutativity becomes significant. The separation of these scales 
is determined by an integer $n$, and is therefore protected from quantum corrections.

The results of this paper thus provides a solid base for 
an interacting higher spin gauge theory
which appears to include gravity.
Although the model is intrinsically noncommutative, it should be viewed 
in the spirit of field theory. In contrast to holographic approaches 
space-time arises as a condensation of matrices here, whose dynamical fluctuations are
described by an effective  (almost-local) field theory.
Most importantly, the present model is well suited for quantization, 
as discussed in the outlook. The present
results should allow to study the quantum theory in detail. In particular, it would be 
very interesting to make contact with the numerical simulations of the 
IKKT model \cite{Nishimura:2019qal,Aoki:2019tby,Kim:2011cr}, 
which provide evidence that an expanding 3+1-dimensional space-time  indeed arises 
at the non-perturbative level.

The paper is rather technical and includes all the required details. 
To make it more accessible, the conceptual considerations 
are kept in the main text while many technical details are delegated to the appendix.
The main results are the classification of modes in sections \ref{sec:modes} and \ref{sec:completeness}, and 
the no-ghost theorem  in section \ref{sec:physical}. The required background is provided in 
sections \ref{sec:defs} and \ref{sec:semi-class}, which should make the paper
mostly self-contained. 
Finally, a disclaimer on  mathematical rigour:
The use of ``Theorem'', ``Lemma'' etc. 
should be understood in a semi-rigorous physicist's sense.
The statements are clear-cut and  justified with formal proofs,
but  full  mathematical precision is not attempted. 

%
%

%
%
%

\section{Basic definitions and algebraic structures}
\label{sec:defs}

The theory under consideration \cite{Sperling:2019xar} is based on  the Lie algebra $\mso(4,2)$
generated by $M^{ab}$,  
\begin{align}
  [M_{ab},M_{cd}] &=\im \left(\eta_{ac}M_{bd} - \eta_{ad}M_{bc} - 
\eta_{bc}M_{ad} + \eta_{bd}M_{ac}\right) \ 
 \label{M-M-relations-noncompact}
\end{align}
for $a,b=0,..,5$,
and a specific class of
unitary representations  $\cH_n$ known as doubletons or 
minireps \cite{Mack:1975je,Fernando:2009fq}, labeled by  $n\in\N$. These are short discrete series unitary irreps of 
$\mso(4,2)$, which have the distinctive feature that they remain 
irreducible if restricted to $SO(4,1)\subset SO(4,2)$.
They are also multiplicity-free lowest weight representations. The special case $n=0$ is excluded.

\paragraph{Fuzzy hyperboloid $H^4_n$.}

The fuzzy hyperboloid $H^4_n$ \cite{Hasebe:2012mz,Sperling:2018xrm} is defined in terms of $SO(4,1)$ vector operators   
\begin{align}
 X^a = r M^{a5}, \qquad a=0,...,4 \ .
\end{align}
Here $r$ has dimension length, and  $\eta_{ab} = \diag(-1,1,1,1,1,-1)$. 
Since $\cH_n$ remains irreducible for $SO(4,1)$, they satisfy
the relations of a 4-dimensional hyperboloid
\begin{align} 
 \eta_{ab} X^a X^b &=  - R^2 \one \  ,\qquad
 R^2 = \frac{r^2}{4}(n^2-4)
  \label{X-constraint}
\end{align}
where the sum is over $a,b=0,...,4$. It is easy to see that the 
$X^a$ generate the full algebra $\End(\cH_n)$, which
 transforms under  $SO(4,2)$  via
\begin{align}
 M^{ab} \triangleright\phi = [M^{ab} ,\phi]  \ , \qquad \phi \in \End(\cH_n) \ .
\end{align}
The quadratic Casimirs of $SO(4,2)$ and $SO(4,2)$ act on $\phi \in \End(\cH_n)$ as
\begin{align}
 C^2[\mso(4,2)] \phi &= \frac 12 [M^{ab},[M_{ab},\phi]], \qquad a,b=0,...,5 \nn\\
 C^2[\mso(4,1)] \phi &= \frac 12 [M^{ab},[M_{ab},\phi]] , \qquad a,b=0,...,4 
 \label{Casimirs-adjoint}
\end{align}
and  the  $SO(4,1)$- invariant matrix Laplacian on $H^4_n$ 
 \begin{align}
 \Box_H  \phi &= [X_a,[X^a,\phi]] = (- C^2[\mso(4,2)] + C^2[\mso(4,1)]) \phi \ 
 \label{Laplacian-def}
\end{align}
 encodes the geometry of $H^4$. 
All indices will be raised or lowered with the appropriate 
  $\eta^{ab}$ throughout the paper, and latin labels 
$a,b$ range from $0$ to $4$ (or possibly 5).
In particular, the following $SO(4,1)$- invariant  Casimir on $\End(\cH_n)$ \cite{Sperling:2018xrm,Sperling:2019xar} 
\begin{align}
\cS^2 &:= \frac 12\sum_{a,b\neq 5} [M_{ab},[M^{ab},.]] 
  + r^{-2} [X_a,[X^a,.]]  \nn\\
  &= 2 C^2[\mso(4,1)] - C^2[\mso(4,2)]
\label{Spin-casimir}
 \end{align}
 can be
interpreted as a spin observable on $H^4_n$, which satisfies
 \begin{align}
 [\cS^2,\Box_H] = 0    \ .
\end{align}
Hence $\Box_H$ and $\cS^2$ can be simultaneously diagonalized, and $\End(\cH_n)$ decomposes into \cite{Sperling:2018xrm}
\begin{align}
 \End(\cH_n) =\cC = \cC^0 \oplus \cC^1 \oplus \ldots \oplus  \cC^n\qquad \text{with} 
 \qquad \cS^2|_{\cC^s} = 2s(s+1)  \ .
 \label{End-decomp-0}
\end{align}
We will see that $\cC^0$ describes the 
space of (scalar) functions on $H^4_n$,  while $\cC^s$ describes spin $s$ modes on $H^4_n$.
The origin of this higher spin structure can be understood by noting that 
$\End(\cH_n)$ should be interpreted as 
  {\em quantized algebra of functions on  $\C P^{1,2}$}, 
which is an equivariant\footnote{i.e. $SO(4,1)$ acts on the entire bundle in a way consistent with the bundle projection.} 
$S^2$-bundle over $H^4$. 
This is best understood in terms of 
coherent states, which are defined as follows: let
\begin{align}
 |x_0\rangle := |0\rangle \qquad \in \cH_n
\end{align}
be the lowest weight state. This is an optimally 
localized state\footnote{In a suitable sense, cf.  \cite{Perelomov:1986tf}, 
or \cite{Schneiderbauer:2016wub} for a discussion in a similar context.} 
at the  ''south pole`` 
of $H^4$, with $\langle x_0|X^a|x_0\rangle = x_0 = R(\frac{n}{2}+1,0,0,0,0)$. 
Then the  coherent state $|x\rangle = g\triangleright|x_0\rangle  \in \cH_n$ is defined by a rotation
$g\in SO(4,1)$ which rotates $x_0$ into $x\in H^4$. Since the stabilizer group of $x_0\in H^4$ is $SO(4)$, 
the expectation values 
\begin{align}
x^a = \langle x|X^a|x\rangle
\label{Hopf-coherent}
\end{align}
span  $H^4 \cong SO(4,1)/SO(4)$.
However there is a hidden fiber bundle over $H^4$, which arises from the fact that $\cH_n$ is  
a representation of $\msu(2,2)\cong \mso(4,2)\supset \mso(4,1)$. Then  the 
coherent states sweep out the space
\begin{align}
 \{|p\rangle = g\triangleright|0\rangle, \ g\in SU(2,2)\} \cong SU(2,2)/SU(2,1) = \C P^{1,2} \times U(1) \ .
\end{align}
Here $\C P^{1,2}$ is a 6-dimensional coadjoint orbit of $SU(2,2)$, 
which is a $S^2$ bundle over $H^4$ via the Hopf map \eq{Hopf-coherent}.
The fiber describes in fact a fuzzy $S^2_n$ spanned by the stabilizer $SU(2)_L$ of $x_0\in H^4$ acting on $|0\rangle$,
which spans an $n+1$-dimensional irrep, leading to the truncation in \eq{End-decomp-0}.
For more details we refer to \cite{Sperling:2018xrm}. The extra $U(1)$ is just the phase 
of the coherent states on $\C P^{1,2}$.

Using these coherent states, we can write down a natural $SO(4,2)$-equivariant
quantization map from the classical space of functions on $\C P^{1,2}$
to the noncommutative or fuzzy functions $End(\cH_n)$:
\begin{align}
 \cQ: \quad \cC(\C P^{1,2}) &\to End(\cH_n)  \nn\\
 \phi(p)  &\mapsto \hat \phi := \int\limits_{\C P^{1,2}} \phi(p) 
\left|p\right\rangle \left\langle p\right| \ .
 \label{quantization-map}
\end{align}
Here $\C P^{1,2}$ is equipped with the canonical $SO(4,2)$-invariant  measure.
This map is essentially one-to-one up to a cutoff \cite{Sperling:2018xrm}, 
mapping square-integrable functions to Hilbert-Schmidt operators.
The inverse map (up to normalization \& cutoff) is given by the symbol 
\begin{align}
 \hat \phi \in End(\cH_n) \mapsto 
\langle p| \hat \phi |p\rangle \ = \phi(p) \ \in \  \cC(\C P^{1,2}) \ .
 \label{quantization-inverse}
\end{align}
Hence $End(\cH_n)$ decomposes into the same unitary irreps as  $L^2(\C P^{1,2})$ below the cutoff,
and the harmonics on the $S^2_n$ fiber
lead to \eq{End-decomp-0}.
Since $\cQ$ respects $SO(4,2)$, the generators act as
\begin{align}
  [M^{ab},\cQ(\phi)] &= \cQ(i\{m^{ab},\phi(x)\})
  \label{Q-intertwiner}
\end{align}
where $\{m^{ab},.\}$ implements the $SO(4,2)$ action on  $\cC(\C P^{1,2})$ via the Poisson bracket arising 
from the canonical (Kirillov-Kostant-Souriau) symplectic structure. This Poisson bracket is defined through
the Lie algebra relations \eq{M-M-relations-noncompact}
for the embedding functions $m^{ab}:\ \C P^{1,2}\hookrightarrow\mso(4,2)\cong\R^{15}$, 
replacing $[.,.]$ by $i\{.,.\}$.
This replacement will be called {\em semi-classical limit} indicated by $\sim$.
In particular, it is easy to see that $M^{ab} = \cQ(m^{ab})$ and 
$X^a = \cQ(x^a)$ (up to normalization).

Due to the intertwiner property of $\cQ$, most of the (Lie-algebraic) computations   
carried out at the Poisson level carry over immediately to the full non-commutative (NC) case in $End(\cH_n)$. 
For example, the Casimirs and Laplacian are respected:
\begin{align}
  [M^{ab},[M_{ab},\cQ(\phi)]]  &= \cQ(-\{M^{ab},\{M_{ab},\phi\}\}),  \nn\\
  \Box_H \cQ(\phi) &=  \cQ(\Box_H\phi) \ , 
 \label{Laplacian-intertwiner}
\end{align}
where  $\Box_H\phi = -\{x^a,\{x_a,\phi\}\}$ on the rhs
is the Laplacian on $H^{4}$.
Thus even though we will mostly work in the 
semi-classical case, most of the results  carry over immediately to the NC case.

\paragraph{Fuzzy space-time $\cM^{3,1}_n$.}

The main space of interest here is  
the fuzzy or quantum space-time $\cM^{3,1}_n$, which is  generated by the
$X^\mu, \ \mu = 0,...,3$,  dropping the $X^4$ generator of  $H^4_n$. 
Then
\begin{align}
 \eta_{\mu\nu} X^\mu X^\nu &=  - R^2 \one - X_4^2 \ ,
\end{align}
and  greek labels $\mu,\nu$ etc. will run from $0$ to $3$ throughout the paper.
Dropping the $X^4$ generator corresponds to a projection  of $H^4$ to $\R^{3,1}$,
so that $\cM^{3,1}_n$ should be interpreted as 2-sheeted hyperboloid, 
as sketched in figure \ref{fig:projection}. 
\begin{figure}[h]
\hspace{2cm} \includegraphics[width=0.5\textwidth]{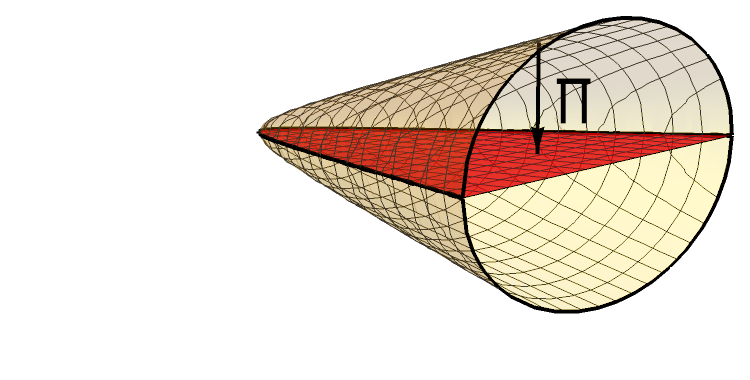}
 \caption{Projection $\Pi$ from $H^4$ to $\cM^{3,1}$  with 
Minkowski signature.}
 \label{fig:projection}
\end{figure}
This interpretation is substantiated via the  matrix d'Alembertian
\begin{align}
 \Box  \phi &=  [T_\mu,[T^\mu,\phi]] =  (C^2[\mso(4,1)] - C^2[\mso(3,1)])\phi \ ,
 \label{dAlembert-fuzzy-def}
\end{align}
which encodes an  $SO(3,1)$-invariant d'Alembertian for $\cM^{3,1}$
 with Lorentzian structure\footnote{It is natural to wonder about the Sitter solutions.
 While this is possible in principle \cite{Gazeau:2009mi,Buric:2017yes}, 
 $End(\cH)$ would imply a non-compact internal fiber and infinitely many dof per unit volume. This is avoided here.}, where 
 \begin{align}
 T^\mu = \frac{1}{R} M^{\mu 4} \ .
\end{align}
It is easy to see that the $X^\mu$ alone generate the full algebra $End(\cH_n)$,
which can now be interpreted as quantized functions on a $S^2$-bundle over $\cM^{3,1}$. 
They satisfy the commutation relations
\begin{align}
 [X^\mu,X^\nu] =: i\Theta^{\mu\nu} = -i r^2 M^{\mu\nu} \ .
 \label{XX-CR}
\end{align}
It turns out that $\Theta^{\mu\nu}$ is related to $T^\mu$ (cf. \eq{theta-P-relation}), which
satisfy the commutation relations 
\begin{align}
 [T^\mu,T^\nu] = -\frac{i}{r^2R^2}\Theta^{\mu\nu}, \qquad
 [T^\mu,X^\nu] = \frac{i}{R} \eta^{\mu\nu}\, X_4 \ .
  \label{XT-CR}
\end{align}
These generators satisfy further constraints due to the special representation $\cH_n$.
To simplify these relations we will focus on the semi-classical (Poisson) 
limit $n \to \infty$ from now on, working with commutative functions of $x^\mu\sim X^\mu$ and $t^\mu \sim T^\mu$, 
but keeping the Poisson or symplectic structure  $[.,.] \sim i \{.,.\}$ encoded in $\theta^{\mu\nu}$.

In order to have a well-defined action,
we will  consider modes on $\cM^{3,1}$ which are square-integrable, in the sense that 
the $SO(4,2)$-invariant inner product is finite,
\begin{align}
0 < \ \langle\phi,\phi'\rangle := \Tr {\phi}^\dagger\phi'
 \ \sim \int_{\C P^{1,2}} \, {\phi}^* \phi'   \ < \infty \ 
\label{inner-basic-def}
\end{align}
where functions $\phi \in L^2(\C P^{2,1})$ are identified with operators $End(\cH_n)$
via \eq{quantization-map}. 
The measure is the symplectic volume form
$\Omega = \frac{(2\pi)^3}{3!}\omega^{\wedge 3}$ on  $\C P^{1,2}$, which is dropped.
All integrals in the paper are understood in this sense, unless stated otherwise.
Accordingly, $\phi\in L^2(\C P^{1,2})$ belongs to some unitary representation of $SO(4,2)$.

Since $SO(4,2)$ is the conformal group on $\R^{3,1}$, one might hope to 
apply  CFT concepts  such as conformal primaries etc.
Indeed $\cH_n$ is a lowest-weight module with ground state $|0\rangle$ which
is an eigenstate of $D = X^4$, whose eigenvalues are raised and lowered with 
 $M^{\mu 5} \pm i M^{\mu 4}$.
However, the main object of interest  is  $End(\cH_n) \cong \cH_n \otimes \cH_n^*$, 
and the square-integrable modes
consists of principal series modules rather than highest or lowest weight modules. 
Therefore the familiar concepts from CFT are not useful here.
Instead   we will develop some more suitable 
structures  in section \ref{sec:higher-spin} which replace these concepts  to some extent.

%
\section{Semi-classical structure of \texorpdfstring{$\cM^{3,1}$}{M(3,1)}}
\label{sec:semi-class}

In the semi-classical limit, the generators $x^\mu$ and $t^\mu$ satisfy 
the following constraints   \cite{Sperling:2018xrm}
\begin{subequations}
\label{geometry-H-M}
\begin{align}
 x_\mu x^\mu &= -R^2 - x_4^2 = -R^2 \cosh^2(\eta) \, , 
 \qquad R \sim \frac{r}{2}n   \label{radial-constraint}\\
 t_{\mu} t^{\mu}  &=  r^{-2}\, \cosh^2(\eta) \, \\
 t_\mu x^\mu &= 0 \ \label{x-t-orth}
\end{align}
\end{subequations}
which arise from the special properties of $\cH_n$. 
We will interpret $x^\mu:\, \cM^{3,1}\hookrightarrow \R^{3,1}$ as Cartesian  
coordinate functions.
Here $\eta$ is a global time parameter defined via
\begin{align}
 x^4 = R \sinh(\eta) \ 
 \label{x4-eta-def}
\end{align}
which defines a foliation of $\cM^{3,1}$ into space-like surfaces $H^3$; this 
will be related to the scale parameter of a FLRW cosmology \eqref{a-eta} with $k=-1$.
Note that $\eta$ distinguishes the two degenerate sheets of $\cM^{3,1}$, 
cf. figure \ref{fig:projection}. 
The $t^\mu$ generators clearly describe the $S^2$ fiber over $\cM^{3,1}$, which is 
 space-like due to \eq{x-t-orth}. 
These generators satisfy the Poisson brackets
\begin{align}
 \{x^\mu,x^\nu\} &= \theta^{\mu\nu}  = - r^2 R^2 \{t^\mu,t^\nu\}  ,  \nn\\
 \{t^\mu,x^\nu\} &= \frac{x^4}{R} \eta^{\mu\nu} \label{Poisson-brackets} \ .
\end{align}
The Poisson tensor $\theta^{\mu\nu}$ can be expressed in terms of $t^\mu$ via  \cite{Sperling:2018xrm}
\begin{align}
 \theta^{\mu\nu} &= \frac{r^2}{\cosh^2(\eta)} 
   \Big(\sinh(\eta) (x^\mu t^\nu - x^\nu t^\mu) +  \epsilon^{\mu\nu\a\b} x_\a t_\b \Big) \ ,
 \label{theta-P-relation} \, 
\end{align}
and it satisfies the constraints
\begin{subequations}
\label{geometry-H-theta}
\begin{align}
 t_\mu \theta^{\mu\a} &= - \sinh(\eta) x^\a , \\
 x_\mu \theta^{\mu\a} &= - r^2 R^2 \sinh(\eta) t^\a , \label{x-theta-contract}\\
 \eta_{\mu\nu}\theta^{\mu\a} \theta^{\nu\b} &= R^2 r^2 \eta^{\a\b} - R^2 r^4 
t^\a t^\b + r^2 x^\a x^\b  
%
\end{align}
\end{subequations}
as well as self-duality relations given in Lemma \ref{lem:selfdual-2}.

We observe that due to the relation \eqref{Poisson-brackets},  the derivations or Hamiltonian vector fields
\begin{align}
 -i[T^\mu,.] \sim \{t^\mu,.\}  \ 
\end{align}
play the role of momentum generators on $\cM^{3,1}$, which satisfy
\begin{align}
 \{t_\mu,\phi\} =   \sinh(\eta)\del_\mu \phi 
 \label{del-t-rel}
\end{align}
for $\phi = \phi(x)$. There is  also an
$SO(3,1)$-invariant 
global time-like vector field 
\begin{align}
 \t := x^\mu \del_\mu  .
 \label{tau-def}
\end{align}

\subsection{Effective metric and d'Alembertian}
\label{sec:metric}

In the matrix model framework, the effective metric on any given background 
is obtained by rewriting the
kinetic term in covariant form \cite{Sperling:2019xar,Steinacker:2010rh}.
For the $\cM^{3,1}$ background  under consideration, this is
\begin{align}
S[\phi] =  - \Tr [T^\mu,\phi][T_\mu,\phi] 
\sim \int d^4 x\,\sqrt{|G|}G^{\mu\nu}\del_\mu\phi \del_\nu \phi \ 
\label{scalar-action-metric}
\end{align}
and one obtains
\cite{Sperling:2019xar}
\begin{align}
   G^{\mu\nu} &= \sinh^{-3}(\eta)\, \g^{\mu\nu},
   \qquad  \g^{\a\b} = \eta_{\mu\nu}\theta^{\mu\a}\theta^{\nu\b} 
  = \sinh^2(\eta) \eta^{\a\b} \ 
   \label{eff-metric-G}
\end{align}
dropping some irrelevant constant.
This metric can be recognized as 
$SO(3,1)$-invariant FLRW metric with signature $(-+++)$, 
\begin{align}
 \dd s^2_G = G_{\mu\nu} \dd x^\mu \dd x^\nu 
   &= -R^2 \sinh^3(\eta) \dd \eta^2 + R^2\sinh(\eta) \cosh^2(\eta)\, \dd \Sigma^2 \ \nn\\
   &= -\dd t^2 + a^2(t)\dd\Sigma^2 \, .
   \label{eff-metric-FRW}
\end{align}
We can read off the cosmic scale parameter $a(t)$  
\begin{align}
a(t)^2 &=  R^2\sinh(\eta) \cosh^2(\eta) \ \stackrel {t\to\infty}{\sim}  \  R^2\sinh^3(\eta) ,  \label{a-eta}\\
\dd t &=  R \sinh(\eta)^{\frac{3}{2}} \dd\eta \ 
\end{align}
which leads to $a(t) \sim \frac 32 t$ for late times.
This metric can also be extracted from the 
''matrix`` d'Alembertian \eq{dAlembert-fuzzy-def} 
\begin{align}
 \Box := [T^\mu,[T_\mu,.]] \  \sim \ -\{t^\mu,\{t_\mu,.\}\}  \ 
   = \sinh^{3}(\eta) \Box_G
 \label{Box-def}
\end{align}
acting on  $\phi \in \cC^0$, 
where\footnote{It is interesting to observe that the invariant volume form 
$d^4 x\, \frac 1{x^4}$ arising from the symplectic volume form \cite{Sperling:2019xar}
does {\em not} coincide with the 
Riemannian volume  $d^4x\, \sqrt{|G|}$. Accordingly, spin 1 gauge transformations
are diffeomorphisms which preserve $\Omega$ rather than the Riemannian volume.}
$\Box_G = -\frac{1}{\sqrt{|G|}}\del_\mu\big(\sqrt{|G|}\, G^{\mu\nu}\del_\nu\big)$.

\subsection{Higher spin sectors  on $\cM^{3,1}$  and $H^3$ substructure}
\label{sec:higher-spin}

Due to the extra generators $t^\mu$, we obtain explicitly the decomposition \eq{End-decomp-0} of the 
full algebra of functions 
into sectors $\cC^s$ which correspond to spin $s$ harmonics on the $S^2$ fiber:
\begin{align}
    \End(\cH_n) =\cC = \cC^0 \oplus \cC^1 \oplus \ldots \oplus  \cC^n\qquad \text{with} 
\quad 
  \cS^2|_{\cC^s} = 2s(s+1)  \, .
  \label{EndH-Cs-decomposition}
\end{align} 
In the semi-classical limit, the  $\cC^s$  are modules\footnote{The module structure also applies 
in the noncomutative case if $\cC^0$ is equipped with the commutative but non-associative pull-back 
algebra structure, due to (3.46) in \cite{Sperling:2018xrm}. Useful discussions with S. Rangoolam 
are acknowledged.}   
over $\cC^0$, which should be viewed as 
sections of (higher spin) bundles over $H^4$. 
More specifically,
$\cC^s$  can be viewed  as totally symmetric traceless space-like 
rank $s$ tensor fields on $\cM^{3,1}$
\begin{align}
 \phi^{(s)} = \phi_{\mu_1 ... \mu_s}(x) t^{\mu_1} ... t^{\mu_s} , 
 \qquad \phi_{\mu_1 ... \mu_s} x^{\mu_i} = 0
 \label{Cs-explicit}
\end{align}
due to  \eqref{geometry-H-M}.
The underlying $\mso(4,2)$ structure provides an $SO(3,1)$ -invariant derivation
\begin{align}
 D\phi &:= \{x^4,\phi\} \ 
  = r^2 R^2 \frac{1}{x^4} t^\mu \{t_\mu,\phi\}   
  = -\frac{1}{x^4}x_\mu\{x^\mu,\phi\} \nn\\
   &= r^2 R\, t^{\mu_1}\ldots t^{\mu_s} t^\mu \, \nabla^{(3)}_\mu\phi_{\mu_1\ldots\mu_s}(x) 
 \label{D-properties}
\end{align}
where  $\nabla^{(3)}$ is the covariant derivative along the 
space-like $H^3 \subset \cM^{3,1}$.
Hence $D$ relates the different spin sectors in \eqref{EndH-Cs-decomposition}:
\begin{align}
  D = D^- + D^+: \ \cC^{s} \ &\to \cC^{s-1} \oplus \cC^{s+1}, \qquad
    D^\pm \phi^{(s)} = [D\phi^{(s)}]_{s\pm 1} \ 
    \label{D-properties-2} 
\end{align}
where $[.]_{s}$ denotes the projection to $\cC^{s}$ defined through 
\eqref{EndH-Cs-decomposition}.
It is easy to see that 
\begin{align}
(D^+)^\dagger = - D^-
\end{align}
w.r.t. the inner product \eq{inner-basic-def}.
Explicitly, $D x^\mu = r^2 R\, t^\mu$ and $D t^\mu =  R^{-1}\, x^\mu$. 
In particular,  $\cC^{(s,0)} \subset \cC^{s}$ is the space of divergence-free traceless
space-like rank $s$ tensor fields on $\cM^{3,1}$, in radiation gauge.

The $D^\pm$ operators allow to organize the $\cC^{s}$ modes
into primals and descendants
\begin{align}
 \boxed{ \
\begin{aligned}
 \cC^{(s,0)} &= \{\phi\in\cC^s; \ D^- \phi = 0 \} &  \mbox{... primal fields}\ \nn\\
 \cC^{(s+k,k)} &= (D^+)^k\cC^{(s,0)} &  \mbox{... descendants} \quad 
\end{aligned}
}
\end{align}
cf. \cite{Sperling:2018xrm}. 
This is somewhat reminiscent of primaries in CFT but the concepts are different.
The primals\footnote{In contrast to primaries in CFT, these are not annihilated by the 
$M^{\mu 5} - i M^{\mu 4}$ operators which lowers the eigenvalue of $D$.
Primal fields do not have an eigenvalue of $D$.} have minimal spin $\cS^2$, 
which is raised and lowered by $D^\pm$; they
correspond to divergence-free spin $s$ tensor fields on $H^4$ in space-like gauge, 
i.e. tangential to $H^3$.
The descendants are {\em space-like} derivatives of the primal fields.
However they should not be considered as 
pure gauge fields, and they are part of the physical Hilbert space. 

This sub-structure encodes two different concepts on the FRW background, which 
arise from the presence of a space-like foliation:
$\cS^2=2s(s+1)$ measures the 4-dimensional spin on $H^4$,
while $(s-k)$ measures the 3-dimensional spin of $\cC^{(s,k)}$ on $H^3$. 
Nevertheless, local Lorentz invariance should be largely restored through  gauge invariance, 
which contains  $\Omega$- volume-preserving diffeos. 
Although these act  in a somewhat unusual manner \cite{Sperling:2018xrm},  
one may expect that they protect the model from 
pathological Lorentz violation. This will be illustrated by the fact that all modes 
propagate according to the same effective d'Alembertian $\Box$.
In  physical terms,
an $SO(4,1)$ irrep $\phi^{(s)} \in \cC^{s}$ encodes a series of massless modes $\phi^{(s,k)}$
in radiation gauge with spin $s-k$ for $k=0,...,s$.

\paragraph{Averaging over $S^2$.}

We can interpret the projection  $[f(t)]_0$  on the scalar sector $\cC^0$ as an averaging
or integral over the $S^2$ fiber described by the $t$ generators,
\begin{align}
 [f(t)]_0 &= \frac{1}{4\pi r^{-2}\cosh(\eta)^2}\int_{S^2_t}  f(t) 
 \label{average-int-def}
\end{align}
such that $[1]_0 = 1$.
This gives the formula
\begin{align}
 [t^\mu t^\nu]_0 &= \frac{\cosh^2(\eta)}{3r^2} P_\perp^{\mu\nu}
   \label{kappa-average}
\end{align}
where
\begin{align}
P_\perp^{\mu\nu} &:= \eta^{\mu\nu} + \frac{1}{R^2\cosh^2(\eta)} x^\mu x^\nu
 \label{H3-projector}
\end{align}
is the  positive semi-definite projector tangential to the space-like $H^3$. 
Furthermore, we have \cite{Sperling:2019xar}
\begin{subequations}
\label{averaging-relns}
\begin{align}
\left[t^{\a}  \theta^{\mu\nu}\right]_{0} 
  &= \frac{1}{3} \Big(\sinh(\eta) ( \eta^{\a\nu} x^\mu - \eta^{\a\mu} x^\nu)  +  x_\b \varepsilon^{\b 4\a\mu \nu}  \Big)\,, \\
 [t^{\mu_1} \ldots t^{\mu_4}]_0 &= 
  \frac 35 \big([t^{\mu_1}t^{\mu_2}][t^{\mu_3} t^{\mu_4}]_0 
   + [t^{\mu_1}t^{\mu_3}][t^{\mu_2} t^{\mu_4}]_0  + 
[t^{\mu_1}t^{\mu_4}][t^{\mu_2} t^{\mu_3}]_0\big) \,. 
\end{align}
\end{subequations}
This also provides a formula for the projection on $\cC^1$,
\begin{align}
[t^{\a} t^\b t^\g ]_{1} 
  &=  \frac 35 \Big([t^{\a} t^\b]_{0} t^\g + t^{\a} [t^\b t^\g]_{0} 
   +  t^\b[t^{\a} t^\g ]_{0}  \Big) \ .
   \label{average-3}
\end{align}
The general Wick theorem
\begin{align}
 [t^{\a_1} ...  t^{\a_{2s}}]_{0} 
  &= a_{2s}\sum [t^{\a_i}t^{\a_j}]  ... [t^{\a_k}t^{\a_l}]
\end{align}
summing over all
contractions can be obtained recursively from 
Lemma \ref{lemma-wick-1} in the appendix.

\subsection{$\cC^s$ and higher spin on  $H^4$}
\label{sec:Cs-higherspin}

In the previous section, $\cC^s$  was identified with space-like
spin $s$ tensor fields on $\cM^{3,1}$.
On the other hand, $\cC^s$ can also be identified 
with totally symmetric, traceless, divergence-free
tangential rank $s$ 
tensor fields $\phi_{a_1\ldots a_s}$ on $H^4$  via \cite{Sperling:2018xrm}
\begin{align}
  \phi^{(s)} = \{x^{a_s},\ldots\{x^{a_{1}},\phi_{a_1\ldots a_{s}}\}\ldots \} \ 
\in \cC^{s} \ .
\label{phis-rep-deriv-tensorfields}
\end{align}
Conversely, a totally symmetric tensor field on $H^4$ can be extracted from $\phi^{(s)}$ via
\begin{align}
 \tilde\phi_{a_1a_2\ldots a_{s}} := \{x_{a_1},...\{x_{a_s},\phi^{(s)}\}_- ...\}_-
  = \cA_{a_1}^{(-)}[ ... [ \cA_{a_s}^{(-)}[\phi^{(s)}]...] \quad \in \cC^{0}
  \label{tensorfields-Am-formula}
\end{align}
anticipating the notation \eq{A2-mode-ansatz},
which is tangential due to $ x_a\{x^a,\phi\} = 0$.
One can also define  intermediate tensor fields such as
\begin{align}
 \phi_{a_s}^{(s)} = \{x^{a_{s-1}},\ldots ,\{x^{a_{1}},\phi_{a_1\ldots a_{s-1} a_s}\}\ldots \}
     \ \in \cC^{s-1}
 \label{phi-a-s}
\end{align}
which are  tangential and
associated to the underlying irreducible rank $s$ tensor field. 
 Using Lemma \ref{lemma-alpha}  we obtain
 \begin{align}
 -\{x^{a}, \tilde\phi_{a}\} 
  &= \a_1(\Box_H - 4r^2)\phi^{(1)} 
 \end{align}
  and similarly using \eq{BoxH-A-relation}
 \begin{align}
  -\{x^{a_1}, \tilde\phi_{a_1a_2}\}
  &= \a_1(\Box_H - 4r^2) \cA_{a_2}^{(-)}[\phi^{(2)}] \nn\\
  &= \a_1 \cA_{a_2}^{(-)}[(\Box_H - 2(2+2)r^2)\phi^{(2)}] 
 \end{align}
and in general
\begin{align}
 -\{x^{a_1}, \tilde\phi_{a_1a_2\ldots a_{s}}\}
  &= \a_1\cA_{a_2}^{(-)}[ ... [ \cA_{a_s}^{(-)}[(\Box_H - r^2(s^2+s+2))\phi^{(s)}]...] \ . 
  \end{align}
Iterating this, we recover \eq{phis-rep-deriv-tensorfields} up to some action of $\Box_H$,
\begin{align}
 (-1)^s \{x^{a_1}, ...\{x^{a_s} \tilde\phi_{a_1a_2\ldots a_{s}}\}\ldots \}
  &= \cO(\Box_H) \phi
  \label{O(BoxH)-id}
\end{align}
where $\cO(\Box_H)$ is a positive and hence invertible operator provided
\begin{align}
   \Box_H > r^2 (s^2+s+2) \qquad \mbox{on} \ \cC^{s} \ . 
 \label{admissible}
\end{align}
We will see that 
this is indeed the case for admissible modes, because \eq{BoxH-estimate-onshell-principal} gives 
\begin{align}
  r^{-2} \Box_H > s^2+s+ 9/4 > s^2+s+2 \ .
\end{align}
Therefore the maps \eq{phis-rep-deriv-tensorfields} and \eq{tensorfields-Am-formula} are 
inverse of each other up to normalization.

%
%

\paragraph{Relation with higher spin field strength.}

It is instructive to work out these formulae more explicitly using 
the tangential derivatives on $H^4$ \cite{Sperling:2018xrm}
\begin{align}
  \eth^a \phi \coloneqq \frac{1}{r^2 R^2}x_b \{\theta^{ab} ,\phi\}, \qquad 
\phi\in\cC \; ,
\label{eth-def}
\end{align}
which satisfy
\begin{align}
 \{x^a,\cdot\} &= \theta^{ab}\eth_b \,    \qquad 
 x^a\eth_a = 0,  \nn\\
  \eth^a x^b &= P^{ab} = \eta^{ab} + \frac{1}{R^2} x^a x^b, \nn\\
 \eth^a \theta^{cd} &= \frac{1}{R^2}(-\theta^{ac}x^d + \theta^{ad}x^c) . 
 \label{eth-id}
\end{align}
It is then straightforward to show (cf. \cite{Sperling:2018xrm})
\begin{align}
 \phi_{a_2\ldots a_{s}} &= \{x^{a_1},\phi_{a_1a_2\ldots a_s}\} 
  = \theta^{a_1b_1}\eth_{b_1}\phi_{a_1a_2\ldots a_{s}} \qquad \in \cC^{1} \nn\\
  &\ \vdots \nn\\
 \phi^{(s)}
  &= \theta^{a_1b_1}...\theta^{a_{s}b_{s}}\eth_{b_{s}}...\eth_{b_1}\phi_{a_1a_2\ldots a_{s}}  
  =: \theta^{a_1b_1}...\theta^{a_{s}b_{s}} \cF_{a_1...a_d;b_1...b_s}
 \end{align}
 noting that $\theta^{a_1b_1}\theta^{a_2b_2} = r^2 R^2 P^{b_1b_2}$. 
Here
\begin{align}
 \cF_{b;a} &= \eth_a\phi_b - \eth_b\phi_a   \nn\\   
 & \ \vdots \nn\\
  \cF_{a_1...a_d;b_1...b_s} &= \eth_{[b_{s}}...\eth_{b_1}\phi_{a_1a_2\ldots a_{s}]} \ .
\end{align}
The last term is a generalization of the curvature or field strength
tensor, which has the symmetry of the Young tableau
${\tiny \Yvcentermath1 \young(aaa,bbb) }$.  This provides a link with Vasiliev's higher
spin theory \cite{Vasiliev:1990en,Didenko:2014dwa}; 
see also \cite{Sperling:2018xrm} for further related discussion.
However, the realization \eq{Cs-explicit} is more transparent.

\subsection{Admissible tensor fields and positivity}
\label{sec:admissible}

This section discusses integrability and positivity aspects, and can be skipped at first reading.

In order to have well-defined kinetic energy and similar quantities, we need 
some refinements of the integrability condition \eq{inner-basic-def}.
Consider for example
\begin{align}
0 \leq \int \{x_{a},\phi\}\{x^{a},\phi\}
  = \int \phi\Box_H\phi \ .
\end{align}
The lhs is positive
since  $\{x^{a},\phi\}$ is tangential to $H^4$, due to $x_a \{x^a,.\} = 0$. Therefore
\begin{align}
 \Box_H > 0 
  \label{BoxH-estimate}
\end{align}
must be positive definite. This argument 
carries over to $End(\cH_n)$ (for Hilbert-Schmidt-operators) using $\cQ$ \eq{quantization-map}.
However, we will need a slightly stronger bound, which can be obtained from group theory.
A heuristic argument for such an improved bound is as follows:
consider the $SO(4,1)$ invariant expression
\begin{align}
 -\int \{M_{ab},\phi\}\{M^{ab},\phi\}
  =  \int \phi  \{M^{ab}, \{M_{ab},\phi\}\}
  = -2\int \phi C^2[\mso(4,1)] \phi
\end{align}
for $a,b=0,...,4$. At the reference point $\xi = (R,0,0,0,0) \in H^4$,
the sum on the lhs separates as
\begin{align}
 - \{M_{ab},\phi\}\{M^{ab},\phi\}
  \stackrel{\xi}{=}  
  2\sum_a\{M^{a0},\phi\} \{M^{a0},\phi\}\ - \sum_{a,b=1}^4\{M^{ab},\phi\} \{M^{ab},\phi\} \ .
\end{align}
The first term is manifestly positive, while the second term  is negative and
involves the 
local stabilizer $SO(4)$ acting on $\phi$. 
Hence the second term measures the spin, and we expect heuristically 
that it contributes $-2s(s+1)$, if we forget about curvature corrections for the moment.
This would give the estimate $-C^2[\mso(4,1)] \geq -s(s+1)$ for integrable modes.

The precise statements required are obtained from representation theory for 
{\em principal series} of unitary representations. They
describe the normalizable fluctuation modes in the present context, corresponding to a 
continuous basis for square-integrable wavefunctions on the hyperboloids, 
analogous to plane waves in the flat case. 
The (bosonic) principal series  of unitary representations $\Pi_{\nu,s}$ of $SO(4,1)$ 
are determined by the spin $s\in\N_0$ and the real (''kinetic``) parameter $\nu\in\R$.
They can be identified with spin $s$ wavefunctions on $H^4$. 
For these representations, the quadratic Casimir
satisfies the following bound  \cite{dixmier1961representations}
\begin{align}
-C^2[\mso(4,1)] = 9/4 + \nu^2 - s(s + 1) > 9/4  - s(s + 1) 
\label{admissible-so41}
\end{align}
assuming\footnote{Note that $\nu$ will {\em not} play the role of a mass in the present context.
The case $\nu=0$ would correspond to some extreme IR case and is ignored here.} $\nu\neq 0$.
This is clearly a refined version of the above heuristic argument, and it
entails via \eq{Spin-casimir} the following bound for $\Box_H$
\begin{align}
 r^{-2} \Box_H \phi^{(s)}  &=  2s(s+1) - C^2[\mso(4,1)] > s^2+s+9/4  
   \label{BoxH-estimate-onshell-principal}
\end{align}
which is slightly stronger than \eq{BoxH-estimate}.
This will imply that the higher spin modes in the present framework
are square-integrable over $H^4$ and form a Hilbert space, 
as discussed below.
We will denote modes which satisfy the condition
\eq{admissible-so41}, i.e. which consist of unitary principal series   
of $SO(4,1)$,
as {\bf admissible modes\footnote{It is interesting to observe using \eq{Spin-casimir}
that the admissible modes are 
precisely those with $C^2[\mso(4,2)] > 9/2$, and it
is  plausible that those are precisely
the principal series irreps of $SO(4,2)$ in $End(\cH_n)$. However, this will not be investigated here. 
There are of course functions (e.g. polynomial functions) which violate these bounds,
but they are not normalizable and not considered here.}}. 
This condition is preserved by $D^\pm$ due to \eq{admissible-D}.

It is interesting to compare this with the (bosonic) principal series unitary representations  
$\Pi_{p,s}$  of $SO(3,1)$, which are determined by the spin $s\in\N_0$  
and a kinetic parameter $p\in\R$. 
They can be identified with spin $s$ wavefunctions on $H^3$,
and satisfy the  bound \cite{naimark2014linear,browne1976irreducible} 
\begin{align}
 -C^2[\mso(3,1)] = p^2 -s^2 + 1 > -s^2+1 \ .
 \label{admissible-so31}
\end{align}
Even though the conditions \eq{admissible-so41} and \eq{admissible-so31} are a priori independent, 
they are closely related for on-shell modes here, i.e. modes satisfying
the on-shell condition \eq{on-shell-all}
\begin{align}
 0 = \Box = C^2[\mso(4,1)] - C^2[\mso(3,1)] \ .
\end{align}
Then the 3-dimensional condition \eq{admissible-so31}
is  slightly stronger than the 4-dimensional condition \eq{admissible-so41} except for $s=1$.
This means that on-shell wave-functions which are square-integrable over some time-slice $H^3$
are automatically integrable over the entire space-time, which is quite remarkable and helpful
 for a theory with time evolution.

In particular, it follows that for admissible modes $\phi \in\cC^s$, 
the  tensor field  $\tilde\phi_{a_1a_2\ldots a_{s}}$ defined in \eq{tensorfields-Am-formula}
are square-integrable with positive-definite inner product
and form a Hilbert space, since
\begin{align}
 \int \tilde\phi^{a_1a_2\ldots a_{s}} \tilde\phi_{a_1a_2\ldots a_{s}}   
 &= \int (-1)^s \phi \{x^{a_1}, ...\{x^{a_s} \tilde\phi_{a_1a_2\ldots a_{s}}\}...\} \
  = \int \phi \cO(\Box_H) \phi
  \label{inner-tensorfield-BoxH}
\end{align}
and $\cO(\Box_H)$ \eq{O(BoxH)-id}  is positive as shown above.
For the Minkowski case see Corollary \ref{cor-pos}.

%

%

%
\section{Matrix model and higher-spin gauge theory}
\label{sec:fluctuations}

Now we return to the  noncommutative setting, and 
define a dynamical model for the fuzzy $\cM^{3,1}$ space-time under consideration.
Consider a Yang-Mills matrix model  with  mass term,
\begin{align}
 S[Y] &= \frac 1{g^2}\Tr \Big([Y^\mu,Y^\nu][Y_{\mu},Y_{\nu}] \, 
  +\frac{6}{R^2} Y^\mu Y_\mu  \Big) \ . 
 \label{bosonic-action}
\end{align}
All indices will be raised and lowered with $\eta^{\mu\nu}$ in the following sections.
This includes in particular the  IKKT or IIB matrix model \cite{Ishibashi:1996xs} with mass term,
which is best suited for quantization because maximal supersymmetry protects from UV/IR mixing \cite{Minwalla:1999px}.
As observed in \cite{Sperling:2019xar}, $\cM^{3,1}$ is indeed a solution 
of this model\footnote{This  ''momentum'' embedding via $T^\mu$ has some similarity with 
the ideas in \cite{Hanada:2005vr} but avoids excessive dof and the associated ghost issues, 
cf. \cite{Sakai:2019cmj}.
The positive mass parameter in \eqref{bosonic-action} simply sets the scale of the background.
For negative mass parameter, $X^\mu$ would be a solution \cite{Steinacker:2017bhb}, 
but the fluctuation analysis would be less clear.}, through
\begin{align}
 Y^\mu = T^\mu \ .
\end{align}
Now consider  tangential
deformations of the above background solution, i.e.
\begin{align}
 Y^\mu = T^\mu  + \cA^\mu \ , 
\end{align}
where $\cA^\mu \in \End(\cH_n) \otimes \R^4$ is an arbitrary Hermitian fluctuation.
The  Yang-Mills action \eqref{bosonic-action}  can be expanded around the solution as
\begin{align}
 S[Y] = S[T]  +  S_2[\cA] + O(\cA^3) \ ,  
 \end{align}
 and the quadratic fluctuations are  governed by  
 \begin{align}
S_2[\cA] = -\frac{2}{g^2} \,\Tr \left( \cA_\mu 
\Big(\cD^2 -\frac{3}{R^2}\Big) \cA^\mu + \cG\left(\cA\right)^2 \right) .
\label{eff-S-expand}
\end{align}
This involves the vector d'Alembertian on $\cM^{3,1}$
\begin{align}
\cD^2 \cA =  \left(\Box  - 2\cI \right)\cA  
\label{vector-Laplacian}
\end{align}
(cf. \eq{Box-def})
which is an $SO(3,1)$ intertwiner, as well as 
\begin{align}
 \cI (\cA)^\mu := - [[ Y^\mu, Y^\nu],\cA_\nu] =  \frac{\im}{r^2 R^2} 
[\Theta^{\mu\nu},\cA_\nu] 
 \eqqcolon -\frac{1}{r^2 R^2}\tilde\cI (\cA)^\mu \ 
 \label{tilde-I-NC}
\end{align}
using \eqref{XT-CR}.
As usual in Yang-Mills theories, $\cA$ transforms under gauge transformations as
\begin{align}
 \d_\L\cA = -i[T^\mu  + \cA^\mu,\L] \sim \{t^\mu,\L\}  + \{\cA^\mu,\L\}
\end{align}
for any $\L\in\cC$,
and the scalar ghost mode
 \begin{align}
\cG(\cA) = -\im [T^\mu,\cA_\mu] \sim \{t^\mu,\cA_\mu\}   
 \label{gaugefix-intertwiner}
 \end{align}
should be removed to get a meaningful theory.
This is achieved by adding a gauge-fixing term $-\cG(\cA)^2$ to the action
as well as the corresponding Faddeev-Popov (or BRST) ghost. Then the quadratic 
action becomes 
\begin{align}
 S_2[\cA] + S_{g.f} + S_{ghost} &= -\frac{2}{g^2}\Tr\, 
\left( \cA_\mu \Big(\cD^2  -\frac{3}{R^2} \Big) \cA^\mu + 2 \obar{c} \Box 
c \right) \ 
\label{eff-S-gaugefixed}
\end{align}
where $c$ denotes the  BRST ghost; see e.g.\ \cite{Blaschke:2011qu} 
for more details.
%

%
\section{Fluctuation modes}
\label{sec:modes}

We should expand the vector modes into higher spin modes according to 
\eqref{EndH-Cs-decomposition}, \eqref{Cs-explicit}
\begin{align}
 \cA^\mu &= A^{\mu}(x) + A^{\mu}_\a(x)\, t^\a +   A^{\mu}_{\a\b}(x)\, t^\a t^\b + \ldots 
  \ \in \ \cC^0   \oplus  \cC^1  \oplus  \cC^2  \oplus \ \ldots
 \label{A-M31-spins}
\end{align}
However these are neither irreducible nor eigenmodes of $\cD^2$, and the goal of this section 
is to find explicitly all eigenmodes of $\cD^2$. This will be achieved using the $\mso(4,2)$
structure and suitable intertwiners.

\paragraph{Intertwiners.}

We recall
the $SO(3,1)$ intertwiners \eq{D-properties-2}
\begin{align}
 D^\pm: \quad \cC^{(n)}\otimes\R^4 &\to \cC^{(n\pm 1)}\otimes\R^4 \nn\\
   \cA_\mu &\mapsto  D^\pm \cA_\mu
\end{align}
It is easy to show using \eq{Box-D-relation} that they satisfy
the following intertwiner property  for $\cD^2$  \cite{Steinacker:2019dii}
\begin{align}
 \cD^2 D^+\cA^{(s)} &=  D^+(\cD^2+\frac{2s+2}{R^2}) \cA^{(s)},  \nn\\
 \cD^2 D^-\cA^{(s)} &=  D^-(\cD^2-\frac{2s}{R^2}) \cA^{(s)},  \qquad  \cA^{(s)} \in \cC^{(s)} \ .
  \label{D2-D+-relation}
\end{align}
In particular,
\begin{align}
 [\cD^2, D^+D^-] &= 0 \ .
\end{align}
We also recall the $SO(3,1)$  intertwiner \eq{tilde-I-NC}
\begin{align}
\tilde\cI: \quad \cC^{s}\otimes\R^4 &\to \cC^{s}\otimes\R^4 \nn\\
   \cA^\mu &\mapsto  \{\theta^{\mu\nu},\cA_\nu\}
\end{align}
which satisfies 
\begin{align}
 [\tilde\cI,D^\pm] = 0 \ .
 \label{D-I-comm}
\end{align}
In analogy to the $SO(4,1)$ case discussed in \cite{Sperling:2018xrm}, this is related to the 
total $\mso(3,1)$ Casimir of the vector fields via
\begin{align} 
 C^2[\mso(3,1)]^{(4)\otimes (ad)} = C^2[\mso(3,1)]^{(4)} + C^2[\mso(3,1)]^{(ad)}
  - \frac 2{r^2} \tilde\cI \ 
  \label{cI-Casimir-rel}
\end{align}
where $(ad)$ indicates the adjoint action \eq{Casimirs-adjoint}.
Hence  $\tilde\cI$ describes some kind of  ''spin-orbit`` mixing.

\subsection{Diagonalization of $\cD^2$}

In \cite{Sperling:2019xar}, three series of eigenmodes $\cA_\mu$ of $\cD^2$ were found, of the form
\begin{align}
\label{A2-mode-ansatz}
 \cA_\mu^{(g)}[\phi^{(s)}] &= \{t_\mu,\phi^{(s)}\}  \quad \in \cC^{s}\,,
\\
 \cA_\mu^{(+)}[\phi^{(s)}] &= \{x_\mu,\phi^{(s)}\}\big|_{s+1} \ \equiv  
\{x_\mu,\phi^{(s)}\}_+ \quad \in \cC^{s+1} \,, \\
 \cA_\mu^{(-)}[\phi^{(s)}] &= \{x_\mu,\phi^{(s)}\}\big|_{s-1} \  \equiv  
\{x_\mu,\phi^{(s)}\}_- \quad \in \cC^{s-1} \,
\end{align}
for any $\phi^{(s)}\in\cC^s$. However there should be another series, and to find it 
we re-derive the previous results 
in a more systematic way.
We start with the easy observation \cite{Sperling:2019xar} 
\begin{align}
 \cD^2 \cA_\mu^{(g)}[\phi] &= 
\cA_\mu^{(g)}\Big[\big(\Box+\frac{3}{R^2}\big) \phi\Big] \, .
 \label{puregauge-D2}  
\end{align}
This means that  $\cA_\mu^{(g)}[\phi]$ is an eigenmode of $\cD^2$ if $\Box\phi = \l \phi$. 
Using the intertwiner properties \eq{D2-D+-relation}, we  obtain new eigenmodes 
by acting with $D^\pm$. To organize this, observe using the Jacobi identity
\begin{align}
 D^+\cA_\mu^{(g)}[\phi^{(s)}] &= \cA_\mu^{(g)}[D^+\phi^{(s)}]  + \frac{1}{R} \cA_\mu^{(+)}[\phi^{(s)}] 
    \nn\\
 D^+D^+\cA_\mu^{(g)}[\phi^{(s)}] &= \cA_\mu^{(g)}[D^+D^+\phi^{(s)}]   + \frac{2}{R} \cA_\mu^{(+)}[D^+\phi^{(s)}]  
 \label{D+D+-relations-1}
 \end{align}
 etc., and similarly
 \begin{align}
 D^-\cA_\mu^{(g)}[\phi^{(s)}] &= \cA_\mu^{(g)}[D^-\phi^{(s)}]  + \frac{1}{R} \cA_\mu^{(-)}[\phi^{(s)}] 
    \nn\\
 D^-D^-\cA_\mu^{(g)}[\phi^{(s)}] &= \cA_\mu^{(g)}[D^-D^-\phi^{(s)}] + \frac{2}{R} \cA_\mu^{(-)}[D^-\phi^{(s)}] \ .
\end{align}
Using also the intertwiner properties \eq{Box-D-relation} between $\Box$ and $D^\pm$
we recover 
\begin{align}
 \cD^2 \cA_\mu^{(+)}[\phi^{(s)}]  
    &= \cA_\mu^{(+)}\Big[\big(\Box + \frac{2s+5}{R^2} 
\big)\phi^{(s)}\Big]
    \label{D2-A2p-eigenvalues} \,, \\
 \cD^2 \cA_\mu^{(-)}[\phi^{(s)}] 
   &= \cA_\mu^{(-)}\Big[\big(\Box + 
\frac{-2s+3}{R^2}\big)\phi^{(s)}\Big] \, .
   \label{D2-A2m-eigenvalues}
\end{align}
 $D^+\cA^{(+)}$ does not give a new mode due to \eq{D+D+-relations-1},
however $D^+\cA_\mu^{(-)}[\phi^{(s)}]$ or $D^-\cA_\mu^{(+)}[\phi^{(s)}]$ do.  
These two modes are linearly dependent modulo $\cA^{(+-g)}$ due to  the Jacobi identity 
\begin{align}
D^+\cA_\mu^{(-)}[\phi^{(s)}] + D^-\cA_\mu^{(+)}[\phi^{(s)}] 
  &=  [D(\{x_\mu,\phi^{(s)}\})]_s 
 = r^2 R \{t_\mu,\phi^{(s)}\} + [\{x_\mu,D\phi^{(s)}\}]_s  \nn\\
  &= r^2 R \cA^{(g)}[\phi^{(s)}] + \cA^{(-)}[D^+\phi^{(s)}] +  \cA^{(+)}[D^-\phi^{(s)}] \ .
\label{DApm-relation}
\end{align}
Hence  either one can be used to represent the new mode (if it is independent).
We choose
\begin{align}
\boxed{
 \cA_\mu^{(n)}[\phi^{(s)}] := D^+\cA_\mu^{(-)}[\phi^{(s)}] \qquad \in \cC^{s} \ . \ 
 }
\end{align}
This provides the following list of  eigenmodes of $\cD^2$ in $\cC^{s}\otimes \R^4$
\begin{align}
 \{\cA_\mu^{(g)}[\phi^{(s)}], \ \cA_\mu^{(+)}[\phi^{(s-1)}], \ \cA_\mu^{(-)}[\phi^{(s+1)}]
 , \ \cA_\mu^{(n)}[\phi^{(s)}]\}
 \label{eigenmodes-noshift}
\end{align}
with eigenvalues 
\begin{align}
 \cD^2 \cA_\mu^{(+)}[\phi^{(s-1)}]  
    &= \cA_\mu^{(+)}\Big[\big(\Box + \frac{2s+3}{R^2} \big)\phi^{(s-1)}\Big]
    \label{D2-A2p-eigenvalues-2} \,, \\
 \cD^2 \cA_\mu^{(-)}[\phi^{(s+1)}] 
   &= \cA_\mu^{(-)}\Big[\big(\Box + \frac{-2s+1}{R^2}\big)\phi^{(s+1)}\Big] \, .
   \label{D2-A2m-eigenvalues-2}  \\
 \cD^2 \cA_\mu^{(g)}[\phi^{(s)}] &= 
\cA_\mu^{(g)}\Big[\big(\Box+\frac{3}{R^2}\big) \phi^{(s)}\Big]  \label{D2-Ag-eigenvalues} \,\\
\cD^2 \cA_\mu^{(n)}[\phi^{(s)}]
 &=  \ \cA_\mu^{(n)}\Big[\big(\Box +\frac{3}{R^2}\big)\phi^{(s)}\Big] \ .
\end{align}
The eigenvalues can be made to coincide upon inserting $D^\pm$ using \eq{Box-D-relation}, 
and
for any eigenmode of $\Box\phi^{(s)} = m^2 \phi^{(s)}$
we obtain 4-tuples of {\em ''regular`` eigenmodes} $\tilde\cA_\mu^{(i)}[\phi^{(s)}] \in \cC^{s}\otimes \R^4$  of $\cD^2$ 
\begin{align}
 \tilde\cA^{(i)}[\phi] = \begin{pmatrix}
      \cA^{(+)}[D^-\phi] \\ \cA^{(-)}[D^+\phi] \\ \cA^{(n)}[\phi] \\ r^2 R \cA^{(g)}[\phi]
    \end{pmatrix} , \qquad i,j\in\{+,-,n,g\} \ 
    \label{A-tilde-def}
\end{align}
for $\phi = \phi^{(s)}$ dropping the index $\mu$,
with the same eigenvalue 
\begin{align}
 \boxed{ \
\begin{aligned}
 \cD^2 \tilde\cA^{(+)}[\phi] &= \big(m^2 + \frac{3}{R^2}\big) \tilde\cA^{(+)}[\phi]    \\
 \cD^2 \tilde\cA^{(-)}[\phi] &= \big(m^2 + \frac{3}{R^2}\big) \,\tilde\cA^{(-)}[\phi] \\
 \cD^2 \tilde\cA^{(g)}[\phi] &= \big(m^2 + \frac{3}{R^2}\big) \, \tilde\cA^{(g)}[\phi]  \\
 \cD^2 \tilde\cA^{(n)}[\phi] &= \big(m^2 + \frac{3}{R^2}\big) \, \tilde\cA^{(n)}[\phi] \ . \
\end{aligned}
 }
  \label{Apmg-degeneracy}
\end{align}
There is  one ''special`` mode in \eq{eigenmodes-noshift} which is not covered by the regular $\tilde\cA^{(i)}$, namely 
$\cA^{(-)}[\phi^{(s,0)}]$ with
\begin{align}
 \cD^2 \cA^{(-)}[\phi^{(s,0)}] 
 = \cA^{(-)}[\big(\Box + \frac{-2s+3}{R^2}\big)\phi^{(s,0)}] \ . 
\label{special-s0}
\end{align}
We will see that it is orthogonal to all regular modes, and altogether these modes are complete.
Hence diagonalizing $\cD^2$ is reduced  to diagonalizing $\Box$ on $\cC^s$. 
In particular,
we obtain the following on-shell modes $\big(\cD^2  -\frac{3}{R^2} \big) \cA = 0$ 
\begin{align}
\{\tilde\cA^{(+)}[\phi^{(s)}], \tilde\cA^{(-)}[\phi^{(s)}], 
  \tilde\cA^{(g)}[\phi^{(s)}],\tilde\cA^{(n)}[\phi^{(s)}]\}  \qquad  &\text{for } 
       \quad   \ \Box \phi^{(s)} = 0  \ \nn\\
\cA^{(-)}[\phi^{(s,0)}]  \qquad  &\text{for } 
       \quad  \ \big(\Box - \frac{2s}{R^2} \big)\phi^{(s,0)} = 0 \ .
  \label{on-shell-all}
\end{align}
The propagation of all these modes is governed by the effective metric $G_{\mu\nu}$
\eq{eff-metric-G} encoded in $\Box$.
In particular, we note  that the on-shell relation $\Box\phi = 0$ 
determines $C^2[SO(4,1)]$ via \eq{dAlembert-fuzzy-def} for any given $SO(3,1)$ mode, corresponding some irreducible 
tensor field on the space-like $H^3$.
To put it differently, the state at any given time-slice $H^3$ 
completely determines the time evolution, up to forward or backward propagation. 
This is non-trivial in the NC case, and the time evolution is completely captured by $SO(4,1)$ group theory, even though
$\cM^{3,1}$ admits only space-like  $SO(3,1)$ isometries. 
Hence we will obtain the standard picture of time evolution even though time does not commute.
This would be hard to see in formulations based on higher-derivative star products.

In  section \ref{sec:inner}, we will establish independence and completeness of 
these modes  after dropping $\tilde\cA^{(n)}[\phi^{(s,s)}]$ (which is not independent)
and $\tilde\cA^{(+)}[\phi^{(s,0)}]\equiv 0$, 
leading to a ghost-free action and a Hilbert space upon gauge-fixing.

\subsection{Diagonalization of $\cI$ and  eigenmodes}

To establish independence of the above modes, we need to distinguish them using 
some extra observable.
Since $\tilde\cI$ is related to the total $SO(3,1)$ Casimir \eq{cI-Casimir-rel} and commutes 
with both $\Box$ and $\cD^2$, we look for 
a basis of common eigenvectors of $\cD^2$ and $\tilde\cI$ in
$\cC \otimes \R^4$. Using the above results, it suffices to diagonalize $\tilde\cI$
on the tuples \eq{A-tilde-def}, \eq{special-s0} of eigenmodes. We can use the relations \eq{tilde-I-Apm}
 \begin{align}
\begin{aligned}
\tilde\cI(\cA^{(+)}[\phi^{(s)}]) &= r^2 (s + 3) \cA^{(+)}[\phi^{(s)}] +  r^2 R 
\cA^{(g)}[D^+\phi^{(s)}] \\
\tilde\cI(\cA^{(-)}[\phi^{(s)}]) &=  r^2 (-s + 2) \cA^{(-)}[\phi^{(s)}] + r^2 R 
\cA^{(g)}[D^-\phi^{(s)}]  \ 
\end{aligned}
\end{align}
and \eq{tildeI-Ag}
\begin{align}
 R\tilde \cI(\cA^{(g)}[\phi])
  &=  (s +3)r^2 R\cA^{(g)}[\phi] + (2s +3)\cA^{(-)}[D^+\phi^{(s)}]
  + 2\cA^{(+)}[D^-\phi^{(s)}] -(2s +1) \cA^{(n)}[\phi] \nn
\end{align}
which gives 
\begin{align}
 \tilde \cI(\cA^{(n)}[\phi]) &= D^+(\tilde \cI(\cA_\mu^{(-)}[\phi^{(s)}])  \nn\\
  &= r^2 (-s + 2) \cA^{(n)}[\phi^{(s)}] 
  + r^2 \cA^{(+)}[D^-\phi^{(s)}] 
  + r^2 R \cA^{(g)}[D^+D^-\phi^{(s)}]  
\end{align}
using $[\tilde\cI,D^\pm] = 0$.
%
In terms of the $\tilde\cA^{(i)}$ \eq{A-tilde-def},
this can be summarized in matrix form as follows
\begin{align}
 \tilde\cI \begin{pmatrix}
     \tilde\cA^{(+)}[\phi] \\ \tilde\cA^{(-)}[\phi] \\ \tilde \cA^{(n)}[\phi] \\ \tilde\cA^{(g)}[\phi]
    \end{pmatrix}
  &= r^2\underbrace{\begin{pmatrix}
      s + 2 & 0 & 0 &  d_{+-} \\
     0 & -s+1 & 0 &  d_{-+} \\
     1 & 0 & -s + 2 &  d_{+-} \\
     2 & 2s +3 & -(2s +1) & s +3
    \end{pmatrix}}_{=:I}
 \begin{pmatrix}
      \tilde\cA^{(+)}[\phi] \\ \tilde\cA^{(-)}[\phi] \\ \tilde\cA^{(n)}[\phi] \\ \tilde\cA^{(g)}[\phi]
    \end{pmatrix}  
    \label{I-matrix-modes}
\end{align}
for $\phi = \phi^{(s,k)}$. 
Here we introduce the notation $D^+D^-\phi = r^2 d_{+-}\phi$ and $D^-D^+\phi = r^2 d_{-+}\phi$ 
assuming that they are diagonalized on $\phi^{(s,k)}$, which is always possible because 
$D^+D^-, D^+D^-, \Box$ are mutually commuting.
To find the eigenvalues of $I$, we compute
{
\begin{align}
 \det(I-x\one) &= d_{-+} (2 s+3) (s-x+2) (s+x-2) \nn\\
  &\quad +(s-x+3) (s+x-1) \left(-2 d_{+-} s+d_{+-}+s^2-(x-2)^2\right) \ .
 \label{char-I}
\end{align}
}
Using the commutation relations \eq{D-D+-CR-general} of $D^\pm$ on $\cC^{(s,k)}$, this 
factorizes  as 
\begin{align}
 \det(I-x\one) 
  &=  \Big((k-s)^2-(x-2)^2\Big)
 \left(-\cK -(x-2)^2 \right) \ .
  \label{char-2}
\end{align}
Here we introduce the useful quantity  
\begin{align}
 -\cK &:= s^2 + \frac{4s^2-1}{k (2s-k)}d_{+-} \ 
 = (s+1)^2 + \frac{(2s+1)(2s+3)}{(k+1) (2s-k+1)}d_{-+} \ 
 \label{chi-def}
\end{align}
(for $k=0$, the second form must be used), which is a measure for the 
kinetic energy of $\phi^{(s,k)}$ on the time slices $H^3$.
This quantity satisfies the important positivity property
\begin{lem}
 For all admissible modes $\phi$, the following estimate 
 holds\footnote{Recall that $\cK$ commutes with $\Box_H$. This  operator inequality 
 is hence a statement for $\cK$ acting on some 
 eigenspace or spectral interval of $\Box_H$.} 
 \begin{align}
 \cK\big|_{\phi} > 0  \ .
 \label{chi-positivity}
\end{align}
\label{chi-pos-lemma}
\end{lem}
which is proved in appendix \ref{sec:positivity-chi}.
We can now read off two ''regular`` (integer) eigenvalues of $\tilde\cI$ 
\begin{align}
 x_\pm = 2 \pm (k-s) \ , 
 \label{eigenvalies-I-simple}
\end{align}
which essentially measures the spin.
The corresponding left eigenvectors 
$v_\pm \cdot I = x_\pm v_\pm $
are 
\begin{align}
 v_- &= \left(
\begin{array}{cccc}
 -\frac{2 k-2 s+1}{k (k-2 s)} & -\frac{2 s+3}{k-2 s-1} & -\frac{2 s+1}{2 s-k} & 1 
\end{array}\right)\nn\\
v_+ &= \left(
\begin{array}{cccc}
 -\frac{-2 k+2 s+1}{k (k-2 s)} & -\frac{-2 s-3}{k+1} & -\frac{2 s+1}{k} & 1 
\end{array}\right) \ .
\label{vpm-def}
\end{align}
The remaining factor in \eq{char-2} leads to two
extra eigenvalues 
\begin{align}
 x'_\pm &= 2\pm\sqrt{-\cK} 
\end{align}
where $\sqrt{-\cK}$ is purely imaginary due to \eq{chi-positivity}. 
The corresponding eigenvectors for $k\neq 0$ are 
\begin{align}
 v'_\pm 
 &= \Big(\frac{2(s \mp\sqrt{-\cK})-2s+1}{s^2+\cK}\ ,
  \frac{3 + 2 s}{s + 1\pm\sqrt{-\cK}}, -\frac{1 + 2 s}{s \pm\sqrt{-\cK}}, 1\Big) \ .
  \label{vprime-def}
\end{align}
Their complexified form is somewhat misleading, and 
one can replace them by the two real modes 
\begin{align}
 v'_{1} &= \frac 12(v'_+ + v'_-) 
   = \Big(\frac{1}{s^2+\cK},\frac{(s+1)(2s+3)}{(s+1)^2+\cK},-\frac{s(2s+1)}{s^2+\cK},1\Big)     \nn\\ 
 v'_2 &= \frac{1}{2\sqrt{-\cK}}(v'_+ - v'_-)
  = \Big(-\frac{2}{s^2+\cK},-\frac{(2s+3)}{(s+1)^2+\cK},\frac{(2s+1)}{s^2+\cK},0\Big)
  \label{vprime12}
\end{align}
which span the 2-dimensional negative eigenspace of $(\cI - 2)^2$.
More precisely, they satisfy
\begin{align}
 (\cI-2)v'_{1} 
  &= -\cK\, v'_2  \ ,  \nn\\
 (\cI-2)v'_{2} 
  &=  v'_1 \ .
\end{align}
We will see in section \ref{sec:inner}
that all the $v_\pm$ and $v_\pm'$ modes are mutually orthogonal w.r.t. the invariant 
but indefinite
inner product, as they must be, and $v_\pm$ have positive norm at least on-shell.

\paragraph{Linear independence and degeneracies.}

Generically, the 4 vectors above have different eigenvalues of $\tilde\cI$, and
are therefore linearly independent.
Linear dependence can only occur if some of these eigenvalues coincide. 
Inspecting the above eigenvalues, we have to investigate the following
special cases:



\begin{itemize}

\item
$x_+=x_-$, which happens  if $k=s$.
This case will be discussed  below.

\item
 $x'_+ = x'_-$, which can only happen for $\cK=0$.
 This is ruled out by \eq{chi-positivity}.

\item
$x_\pm$  coincide with $x'_\pm$ if
 $\pm (k-s) = \sqrt{-\cK}$.
Again this cannot happen since $\cK>0$ \eq{chi-positivity}.

\item Finally for $k=0$ and $s=0$ some of the modes disappear, as discussed below.

\end{itemize}

Hence except possibly for these  special cases, the 4
regular modes  $\tilde\cA^{(i)}$ are  linearly independent.
This strongly suggests that they provide a complete set of modes,
which will be proved in section \ref{sec:completeness}.

\subsubsection{The primal sector $k=0$} 

In this case, we cannot use the above results since 
$\cA^{(+)}[D^-\phi]\equiv 0$, so that there are only 3-tuples of regular modes,
supplemented by the special mode $\cA^{(-)}[\phi]$.
For the 3-tuples, we then have 
\begin{align}
 \tilde\cI \begin{pmatrix}
     \tilde\cA^{(-)}[\phi] \\ \tilde\cA^{(n)}[\phi] \\ \tilde\cA^{(g)}[\phi]
    \end{pmatrix}
  &= r^2\underbrace{\begin{pmatrix}
      -s+1 & 0 &  d_{-+} \\
     0 & -s + 2 &  0 \\
      2s +3 & -(2s +1) & s +3
    \end{pmatrix}}_{=:I}
 \begin{pmatrix}
     \tilde\cA^{(-)}[\phi] \\ \tilde\cA^{(n)}[\phi] \\  \tilde\cA^{(g)}[\phi]
    \end{pmatrix}  
    \label{I-matrix-modes-k0}
\end{align}
for $\phi = \phi^{(s,0)}$. 
To find the eigenvalues of $I$, we compute
\begin{align}
 \det(I-x\one) 
  &= (s+x-2) \big(-\cK - (x-2)^2\big)
\end{align}
where the 2nd form of $\cK$ in \eq{chi-def} must be used.
This has one ''regular`` root 
\begin{align}
 x_0 = -s+2
\end{align}
with eigenvector
\begin{align}
 v_0 = (0,1,0)
 \label{v0-modes-k0}
\end{align}
corresponding to $\cA^{(n)}$. The two other eigenvectors corresponding to the roots
\begin{align}
x_{\pm} &= 2\pm\sqrt{-\cK}
\end{align}
are given by
\begin{align}
 v'_\pm 
  &=\left(\frac{2s+3}{s+1 \pm \sqrt{-\cK}} \ ,\frac{-2 s-1}{s\pm\sqrt{-\cK}} \ ,1\right) \ .
 \label{vpm-modes-kpm}
\end{align}
It can be checked explicitly using the results of section \ref{sec:inner} that these three modes 
are mutually orthogonal. 
Again, we can replace the complex modes 
$v'_\pm$ by 2 real modes  
\begin{align}
 v'_{1} &= \frac 12(v'_+ + v'_-) 
   = \left(\frac{(2s+3)(s+1)}{(s+1)^2+\cK} , -\frac{(2s+1)s}{s^2+\cK} ,1   \right)  \nn\\ 
 v'_2 &= \frac{1}{2\sqrt{-\cK}}(v'_+ - v'_-)  
  = \left(-\frac{(2s+3)}{(s+1)^2+\cK} , \frac{2s+1}{s^2+\cK} , 0 \right) \ 
  \label{vprime-def-k0}
\end{align}
which are linearly independent.
In addition to the above three modes, there is an extra mode:

\paragraph{Special massless spin $s$ mode $\cA_\mu^{(-)}[\phi^{(s,0)}]$.}

For $k=0$ consider the extra  mode 
\begin{align}
 v_0^- := \cA_\mu^{(-)}[\phi^{(s,0)}] \ .
 \label{v0m-special}
\end{align}
This is not contained in the previous modes $v_{\pm}, v_0$ because $\phi^{(s,0)}$ 
cannot be written as $D^+ \phi'$. Hence it complements the 3 regular modes, so that  each $\phi^{(s,0)}$
determines again 4 independent modes. 
The on-shell condition  $\big(\cD^2  -\frac{3}{R^2} \big) \cA^{(-)}[\phi^{(s,0)}] = 0$ 
 takes the slightly different form $(\Box -  \frac{2s}{R^2})\phi^{(s,0)} = 0$, 
 due to \eq{D2-A2m-eigenvalues}.
This mode satisfies
\begin{align}
 x^\mu \cA_\mu^{(-)}[\phi^{(s,0)}] = 0 \ 
 \label{A-s0-spacelike}
\end{align}
due to \eq{Apmg-time} i.e. it is space-like, since
$x^\mu$ defines the time-like direction (e.g. at a reference point
$\xi = (\xi_0,0,0,0)$ on $\cM^{3,1}$).
Positivity of the inner product then follows immediately, in agreement with
the direct computation in section \ref{sec:inner}.
Moreover \eq{A2-gaugefix} implies that this mode is  physical, and we will see that for $s=2$
it provides the 2 standard degrees of freedom 
of the physical graviton \cite{Sperling:2019xar}.

\paragraph{The case $s=0$.}

In this case (which implies $k=0$), not only the $\cA_\mu^{(+)}$ mode vanishes but also  $\cA_\mu^{(n)}=0$,
because $\cA^{(-)}[\phi^{(0)}] = 0$. The above special mode 
$v_0^-$  \eq{v0m-special} also disappears, and only the $v'_{\pm}$ survive 
among the above modes, or equivalently $\cA_\mu^{(-)}[D^+\phi]$ 
and $\cA_\mu^{(g)}[\phi]$.
We will see below that their inner products are non-degenerate, and these 2 modes are complete 
for $s=0$. This is consistent with the case of $H^4_n$
studied in \cite{Sperling:2018xrm} and the case of $S^4_N$ in \cite{Sperling:2017gmy}, 
where also two tangential modes were obtained for $s=0$, and 4 modes for $s\geq 1$.

\subsubsection{The scalar sector $k=s \neq 0$} 
\label{sec:s-s-tau}

For $k=s \neq 0$,  $\phi^{(s,s)} = (D^+)^s \phi^{(0)}$ is the $s$-fold space-like divergence of a 
scalar mode\footnote{Recall that $(D^+)^s\phi^{(0)}$ are {\em space-like} scalar modes 
in the sense that the 3-dimensional $SO(3,1)$ spin on $H^3$ vanishes, since $D$
commute with $SO(3,1)$.}.
Then the eigenvalues $x_\pm$  of $\tilde\cI$ and
in fact also the corresponding modes $v_\pm$ \eq{vpm-def} coincide,
\begin{align}
 v_+ = v_- = \Big(
 \frac{1}{s^2}, \frac{2s+3}{s+1}, -\frac{2s+1}{s}, 1 \Big) \ .
\label{v-null-1}
\end{align}
However, 
we will see in section \ref{sec:v-null-vanish}  that in fact $v_+ = v_- = 0$ vanishes identically.
A substitute  can be found by formally taking the limit 
\begin{align}
 v_{\rm extra} &:= \lim\limits_{k\to s}\frac {v_+-v_- }{k-s} \
 = \lim\limits_{k\to s}\frac 1{k-s}
 \left(\begin{array}{cccc}
 \frac{2 k-2 s-1}{k (k-2 s)} +\frac{2 k-2 s+1}{k (k-2 s)},
 & \frac{2 s+3}{k+1} +\frac{2 s+3}{k-2 s-1} ,
 & -\frac{2 s+1}{k} +\frac{2 s+1}{2 s-k} ,
 & 0
\end{array}\right) \nn\\[1ex]
&=  -2\Big(\frac{2}{s^2},\frac{2s+3}{(s+1)^2} , -\frac{2s+1}{s^2} , 0 \Big) \ .
\label{extra-mode-k=l}
\end{align}
We will see in section \ref{sec:k=s-inner} that $v_{\rm extra}$ has positive norm and is orthogonal to $v'_\pm$,
and  there are no further scalar modes.

\section{Inner product matrix}
\label{sec:inner}

Now that we have identified the eigenmodes of $\cD^2$, we  can 
compute the inner product matrix with respect to \eq{inner-basic-def}. 
This will confirm and complete the results of the previous section, and allow to
determine the signature of the inner product matrix for all admissible modes.
We can then establish a no ghost theorem providing a 
Hilbert space of physical modes. Moreover, the off-shell results provide all the information needed to 
obtain the full propagator.

For $\phi,\phi'\in\cC^{s}$  we define the inner product matrix\footnote{It is important to observe that the modes are 
integrable over the entire $\cM^{3,1}$, rather than just the space-like $H^3$. The reason is that 
we consider the principal series unitary irreps of $SO(4,2)$ in $End(\cH_n)$, which 
correspond to square-integrable tensor fields on $H^4$. 
This allows to use invariance relations such as $(D^-)^\dagger = - D^+$.
Although semi-classically one could define an inner product based on $H^3$, this 
would not  make sense in the fully NC case.}
\begin{align}
 \cG^{(i,j)} = \left\langle \tilde\cA_\mu^{(i)}[\phi'],\tilde\cA^{\mu (j)}[\phi] \right\rangle
 , \qquad i,j\in\{+,-,n,g\}
  \label{Gij-matrix}
\end{align} 
with the $\tilde\cA^{(i)}$  defined in \eq{A-tilde-def}.
The matrix elements are computed explicitly in appendix \ref{sec:inner-prod-calc}.
They can be evaluated easily e.g. in Mathematica, since 
the entries $D_{\pm}$ and $\Box$ mutually commute, and can be simultaneously diagonalized 
for any fixed mode $\phi^{(s,k)}$.
Then the space of modes boils down to 4-dimensional blocks 
$\tilde\cA^{(i)}$ which are mutually orthogonal. The metric in the blocks is non-degenerate but indefinite since $\eta_{\mu\nu}$ has
Minkowski signature, and 
one can verify explicitly using the commutation relations \eq{D+D--CR-general} that 
$\tilde\cI$ is hermitian, i.e.
\begin{align}
 (\cG\, I^T)^{(i,j)} = (I\, \cG)^{(i,j)} \, 
 \label{I-G-consistency}
\end{align}
where $I$ is the matrix defined in \eq{I-matrix-modes}. This provides a highly non-trivial consistency check.

Let us discuss the results in detail,
assuming first  $s\neq k\neq 0$.
One can then check explicitly that all modes $\{v_\pm,v'_\pm\}$ \eq{vpm-def}, \eq{vprime-def} 
are mutually orthogonal, as they must be.
The norm of the vectors $v_\pm$  is obtained (e.g. using Mathematica) as follows
\begin{align}
  \langle v_+ , v_+\rangle 
 &= r^4\frac{2 (k-s)^2 \left(\cK+(s-k)^2 \right)}
 {k^2 (k+1) (2 s-k)}
 \left(\cK - R^2 \Box +s^2 + k -1 \right) \ .
 \label{inner-v+}
 \end{align}
Here $\Box$ is understood to act\footnote{recall that $\Box$ commutes with $\cK$ and $\Box_H$.} 
on $\phi$ resp. $\phi'$, 
as resulting from the inner product formulas in section \ref{sec:inner-prod-calc}.
The factor $(\cK+(s-k)^2)$ is positive since $\cK>0$, due to lemma \ref{chi-pos-lemma}.
The factor 
\begin{align}
 (\cK - R^2 \Box +s^2 + k -1) = r^{-2}\Box_H  -k(2s-k-1) -2s -2
 \label{K-Box-est}
\end{align}
(using \eq{BoxH-explicit-CR}) is positive using the estimate
\begin{align}
  r^{-2}\Box_H &\geq k (2s - k-1) +2s+2 \ ,
\end{align}
which follows from the admissibility condition  \eq{admissible} for $\Box_H$
\begin{align}
 r^{-2}\Box_H  > s^2+s+2  &\geq k (2s - k-1) +2s+2  
\end{align}
which reduces to $(s-k)(s-k-1) \geq 0 $.
Therefore $\langle v_+ , v_+\rangle  > 0$, and
similarly
\begin{align}
 \langle v_- , v_-\rangle 
 &= r^4\frac{2 (k-s)^2 \left(\cK + (s-k)^2\right)}
  {k (2 s-k+1) (2 s-k)^2}
 \left(\cK - R^2\Box +s^2 -k +2 s-1\right) \ > 0 \ .
 \label{inner-v-}
\end{align}
Now consider the $v'_\pm$ modes. Since they are complexified, 
we refrain from computing their 
scalar product. The overall signature of $\cG^{(i,j)}$ can be determined more easily from the 
determinant of the full inner product matrix \eq{inner-product-det},
which is found to be 
\begin{align}
 \det(\cG^{(i,j)}) 
 &= r^{16}\, \frac{ d_{+-}  
 (k+1) (2s-k+1) (k-s)^2}  {\left(4 s^2-1\right) (4 s (s+2)+3)^2} \cK
 \left(\cK+(s+1)^2\right)\cdot \nn\\
 &\quad \cdot \left((-R^2\Box +\cK +s^2 +s-1)^2 - (s-k)^2\right) 
 \left((R^2\Box +k^2-2 k s+1-s)^2 +\cK\right) \ . 
 \label{inner-product-det}
\end{align}
The first factor in the second line arises from the $v_\pm$ modes, and is 
positive as shown above.
The last factor arises from the $v'_\pm$ modes and is also positive.
Using $d_{+-}<0$ and $\cK >0$ we obtain 

\begin{lem}
\label{lemma-orthmodes-generic}
In any 4-dimensional space of modes $\tilde\cA^{(i)}[\phi], \  i\in\{+,-,n,g\}$
for admissible $\phi\in\cC^{(s,k)}$ with $s\neq k$ and $k\neq 0$,
the metric $\cG^{(i,j)}$ is non-degenerate with signature  $(+++-)$. 
\end{lem}
 This is the core of the no-ghost theorem, as discussed below.
 The special cases $k=0$ and $s=k$ will be discussed separately below.
 Off-shell, the signature $(+++-)$ should be important e.g. in the context of loop computations
 and  to establish perturbative unitarity and causality statements.

\paragraph{Explicit inner product for $v'_{12}$ modes.}

We can compute the inner product for the $v'_{1,2}$ modes defined in \eq{vprime12}.
Their inner product is given by 
\begin{align}
 \langle v'_i,v'_j\rangle 
  &= r^4\frac{(2s+1)(-\cK) \left(\cK+s^2 -k (k-2 s)\right)}
  {(\cK+s^2)^2 (\cK +(s+1)^2)}\,
  \begin{pmatrix}
   a & b \\
   b & -\frac{a}{\cK}
  \end{pmatrix},  \nn\\[1ex]
 a &= -\frac{\cK+s^2}{2 s+1} \left(R^2\Box +k^2-2 k s+s+2\right)
       +s \left(R^2\Box +k^2-2 k s+1\right) \ ,  \nn\\
 b &= -R^2 \Box -1 - k^2  + 2 k s + \frac{\cK + s^2}{2 s + 1}
\label{vprime-inner}
\end{align}
noting that 
\begin{align}
 k (2s-k)(\cK+s^2) = -(4s^2-1)d_{+-} \ .
\end{align}
Here the $2\times 2$ matrix  has negative determinant 
\begin{align}
 \det \begin{pmatrix}
   a & b \\
   b & -\frac{a}{\cK}
  \end{pmatrix} 
  = -\frac{1}{\cK (2 s+1)^2}
  (\cK +s^2) (\cK +(s+1)^2) \left(\cK + (R^2\Box +k^2-2 k s-s+1)^2\right) <0 \ ,
  \label{det-inner-vp12}
\end{align}
and we recognize the last factor from  \eq{inner-product-det}. 
One could now select a canonical basis of two null vectors if desired.

\subsection{The primal sector $k=0$}
\label{sec:k=0}

For primal modes $k=0$, 
the $ \tilde \cA^{(+)}[\phi^{(s,0)}] = \cA^{(+)}[D^-\phi^{(s,0)}]=0$ mode vanishes, and the inner product matrix simplifies 
accordingly.
One can check again that the 3 eigenmodes $v'_{\pm},v_0$ in 
\eq{v0-modes-k0}, \eq{vpm-modes-kpm}
are mutually orthogonal, with 
\begin{align}
 \langle v_0,v_0\rangle 
   &=\frac{r^4 s}{(2 s+1)^2} \big(\cK + s^2 \big) 
  \left(-R^2\Box + \cK +s^2 -1\right) \ .
\end{align}
This is again positive for admissible on-shell modes using $\cK > 0$.
The 
determinant of the inner product matrix for these 3 modes is
\begin{align}
 \det(\cG^{(i,j)}) 
 &= r^{12}\frac{d_{-+} s }{(2 s+1)^2 (2 s+3)}
 \cK\left(\cK - R^2\Box  +s^2 -1\right) 
 \left(\cK + (1 + R^2\Box  - s)^2\right) \ . 
 \label{det-Gij-k0}
\end{align}
Since $d_{-+}<0$,
it follows as in \eq{K-Box-est} ff. that the determinant is 
negative for all admissible modes.

Now recall the extra special mode  $\cA_\mu^{(-)}[\phi^{(s,0)}]$ \eq{v0m-special}.
It is easy to see from the explicit formulas for the inner products
in section \ref{sec:inner-prod-calc} that this mode is orthogonal to all other modes,
and its inner product is positive as already observed in \cite{Sperling:2019xar}.
Therefore we have 
\begin{lem}
\label{lemma-orthmodes-k0}
In any 3-dimensional space of modes $\tilde\cA^{(i)}[\phi], \  i\in\{-,n,g\}$
for admissible $\phi\in\cC^{(s,0)}$ with $s\neq 0$,
the metric $\cG^{(i,j)}$ is non-degenerate with signature  $(++-)$. 
These 3 modes are orthogonal to $\cA^{(-)}[\phi^{(s,0)}]$,
which has positive norm.
\end{lem}

\paragraph{Explicit inner product for $v'_{12}$ modes.}

We can compute the inner product for the $v'_{1,2}$ modes defined in \eq{vprime-def-k0}.
Their inner product is given by 
\begin{align}
 \langle v'_i,v'_j\rangle 
  &= \frac {r^4 (2s+1) (-\cK)}{(\cK + (s+1)^2)  (\cK + s^2)} \,
  \begin{pmatrix}
   a & b \\
   b & -\frac{a}{\cK}
  \end{pmatrix},  \nn\\[1ex]
  a &=  -\frac{\cK + (s+1)^2}{2s+1} \left(R^2\Box +s+2\right) + (s+1)\left( R^2\Box+2\right)   \nn\\
  b &= -R^2\Box -2 + \frac{\cK + (s+1)^2}{2s+1} 
\end{align}
noting that 
\begin{align}
 -\cK &= (s+1)^2 + (2s+3)d_{-+} \ .
\end{align}
The $2\times 2$ matrix  again has negative determinant,
\begin{align}
 \det \begin{pmatrix}
   a & b \\
   b & -\frac{a}{\cK}
  \end{pmatrix} 
  &= -\frac 1{\cK} (\cK+(s+1)^2) (\cK + s^2) 
   \left(\cK + (1 + R^2\Box - s)^2\right) <0 \ ,
     \label{det-inner-vp12-k0}
\end{align}
and we recognize the last factor from  \eq{det-Gij-k0}. 
A  basis of two null vectors can be found if desired.

\paragraph{The $s = 0$ sector.}

As discussed above there are only two modes
$v'_{\pm}$ in this case, 
since $\cA_\mu^{(n)}  \equiv 0$ vanishes identically.
The considerations of the $v'_{1,2}$ modes defined in \eq{vprime-def-k0}
goes through, and the determinant of the $2\times 2$ inner product matrix is 
still given by \eq{det-inner-vp12-k0} evaluated at $s=0$,
\begin{align}
 \det \begin{pmatrix}
   a & b \\
   b & -\frac{a}{\cK}
  \end{pmatrix} 
  &= -(\cK+1)  \left(\cK + (1 + R^2\Box)^2\right) <0 \ .
     \label{det-inner-vp12-k00}
\end{align}
Hence the signature is $(+-)$, and we obtain 
\begin{lem}
In any 2-dimensional space of modes $\tilde\cA^{(i)}[\phi], \  i\in\{-,g\}$
for admissible $\phi\in\cC^{(0,0)}$,
the metric $\cG^{(i,j)}$ is non-degenerate with signature  $(+-)$. 
\label{lemma-orthmodes-k00}
\end{lem}

\subsection{The scalar sector $k=s\neq 0$}
\label{sec:k=s-inner}

In this case, \eq{inner-product-det} gives $ \det(\cG^{(i,j)}) =0$.
This means that there is a null mode, 
which is of course precisely the mode found in \eq{v-null-1}.
In fact we  show in section \ref{sec:v-null-vanish} that it vanishes identically,
\begin{align}
v_{\pm} = (\frac 1{s^2}, \frac{2s+3}{1 + s}, - \frac {2s+1}s, 1) \ = v_{\rm null} = 0 \ .
\label{v-null-vanish}
\end{align}
One can check that the extra mode \eq{extra-mode-k=l} 
\begin{align}
 v_{\rm extra} = 
-2\left(\frac{2}{s^2},\frac{2s+3}{(s+1)^2} , -\frac{2s+1}{s^2} , 0 \right)
\end{align}
is orthogonal to both $v'_{\pm}$, and its inner product 
is positive,
\begin{align}
 \langle v_{\rm extra},v_{\rm extra}\rangle
  &=  \frac {r^4}{s^3 (1 + s)}  \cK(\cK  - \Box R^2 + s^2 + s-1) > 0 \ .
\end{align}
However, we will see that
$v_{\rm extra}$ is not physical.
This extra mode also explains why there is only one factor $(k-s)^2$ in $\det(\cG^{(i,j)})$ \eq{inner-product-det},
which arises from the inner products of either $v_\pm$ \eq{inner-v+}.
The $v'_{\pm}$ modes can again be replaced by $v'_{1,2}$ \eq{vprime12}, and 
the inner product in the space spanned by  $v'_{1,2}$ has signature $(+-)$, which 
can be inferred from 
\begin{align}
\det\langle v'_i,v'_j\rangle 
 &= - \frac{r^8 \cK^3 
 \left(\cK + (R^2\Box-s (s+1)+1)^2\right)} {(\cK+s^2)^3 \left(\cK+ (s+1)^2 \right)} <0
 \label{inner v12-ks}
\end{align}
as before, using  $s^2 (\cK+s^2) = -(4s^2-1)d_{+-}$.
This means that there are 3 linearly independent modes 
$\{v'_{1,2},v_{\rm extra}\}$  whose metric has signature $(-++)$, and we have established
\begin{lem}
In any 3-dimensional space of modes $\tilde\cA^{(i)}[\phi], \  i\in\{+,-,g\}$
for admissible $\phi\in\cC^{(s,s)}$ with $s\neq 0$,
the metric $\cG^{(i,j)}$ is non-degenerate with signature  $(++-)$. 
\label{lemma-orthmodes-ks}
\end{lem}

These modes are equivalently spanned by $v_{\rm extra},v'_1,v'_2$, while 
the  $\tilde\cA^{(n)}[\phi]$ mode
is a linear combination of these modes via \eq{v-null-vanish}.

To summarize, we have identified the following scalar modes: 

\paragraph{\und{$\cA\in\cC^0$}:}

The scalar modes $\cA\in\cC^0$ are given by 
$\tilde\cA^{(-)}[\phi]$ and $\cA^{(g)}[\phi]$ for $\phi\in\cC^0$, with non-degenerate metric with signature $(+-)$.
The mode $\cA^{(n)}[\phi]$ vanishes identically. 

\paragraph{\und{$\cA\in\cC^s$:}}

The scalar modes  $\cA\in\cC^s$ for $s\neq 0$ are given by 
$\tilde\cA^{(-)}[\phi]$ and $\tilde\cA^{(+)}[\phi]$ and $\tilde\cA^{(g)}[\phi]$ 
for $\phi = \phi^{(s,s)} = (D^+)^s \phi^{(0)}$, 
with non-degenerate metric with signature $(++-)$.
We will see that for $s=1$, the only physical mode in this sector  leads to scalar metric perturbations, 
and in particular to the linearized Schwarzschild solution \cite{Steinacker:2019dii}.

\subsection{Completeness}
\label{sec:completeness}

Now we want to understand whether the above modes are complete, i.e. 
if they span the space of all fluctuations $\cA$. 
This will be addressed  by counting the number of degrees of freedom
(dof), i.e. real scalar fields on $\cM^{3,1}$, 
at each sector $\cA_\mu \in \cC^s \otimes \R^4$.

\paragraph{\und{$\cA_\mu \in \cC^0 \otimes \R^4$}.}

This sector clearly contains 4 dof.
Among the above modes, only the spin 1 mode $\cA^{(-)}[\phi^{(1)}]$ and 
the spin 0 modes $\cA^{(g)}[\phi^{(0)}]$ are in $\cC^0 \otimes \R^4$, 
while  $\cA^{(n)}[\phi^{(0)}]$ vanishes.
It follows from the previous considerations that all these 
modes are independent.
Now $\cA_\mu^{(-)}[\phi^{(1)}]$ i.e. $\phi^{(1)}$ encodes the most general space-like vector field on 
$\cM^{3,1}$, cf. \eq{A-s0-spacelike}, which amounts to 
3 degrees of freedom. Together with the spin 0 mode  
$\cA^{(g)}[\phi^{(0)}]$ we obtain 4 dof, which is precisely the content of 
$\cC^0 \otimes \R^4$. It follows that the above list of modes is complete.
 These modes are elaborated explicitly in section \ref{sec:C0-modes-explicit}.

\paragraph{\und{$\cA_\mu \in \cC^s \otimes \R^4$}.}

This sector  contains $4(2s+1)$ dof.
It is convenient to ignore the $(s,k)$ substructure of the $\cC^s$ here.
Among the above modes, $\cA^{(-)}[\phi^{(s+1)}]$,
$\cA^{(n)}[\phi^{(s)}], \cA^{(g)}[\phi^{(s)}]$  and  $\cA^{(+)}[\phi^{(s-1)}]$
are in $\cC^s \otimes \R^4$. If they were all independent, this would 
provide all the $(2s+3)+2(2s+1)+(2s-1) = 4(2s+1)$ dof.
The above results show that these modes are linearly independent
{\em except} for the scalar sector discussed in section \ref{sec:k=s-inner}, which provides only 
3 rather than 4 modes due to the relation \eq{v-null-vanish}. 
Therefore there must be {\em one exceptional scalar dof} for each $s\geq 1$,
\begin{align}
 \cA^{(ex,s)}\in \cC^s \otimes \R^4 , \qquad s \geq 1 \ .
\end{align}
Since none of the regular scalar modes $\tilde\cA^{(i)}$ is null,
we can choose  $\cA^{(ex,s)}$  to be orthogonal to all  
$\tilde\cA^{(i)}$. Due to the explicit form of the $\cA^{(i)}$, this implies that 
the $\cA^{(ex,s)}$ can be chosen as follows
\begin{align}
\{t^\mu,\cA^{(ex,s)}_\mu\} = 0 = \{x^\mu,\cA^{(ex,s)}_\mu\} \ , \qquad 
 \cA^{(ex,s)} = (D^+)^{s-1}\cA^{(ex,1)} \ .
\label{exceptional-modes-properties}
\end{align}
Orthogonality  implies that
this sector is  respected by $\cD^2$, and the  physical constraint 
is satisfied.
Further details are discussed in appendix \ref{sec:except-modes},
however the explicit form of $\cA^{(ex,s)}$ is not known.

Taking these exceptional modes into account, 
we have recovered all  $4(2s+1)$ dof in $\cC^s \otimes \R^4$,
 so that the list of modes is complete.
Together with the above lemmas, we have shown
\begin{thm}
The  $\tilde\cA^{(i)}[\phi^{(s)}]$ modes \eq{A-tilde-def} along with the $\cA^{(-)}[\phi^{(s,0)}]$ for all $s\geq 0$ and
the exceptional modes $\cA^{(ex,s)}$ for $s\geq 1$
span the space of all fluctuations $\cA$.
A basis is  obtained by dropping $\tilde\cA^{(n)}[\phi^{(s,s)}]$ and  $\tilde\cA^{(+)}[\phi^{(s,0)}]$.
\label{thm:basis}
\end{thm}

From a representation theory point of view, we have essentially
decomposed the tensor product
\begin{align}
 \cC \otimes \R^4 = \oplus (...)
\end{align}
into $SO(3,1)$ irreps. It is  natural to expect that
that each irrep in $\cC$ arises  with multiplicity 4 on the rhs, and we have seen
that this holds indeed for the regular modes.
However for non-compact Lie groups, the appearance of extra modes $\cA^{(ex,s)}$ in the tensor product is not too surprising.

\subsection{Physical constraint, Hilbert space and no ghost}
\label{sec:physical}

We first observe that an (admissible, i.e. integrable) fluctuation mode $\cA$
satisfies the gauge-fixing condition 
$\{t^\mu,\cA_\mu\} = 0$ if and only if it is orthogonal to all pure gauge modes,
\begin{align}
 \langle \cA^{(g)},\cA\rangle = 0 \ .
 \label{gaugefix-inner}
\end{align}
Now consider an on-shell mode $\cA\in\cC^s$ in some 4-dimensional mode space $\tilde \cA^{(i)}[\phi], \ i\in\{+-ng\}$ determined by some $\phi\in\cC^{(s,k)}$
with $\Box\phi = 0$ and  $s>k>0$. 
Since that 4-dimensional space of modes
has signature $(+++-)$ due to Lemma \ref{lemma-orthmodes-generic} and $\cA^{(g)}$ is null, 
the gauge-fixing constraint \eq{gaugefix-inner} leads to a 3-dimensional subspace 
with signature $(++0)$, which contains $\cA^{(g)}$.
Then the usual definition 
\begin{align}
 \cH_{\rm phys} = \{\mbox{gauge-fixed on-shell modes}\}/_{\{\mbox{pure gauge modes}\}  }
 \label{H-phys}
\end{align}
leads  to 2 modes with positive norm. 
This establishes the generic part of 
\begin{thm}
 The space  $\cH_{\rm phys}$ \eq{H-phys} of admissible solutions of 
 $\big(\cD^2-\frac{3}{R^2}\big)\cA = 0$ which are gauge-fixed $\{t^\mu,\cA_\mu\} = 0$ 
 modulo pure gauge modes inherits a positive-definite inner product, and forms a Hilbert space.
\label{thm:no-ghost}
\end{thm}

\begin{proof}
The same argument works for the on-shell modes $\tilde \cA^{(i)}[\phi]\in\cC^s$ with 
primal $\phi\in\cC^{(s,0)}$.
For $s\neq 0$ there are 2 physical modes. One is given by a linear combination of the 
 $\tilde \cA^{(i)}[\phi], \ i\in\{-ng\}$ which has signature $(++-)$ before gauge fixing.
 In addition there is an extra on-shell physical mode 
 $\cA^{(-)}[\phi^{(s,0)}] \in \cC^{s-1}$ for
 $\big(\Box - \frac{2s}{R^2} \big)\phi^{(s,0)} = 0 $ \eq{on-shell-all}.
 
For $s=0$,  no physical mode arises from the $\tilde \cA^{(i)}\in\cC^0$ with $i\in\{-g\}$ 
which has signature $(+-)$ before gauge fixing, due to Lemma  \ref{lemma-orthmodes-k00}.
For the scalar on-shell modes 
$\tilde \cA^{(i)}[\phi]\in\cC^s, \ \ i\in\{+-g\}$ with 
$\phi\in\cC^{(s,s)}$, there is one physical linear combination 
according to Lemma \ref{lemma-orthmodes-ks}.
Finally, the exceptional modes $\cA^{(ex,s)}$ \eq{exceptional-modes-properties}
are physical, and their norm is positive because the $2s+1$ dof 
in $\cC^s\otimes \R^4$ with negative norm are
already accounted for by the regular modes as shown above.

The admissibility condition implies
square-integrability as discussed in 
\eq{inner-tensorfield-BoxH}. Together with the completess 
theorem \ref{thm:basis}, the statement follows. 

\end{proof}

Observe that the inner product \eq{Gij-matrix} for vector modes is precisely 
realized in the quadratic action \eq{eff-S-gaugefixed}. Hence the above theorem
is tantamount to the statement that the quadratic action is free of ghosts, i.e. physical modes with negative norm.
Although the  result is established only at the semi-classical (Poisson) level, 
most of the steps  would  go through in the non-commutative case using the $\mso(4,2)$-covariant
quantization map $\cQ$ \eq{quantization-map}, with minor adaptions due to the cutoff. 
Hence we expect  that the theorem holds also in the  non-commutative case.

There is no obstacle to determine $\cH_{\rm phys}$ explicitly.
 It turns out that none of the modes $v_\pm$  and $v'_\pm$ satisfy the physical constraint, 
 hence non-trivial combinations are required, and we can just as well us the $\tilde\cA^{(i)}$ modes. 
A simplification  arises for low spin, since the $\cA^{(-)}[\phi^{(2,*)}]\in\cC^1$ modes
are all physical due to \eq{A2-gaugefix}. 
This leads to the following sectors of  $\cH_{\rm phys}$:

\paragraph{The physical modes $\cA_\mu \in \cC^0$.}

As explained above, the off-shell modes $\cA_\mu \in \cC^0$ comprise 
the spin 1 mode $\cA^{(-)}[\phi^{(1)}]$ and 
the spin 0 modes $\cA^{(g)}[\phi^{(0)}]$ are in $\cC^0 \otimes \R^4$.
These modes are elaborated explicitly in section \ref{sec:C0-modes-explicit}.
Among these, only the spin 1 modes $\cA^{(-)}[\phi^{(1,0)}]$ are physical, and 
\begin{align}
  \cH_{\rm phys} \cap \cC^0 = \{\cA^{(-)}[\phi] \ \mbox{for} \  \phi\in\cC^{(1,0)},\ 
   \big(\Box - \frac{2}{R^2} \big)\phi = 0 \} \ .
\end{align}
These modes satisfy $\del^\mu\cA_\mu = 0 = x^\mu \cA_\mu$,
and describe a spin 1 Yang-Mills (or Maxwell) field.

\paragraph{The physical modes $\cA_\mu\in \cC^1$.}

They arise from the 12 
off-shell modes $\cA^{(-)}[\phi^{(2)}]$, $\cA^{(n)}[\phi^{(1)}]$, $\cA^{(g)}[\phi^{(1)}]$ and $\cA^{(ex,1)}$ modulo 
the relation \eq{vnull-s1-vanish}. 
Among these, all  $\cA^{(-)}[\phi^{2)}]$ are physical due to \eq{A2-gaugefix}, 
and so is the exceptional scalar mode $\cA^{(ex,1)}$, whose  on-shell condition is not known explicitly.
We claim that there are no further physical states in this sector, so that
\begin{align}
  \cH_{\rm phys} \cap \cC^1 = \{\cA^{(-)}[\phi]\ \mbox{for} \  \phi\in\cC^{(2,*)}, \ 
   \big(\Box - \frac{4}{R^2} \big)\phi = 0 \}\
   \cup\ \{\cA^{(ex,1)}; \ (\cD^2-\frac{3}{R^2}) \cA^{(ex,1)}=0  \}\ .
   \label{Hphys-C1}
\end{align}
They satisfy $\{t^\mu,\cA_\mu\} = 0$, and $x^\mu \cA_\mu[\phi^{(2,0)}]=0$.
To see this, 
note that  $\cA^{(n)}[\phi^{(1,0)}]$ is in the same tuple of primal spin 1 modes as 
$\tilde\cA^{(-)}[\phi^{(1,0)}]$ and $\tilde\cA^{(g)}[\phi^{(1,0)}]$ which contains only one physical mode 
due to Lemma \ref{lemma-orthmodes-k0},
 given by $\cA^{(-)}[D^+\phi^{(1,0)}]=\tilde\cA^{(-)}[\phi^{(1,0)}]$. 
 Note that the on-shell condition in \eq{Hphys-C1}
for $\phi = D^+\phi^{(1,*)}$
is equivalent to $\Box\phi^{(1,*)} = 0$ due to \eq{Box-D-relation}.
 Similarly,
$\cA^{(+)}[\phi^{(0)}] \sim \tilde\cA^{(+)}[D^+\phi^{(0)}]$ is in the same tuple of scalar modes as 
$\tilde\cA^{(-)}[D^+\phi^{(0)}]$ 
and $\tilde\cA^{(g)}[D^+\phi^{(0)}]$, and due to 
Lemma \ref{lemma-orthmodes-ks} only $\tilde\cA^{(-)}[D^+\phi^{(0)}] = \cA^{(-)}[D^+D^+\phi^{(0)}]$ is physical.
Again the on-shell condition in \eq{Hphys-C1}
for $\phi = D^+D^+\phi^{(0)}$
is equivalent to $\Box D^+\phi^{(0)} = 0$.

The modes in \eq{Hphys-C1} govern the linearized gravity sector,
as discussed below.

\paragraph{The physical modes $\cA_\mu\in \cC^s$ with $s\geq 2$.}

For the regular modes, the physical constraint $\{t^\mu,\cA_\mu\} = 0$  
must be solved directly. 
 To determine $\cH_{\rm phys}$, we can drop any contribution from $\cA^{(g)}$. 
  The simplest case is the  mode \eq{on-shell-all} which is always physical 
  due to \eq{A2-gaugefix-1},
\begin{align}
  \{\cA^{(-)}[\phi^{(s+1,0)}]\ \mbox{for}  \ 
   \big(\Box - \frac{2(s+1)}{R^2} \big)\phi^{(s+1,0)} = 0 \} \ \subset \ \cH_{\rm phys} \cap \cC^s  \ .
   \label{Hphys-Cs-except}
\end{align}
All other modes 
are contained in some regular $\tilde\cA^{(i)}$ tuple, and we need to work a bit harder.
The gauge fixing constraint for the $\cA^{(\pm)}$  modes is given in \eq{A2-gaugefix-1}, and
for the $\cA^{(n)}$  mode it is
\begin{align}
 \{t^\mu,\cA_\mu^{(n)}[\phi^{(s)}]\} &= \{t^\mu,D^+\cA_\mu^{(-)}[\phi^{(s)}]\}
  = D^+\{t^\mu,\cA_\mu^{(-)}[\phi^{(s)}]\} - \frac 1R \{x^\mu,\cA_\mu^{(-)}[\phi^{(s)}]\}_+   \nn\\
  &= \frac{-s+2}{R} D^+D^-\phi^{(s)} - \frac 1R \{x^\mu,\cA_\mu^{(-)}[\phi^{(s)}]\}_+ \nn\\
  &= \frac{1}{R}\Big( (-s+3)D^+D^- + \a_s (\Box_H - 2 r^2(s+1))\Big) \phi^{(s)}  
  \label{An-phys-1}
 \end{align} 
 using \eq{magic-formula-xmu}, consistent with \eq{An-Ag-inner}. 
 We should hence determine all on-shell linear combinations 
\begin{align}
 \cA_\mu^{(\rm phys)}[\phi] &=  c_+ \tilde\cA^{(+)}[\phi] + c_-\tilde\cA^{(-)}[\phi] 
     + c_n\cA^{(n)}[\phi] , \qquad \Box\phi = 0
\end{align}
 for $\phi\in\cC^s$ which satisfy the gauge-fixing constraint
\begin{align}
 0 &= R\{t^\mu,\cA^{(\rm phys)}_\mu[\phi]\} \nn\\
 &= \Big(c_+ (s+2) D^+ D^-
 + c_- (-s+1) D^- D^+
  + c_n \big(\a_s (\Box_H - 2 r^2(s+1)) + (-s+3) D^+ D^-\big)\Big) \phi \ .
  \label{physical-constraint-c}
\end{align} 
Replacing $\Box_H$ on-shell using \eq{BoxH-explicit-CR} allows to recast this
 into a 3-dimensional constraint on $H^3$,
 but does not lead to a simple expression.  The first two terms are  non-trivial since $s\geq 2$.
%
%
%
%

Consider first
the primal tuple $\tilde\cA^{(i)}[\phi^{(s,0)}]$ for $i=-,n,g$. This contains one physical mode 
due to Lemma \ref{lemma-orthmodes-k0}, which we can choose to be a linear combination with $c_n=1$, 
\begin{align}
 \{c_-\tilde\cA^{(-)}[\phi^{(s,0)}] + \tilde\cA^{(n)}[\phi^{(s,0)}] \ \ \mbox{for}  \ 
   \Box\phi^{(s,0)} = 0 \ \mbox{and} \ \eq{physical-constraint-c} \} \ \subset \ \cH_{\rm phys} \cap \cC^s  
   \label{Hphys-Cs-primal}
\end{align}
where $c_-$ is determined by solving the above constraint.
Next, the scalar tuple  $\tilde\cA^{(i)}[\phi^{(s,s)}]$ for $i=+,-,g$ contains also one physical mode 
due to Lemma \ref{lemma-orthmodes-ks}, which we can choose to be 
\begin{align}
 \{\tilde\cA^{(-)}[\phi^{(s,s)}] + c_+\tilde\cA^{(+)}[\phi^{(s,s)}] \ \ \mbox{for}  \ 
   \Box\phi^{(s,s)} = 0 \ \mbox{and} \ \eq{physical-constraint-c} \} \ \subset \ \cH_{\rm phys} \cap \cC^s  \ .
   \label{Hphys-Cs-scalar}
\end{align}
Next, the generic tuple $\tilde\cA^{(i)}[\phi^{(s,k)}]$ for $i=+,-,n,g$ and  $s\neq k\neq 0 $ contains two physical modes
due to Lemma \ref{lemma-orthmodes-generic}, which we can choose to be 
\begin{align}
 \{\tilde\cA^{(-)}[\phi^{(s,k)}] + c_+\tilde\cA^{(+)}[\phi^{(s,k)}] \ \ \mbox{for}  \ 
   \Box\phi^{(s,k)} = 0 \ \mbox{and} \ \eq{physical-constraint-c} \} \ &\subset \ \cH_{\rm phys} \cap \cC^s  \nn\\
    \{c_-\tilde\cA^{(-)}[\phi^{(s,k)}] + \tilde\cA^{(n)}[\phi^{(s,k)}] \ \ \mbox{for}  \ 
   \Box\phi^{(s,k)} = 0 \ \mbox{and} \ \eq{physical-constraint-c} \} \ &\subset \ \cH_{\rm phys} \cap \cC^s  
   \label{Hphys-Cs-generic}
\end{align}
the first of which was found in \cite{Sperling:2019xar}. 
Finally, the exceptional scalar modes $\cA^{(ex,s)}$ are always physical,
upon imposing the on-shell condition
\begin{align}
 \{\cA^{(ex,s)}; \ (\cD^2-\frac{3}{R^2}) \cA^{(ex,s)}=0  \}\ \ &\subset \ \cH_{\rm phys} \cap \cC^s  
\end{align}
This completes the list of physical modes.

\paragraph{Discussion.}

To summarize, 
the model contains generically 2 physical modes parametrized by $\phi^{(s)} \in\cC^s$ with $\Box \phi^{(s)} = 0$ for each spin $s\geq 2$,
up to the exceptional cases discussed above. 
The $\phi^{(s)}$  are ''would-be massive`` spin $s$ modes, i.e.
they contain the $2s+1$ dof of massive spin $s$ multiplets with vanishing mass parameter, and  they decompose
further into a series of irreducible massless spin $s-k$ modes 
(in radiation gauge) 
as discussed in section \ref{sec:higher-spin}. 
These modes mix under the higher-spin gauge transformations.
It is hence plausible that some of these modes  become massive
in the interacting theory, which remains to be clarified.

Furthermore, we recall that at least for the regular modes,
the above Hilbert space is determined uniquely 
be the wavefunction on any space-like slide $H^3$.
More precisely, the 4-dimensional Casimir  $C^2[\mso(4,1)]$
is determined on-shell by the  space-like Casimir $C^2[\mso(3,1)]$. 
These statements  apply also in the fully noncommutative case, resulting in a 
picture which is  quite close to the usual setup in field theory.

%
%
%
%



%

\section{Metric fluctuation modes}
\label{sec:graviton}

To illustrate the physical relevance of the above results, we 
briefly discuss how  metric fluctuations arise from the above  
modes, elaborating on \cite{Sperling:2019xar}.
The effective metric for functions of $\cM^{3,1}$ on a perturbed background $Y = T + \cA$ can be extracted from the kinetic 
term  in \eqref{scalar-action-metric}, which defines the 
bi-derivation
\begin{align}
\begin{aligned}
 \g:\quad \cC\times \cC  \ &\to  \quad \cC  \\
  (\phi,\phi') &\mapsto \{Y^\a,\phi\}\{Y_\a,\phi'\} \ .
  \label{metric-full}
  \end{aligned}
\end{align}
Specializing to $\phi=x^\mu, \phi' = x^\nu$ we obtain  the coordinate form
\begin{align}
\g^{\mu\nu} &=  \obar\g^{\mu\nu} + \d_\cA \g^{\mu\nu} + [\{\cA^\a,x^\mu\}\{\cA_\a,x^\nu\}]_0
\label{gamma-nonlinear}
\end{align}
 whose linearized contribution in $\cA$ is given by
\begin{align}
\begin{aligned}
 \d_\cA \g^{\mu\nu} 
  = \sinh(\eta) \{\cA^\mu,x^\nu\}_0 + (\mu \leftrightarrow \nu) \ .
  \label{gravitons-H1}
  \end{aligned}
\end{align}
The projection on $\cC^0$ ensures that this is the metric for functions on $\cM^{3,1}$.
Clearly only $\cA \in \cC^1$ can contribute to $\d_\cA \g^{\mu\nu}$, which we assume henceforth.
To evaluate this explicitly,
it is convenient to consider the following rescaled graviton mode: 
\begin{align}
  h^{\mu\nu}[\cA] &\coloneqq   \{\cA^\mu,x^\nu\}_0 + (\mu \leftrightarrow \nu) ,
 \qquad h[\cA] = 2\{\cA^\mu,x_\mu\}_0 \ .
 \label{tilde-H-def}
\end{align}
Including the conformal factor in \eq{eff-metric-G},
this leads to the effective metric fluctuation \cite{Steinacker:2019dii}
\begin{align}
 \d G^{\mu\nu} &= \b^2\big(h^{\mu\nu} -\frac{1}{2} \eta^{\mu\nu}\,h\big) \ .
\end{align}
Let us discuss the mode content of $h^{\mu\nu}[\cA]$.
Recall that the 12 off-shell dof in $\cA_\mu = \cA_{\mu;\a} t^\a \in\cC^1$
are realized by $\cA^{(-)}[\phi^{(2)}]$,
$\cA^{(n)}[\phi^{(1)}], \cA^{(g)}[\phi^{(1)}]$, 
$\cA^{(+)}[\phi^{(0)}]$ and $\cA^{(ex,1)}$. Hence the 10 dof of the 
most general off-shell metric fluctuations are provided by 
$\cA^{(-)}[\phi^{(2)}]$, 
$\cA^{(g)}[\phi^{(1)}]$, and the scalar modes $\cA^{(+)}[\phi^{(0)}]$ and $\cA^{(ex,1)}$. The 6 physical metric fluctuations\footnote{In particular, even though  
$\cA^{(n)}[\phi^{(1)}]$ encodes off-shell dof of a  space-like 2-form in \eq{A-M31-spins},
it is not physical. However the 2-form may be determined by 
the metric modes arising from $\cA^{(-)}[\phi^{(2,1)}]$, which are  
physical. } 
arise from $\cA^{(-)}[\phi^{(2)}]$ and $\cA^{(ex,1)}$.
According to the results of section \ref{sec:physical},
the 5 physical  
would-be massive modes $\cA^{(-)}[\phi^{(2)}]$ decompose into the
massless graviton $\cA^{(-)}[\phi^{(2,0)}]$, 
one massless vector mode $\cA^{(-)}[\phi^{(2,1)}]$, and one scalar mode
$\cA^{(-)}[\phi^{(2,2)}]$.
The vector field can be extracted by
\begin{align}
 \{t_\mu,h^{\mu\nu}\} 
  &=  \{t_\mu, \{\cA^\mu,x^\nu\}_-\} + \{t_\mu, \{\cA^\nu,x^\mu\}_-\} \nn\\
  &=  \{\{t_\mu, \cA^\mu\},x^\nu\}_-  - \frac{2}{R} D^- \cA^\nu \nn\\
  &\stackrel{phys}{=} - \frac{2}{R} D^- \cA^\nu 
  \label{metric-div-general}
\end{align}
using the Jacobi identity and \eq{A2-gaugefix}, which
vanishes  for the $\cA^{(-)}[\phi^{(2,0)}]$ mode. 
Together with
\begin{align}
 \{t_\nu,\{t_\mu,h^{\mu\nu}\}\} 
  &=  \{t_\nu,\{\{t_\mu, \cA^\mu\},x^\nu\}_-\} 
     - \frac{2}{R}  \{t_\nu,D^- \cA^\nu\}  
 \ \stackrel{phys}{=} - \frac 1{R^2} h
  \label{metric-div-double}
\end{align}
we obtain
\begin{align}
 \{t_\nu,\{t_\mu,h^{\mu\nu}\}\} + \frac{1}{R^2}h = 0 
   \qquad \mbox{for physical} \ \cA \ .
  \label{metric-div-double-2}
\end{align}
This constraint is satisfied by the 
physical scalar metric mode arising from $\cA^{(-)}[D^+D^+\phi]$, which underlies
the linearized Schwarzschild solution \cite{Steinacker:2019dii}.

\section{Conclusions and outlook}

The results of this paper demonstrate that the model under 
consideration defines a consistent and ghost-free 
higher spin gauge theory in 3+1 dimensions, at least at the linearized level.
It leads to  truncated towers of higher-spin modes, which include spin 2
fluctuation modes of the effective metric 
leading  to Ricci-flat  metric perturbations as shown in \cite{Sperling:2019xar}, 
and  the linearized Schwarzschild solution as shown in  \cite{Steinacker:2019dii}.
Since it  is defined in terms of a maximally supersymmetric Yang-Mills  matrix model,
it is  plausible that this defines in fact a consistent quantum theory 
which includes gravity.
The crucial feature in contrast to standard Yang-Mills theories is that space-time is not put in 
by hand, but emerges in the semi-classical limit 
from the background solution given in terms of 3+1 large (in fact infinite) matrices.

Let us briefly discuss briefly the quantization of the model.
Even though the noncommutative space has only finitely many degrees of freedom per volume,
it is not automatic that the theory is finite and approximately local, because of UV/IR mixing \cite{Minwalla:1999px}.
It is well-known that in  NC field  theories, the UV degrees of freedom are dominated by  string-like modes,
which have both IR and UV properties and violate the Wilsonian paradigm. 
In order to have a good locality {\em and} UV behavior, their contributions in loops must cancel. 
It is also known that in 4 dimensions, sufficient cancellations occur basically only in the 
maximally supersymmetric $\cN=4$ case \cite{Matusis:2000jf,Jack:2001cr,Hanada:2014ima}. 
But this is precisely what happens in the IKKT matrix model.
In fact, one can view the present model as noncommutative $\cN=4$ SYM \cite{Aoki:1999vr}  with 
$\hs$ - valued gauge fields, where $\hs$ is the {\em finite} higher-spin-like
''algebra`` generated by $\theta^{\mu\nu}$ or $t^\mu$. 
This suggests that the theory should be UV finite at all loops, and 
it is manifest from the generic formulas in \cite{Ishibashi:1996xs,Blaschke:2011qu} 
that the one-loop effective action is indeed
finite, cf. \cite{Steinacker:2016vgf}.
Of course the argument is not fully justified since $\hs$ is not a standard 
Lie algebra but includes some $x$-dependence; nevertheless the similarity with $\cN=4$ SYM suggests that 
the present model might provide a UV-finite quantum theory including spin 2. This is certainly intriguing, and 
vindicates more detailed investigations.

Although the model is not yet sufficiently developed, it is tempting to compare and relate it with 
other approaches to quantum gravity.
Conformal or quadratic gravity (cf. \cite{Salvio:2018crh} and references therein)
 is reminiscent of Yang-Mills theory and is 
renormalizable \cite{stelle1977renormalization}, but contains ghosts. A similar issue may be expected 
in asymptotic safety scenarios \cite{Reuter:2012id}.
In contrast, we have seen that
the  present model does {\em not} contain ghosts, as the 
fundamental degrees of freedom are different and arise from matrix fluctuations.
String theory in its conventional formulation 
can claim to provide 9+1-dimensional (quantum) gravity, 
however compactification to 3+1 dimensions leads to a lack of predictivity known as
the landscape problem.
This is avoided in the IKKT model, which can be viewed as different, constructive approach 
to string theory.
Hence the present matrix model and the type of background under consideration 
may provide the basis for a consistent and useful 3+1-dimensional 
quantum theory including gravity, however it remains to be seen whether the resulting physics
is viable.

\paragraph{Acknowledgements.}

 I would like to thank C-S Chu, C. Iazeolla, H. Kawai, J. Nishimura and  E. Skvortsov
 for useful discussions, and I am grateful to
 E. Delay and W. Schlag for pointing me to the appropriate mathematical literature.
  This work was supported by the Austrian Science Fund (FWF) grant P32086-N27.

\section{Appendix}

\subsection{Useful identities for the vector modes $\cA$}

We recall the following gauge-fixing identities for the vector modes $\cA^\mu$ 
\begin{align}
\{t^{\mu},\cA_\mu^{(+)}[\phi^{(s)}]\} &= \frac{s+3}{R} D^+\phi^{(s)} \,,\nn\\
\{t^{\mu},\cA_\mu^{(-)}[\phi^{(s)}]\} &= \frac{-s + 2}{R} D^-\phi^{(s)} \, 
 \label{A2-gaugefix}
\end{align}
for $\phi^{(s)} \in\cC^s$, which follow  from (A.34) in \cite{Sperling:2019xar} 
\begin{align}
 R \{t^{\mu},\{x_\mu,\phi^{(s)}\}\} &= \frac 12\Big(\frac 12 \cS^2 - s(s+1) + 4\Big) D\phi^{(s)} 
 = (s+3) D^+\phi^{(s)} + (-s + 2) D^-\phi^{(s)} \nn\\
 R \{x_\mu,\{t^{\mu},\phi^{(s)}\}\}
 &= (s-1) D^+\phi^{(s)} - (s +2) D^-\phi^{(s)} \ .
 \label{t-x-CR-end}
\end{align}
In particular, we note
\begin{align}
 \{t^{\mu},\cA_\mu^{(+)}[D^-\phi^{(s)}]\} &= \frac{s+2}{R} D^+D^-\phi^{(s)} \,,\nn\\
\{t^{\mu},\cA_\mu^{(-)}[D^+\phi^{(s)}]\} &= \frac{-s + 1}{R} D^-D^+\phi^{(s)} \, .
 \label{A2-gaugefix-1}
\end{align}
The time component of $\cA^{(\pm)}$ along the vector field $\t$ \eq{tau-def} can be obtained 
using \eq{D-properties} as
\begin{align}
 x^\mu \cA_\mu^{(\pm)}[\phi^{(s)}] &= -x_4 D^\pm \phi^{(s)} \ .
 \label{Apmg-time}
\end{align}

\subsubsection*{Intertwiner relations for $\Box$ and $\cA$}

The following relations were shown in \cite{Sperling:2019xar}
\begin{align}
  \Box D^-\phi^{(s)} &= D^-\Big(\Box - \frac {2s}{R^2}\Big)\phi^{(s)}  \nn\\
  \Box D^+\phi^{(s)} &= D^+\Big(\Box + \frac {2s+2}{R^2}\Big)\phi^{(s)}  \nn\\
  \Box D^+D^- \phi^{(s)} &= D^+D^-\Box  \phi^{(s)}
    \label{Box-D-relation}
\end{align}
as well as
 \begin{align}
 \label{tilde-I-Apm}
\begin{aligned}
\tilde\cI(\cA_\mu^{(+)}[\phi^{(s)}]) &= r^2 (s + 3) \cA_\mu^{(+)}[\phi^{(s)}] +  r^2 R 
\{t_{\mu},D^+\phi^{(s)}\} \\
\tilde\cI(\cA_\mu^{(-)}[\phi^{(s)}]) &=  r^2 (-s + 2) \cA_\mu^{(-)}[\phi^{(s)}] + r^2 R 
\{t_{\mu},D^-\phi^{(s)}\}  \ 
\end{aligned}
\end{align}
and 
\begin{align}
 \cD^2 \cA_\mu^{(+)}[\phi^{(s)}] 
 &= \cA_\mu^{(+)}\Big[\Big(\Box + \frac{2s+5}{R^2} \Big)\phi^{(s)}\Big] \nn\\
 \cD^2 \cA_\mu^{(-)}[\phi^{(s)}] 
   &= \cA_\mu^{(-)}\Big[\Big(\Box + \frac{-2s+3}{R^2}\Big)\phi^{(s)}\Big] \, .
\end{align}
Since $\cD^2\cA = (\Box  +\frac{2}{r^2R^2} \tilde\cI)\cA$,
these two relations can be combined to obtain  
\begin{align}
 \Box \cA_\mu^{(\pm)}[\phi^{(s)}] 
   &=\cA_\mu^{(\pm)}\Big[\Big(\Box -\frac{1}{R^2}\Big)\phi^{(s)}\Big] \, 
- \frac{2}{R}\cA_\mu^{(g)}[D^\pm\phi^{(s)}] \ .
   \label{Box-A2m-eigenvalues}
\end{align}

\subsubsection*{Intertwiner relations for $\Box_H$ and $\cA$}

The $SO(4,1)$ intertwiner relation 
 \begin{align}
 r^2 C^2[\mso(4,1)]^{\rm (full)} \cA_a[\phi^{(s)}] 
 &= -(\Box_H  + 2 \cI^{(5)} - r^2(\cS^2 +4)) \cA_a[\phi^{(s)}] 
  = \cA_a[r^2 C^2[\mso(4,1)]\phi^{(s)}] \nn
\end{align}
(cf. (D.30) in \cite{Sperling:2018xrm})
can be used to derive several useful identities for $\Box_H$. In particular
for  $\cA_a = \cA^{(\pm)}_a[\phi^{(s)}] = \{x_a,\phi^{(s)}\}_\pm$ and $a=0,...,4$, one obtains
\begin{align}
 \Box_H \cA^{(-)}_a[\phi^{(s)}] 
  &= \cA_a^{(-)}[(\Box_H - 2r^2s)\phi^{(s)}] \nn\\
 \Box_H \cA^{(+)}_a[\phi^{(s)}] 
  &= \cA_a^{(+)}[(\Box_H + 2r^2(s+1))\phi^{(s)}] 
  \label{BoxH-A-relation}
\end{align}
 using $ \cI^{(5)}\cA^{(-)}_a[\phi^{(s)}] = r^2(2-s)\cA^{(-)}_a[\phi^{(s)}]$
 and  $\cI^{(5)}\cA^{(+)}_a[\phi^{(s)}] = r^2(s+3)\cA^{(+)}_a[\phi^{(s)}]$, cf. (5.48) in \cite{Sperling:2018xrm}.
%
This implies for $a=4$
\begin{align}
  \Box_H D^-\phi^{(s)} &= D^-((\Box_H - 2r^2s)\phi^{(s)})  \nn\\
  \Box_H D^+\phi^{(s)} &= D^+((\Box_H + 2r^2(s+1))\phi^{(s)})  \nn\\
  \Box_H D^+D^- \phi^{(s)}  
  &= D^+D^-\Box_H  \phi^{(s)} \ .
    \label{BoxH-D-relation}
\end{align}
These are completely analogous to the relation for $\Box$ \eq{Box-D-relation}, and can also be checked directly.
%
It is also easy to see (e.g. using their expression in terms of Casimirs) that
\begin{align}
 [\Box_H,\Box] = 0 = [\Box_H,\cD^2] \ .
 \label{Box-H-Box-comm}
\end{align}
Together with \eq{BoxH-estimate-onshell-principal}, we also obtain 
\begin{align}
  (-C^2[\mso(4,1)] + (s+1)(s+2)) D^+ \phi^{(s)} &= (\Box_H - r^2(s+1)(s+2))D^+ \phi^{(s)} \nn\\
   &= D^+ (\Box_H - r^2s(s+1))\phi^{(s)}\nn\\
   &= D^+ (-C^2[\mso(4,1)] + s(s+1))\phi^{(s)}
   \label{admissible-D}
\end{align}
an similarly for $D^-$,
which means via \eq{admissible-so41} that $D^\pm$ preserves admissible modes.

\subsubsection*{Evaluation of $\tilde \cI(\cA^{(g)})$}

Consider for $\phi\in\cC^{(s)}$
\begin{align}
 \tilde \cI(\cA^{(g)}[\phi]) &= \{\theta^{\mu\nu},\{t_\nu,\phi\}\}  
   = \{\{x^\mu,x^\nu\},\{t_\nu,\phi\}\}  \nn\\
  &=  - \{\{x^\nu,\{t_\nu,\phi\}\},x^\mu\} - \{\{t_\nu,\phi\},x^\mu\},x^\nu\} \nn\\
  &=  - \{\{x^\nu,\{t_\nu,\phi\}\},x^\mu\} 
    + \{\{\phi,x^\mu\},t_\nu\},x^\nu\} + \{\{x^\mu,t_\nu\},\phi\},x^\nu\}   \nn\\
  &=  - \frac 1R\{(s-1) D^+\phi - (s +2) D^-\phi,x^\mu\}  \nn\\
   &\quad - \{\{\cA_\mu^{(+)}[\phi],t_\nu\},x^\nu\} - \{\{\cA_\mu^{(-)}[\phi],t_\nu\},x^\nu\} 
   - \frac 1R \{D\phi,x^\mu\}  \nn\\
  &=   \frac{s}R  \cA^{(-)\mu}[D^+\phi]
    -\frac{(s +1)}R\cA^{(+)\mu}[D^-\phi]  
    + \frac{(s +3)}R D^-\cA_\mu^{(+)}[\phi]
          - \frac{(s-2)}R D^+\cA_\mu^{(-)}[\phi]  \nn
\end{align}
using \eq{t-x-CR-end}.
Using the definition of $\cA_\mu^{(n)}$ and \eq{DApm-relation},  this gives
\begin{align}
 R\,\tilde \cI(\cA^{(g)}[\phi])
  &=  (s +3)r^2 R \cA^{(g)} + (2s +3)\cA^{(-)\mu}[D^+\phi]
  + 2\cA^{(+)\mu}[D^-\phi] -(2s +1) \cA_\mu^{(n)}[\phi] \ .
  \label{tildeI-Ag}
\end{align}

\subsection{Explicit vector modes $\cA\in \cC^0$}
\label{sec:C0-modes-explicit}

We give explicitly the fluctuation modes discussed in section \ref{sec:completeness}.
For $\phi^{(1)} =  \phi_\a t^\a$ we have
\begin{align}
 \cA^{(-)}_\mu[\phi^{(1)}] &= \{x^\mu,\phi_{\a} t^{\a}\}_{0}  
  = \del_\nu\phi_{\a} [\theta^{\mu\nu}t^{\a}]_0
   +  \phi_{\a} \{x^\mu, t^{\a}\} \nn\\
&= \frac{1}{3}
 \sinh(\eta) (x^\mu\del^\a\phi_\a  - (\t+3)\phi_\mu )  
  + \frac{1}{3} x_\b \varepsilon^{\b 4\a\mu \nu} \del_\nu\phi_\a  
   \label{Am-1}
\end{align}
using \eq{averaging-relns}.
The last term is the 3-dimensional rotation on $H^3$.
This vector field separates into the space-like divergence-free field
\begin{align}
 \cA^{(-)}_\mu[\phi^{(1,0)}] &= - \frac{1}{3} \sinh(\eta) (\t+3)\phi_\mu   
  + \frac{1}{3} x_\b \varepsilon^{\b 4\a\mu \nu} \del_\nu\phi_\a , \nn\\
  \del^\mu \cA_\mu  &= 0 =  x^\mu \cA_\mu
  \label{spacelike-A10}
\end{align}
(hence in radiation gauge)
using \eq{Apmg-time}, and the scalar mode
\begin{align}
 \cA^{(-)}_\mu[D\phi] &=\frac{r^2R}{3}
 \sinh(\eta) (x_\mu\del^\a\del_\a\phi  - (\t+3)\del_\mu\phi)  , \nn\\
  \del^\mu\cA_\mu  &= - \frac{1}{R^2\sinh^2(\eta)} x^\mu \cA_\mu 
  \label{A-Dphi0-explicit}
\end{align}
using \eq{A2-gaugefix}
for $\phi\in\cC^0$, which is neither space-like nor 
divergence-free\footnote{Incidentally, the explicit form of $\cA^{(-)}_\mu[D\phi]$ shows that the scalar mode
$\cA_\mu^{(\t)}[\phi] = x_\mu\phi$  as discussed in 
\cite{Steinacker:2019dii} is  a linear combination of $\cA_\mu^{(-)}$ and $\cA_\mu^{(g)}$.}. 
The  remaining mode in $\cC^0$ is the pure gauge mode  
\begin{align}
 \cA^{(g)}_\mu[\phi^{(0)}] &= \{t_\mu,\phi^{(0)}\} = \sinh(\eta)\del_\mu \phi^{(0)} \ .
\end{align}
This illustrates the sub-structure of tensor fields resulting from the reduced $SO(3,1)$ covariance.
The only physical mode in this sector is $\cA^{(-)}_\mu[\phi^{(1,0)}]$, which corresponds
to a massless vector field.

\subsection{Wick theorem for averaging over $S^2$}

\begin{lem}
\label{lemma-wick-1}
\begin{align}
 [t^{\a_1} ...  t^{\a_{2s}}]_{0} 
  &= b_{2s} \sum_{i<j} [t^{\a_i} t^{\a_j}] [t .... t]_0, 
  \qquad b_{2s}  = \frac{3}{s(2s+1)} 
\label{bs-explicit}
\end{align}
\end{lem}
i.e. $b_2=1, \quad b_4 = \frac{3}{10}, \quad b_6 = \frac 17$ etc.

\begin{proof}
The structure of the rhs follows from the fact that all totally symmetric  $SO(3,1)$-invariant tensors 
are obtained from $\eta^{\a\b}$.
The constants $b_{2s}$ can be determined  either using a recursive combinatorial argument by  
contracting with $\eta_{\a_1\a_2}$, 
or implicitly \& recursively from
\begin{align}
 [t_3^{2s}]_0 &= \frac 12 2s(2s-1) b_{2s}[t_3t_3][t_3^{2s-2}]_0 = ... 
 = \frac 1{2^s} b_{2s} b_{2s-2} ... b_2 (2s)! [t_3 t_3]^s \nn\\
 &= 3^s \frac 1{2^s}[t_3 t_3]^s \frac{(2s)!}{s!(2s+1)(2s-1) ... 1} \nn\\
 &= 3^s [t_3 t_3]^s \frac{(2s)! s!}{s!(2s+1)!} 
  = 3^s [t_3 t_3]^s \frac{1}{2s+1} 
   = \frac{\cosh^s(\eta)}{r^{2s}} \frac{1}{2s+1} \ 
\end{align}
at the reference point $\xi$.
Taking into account the local radius of $S^2$,
this agrees with \eq{average-int-def}
\begin{align}
 \frac 1{4\pi}\int_{S^2}\cos(\vartheta)^{2s}   2\pi\sin(\vartheta)d\vartheta 
  &= \frac 12 \int_{-1}^1 du u^{2s} = \frac{1}{2} \frac{1}{2s+1}[u^{2s+1}]^1_{-1} 
  = \frac {1}{2s+1} \ .
\end{align}
\end{proof}

We will also need the following  variant of Wicks theorem:

\begin{lem}
\label{lemma-wick-2}
\begin{align}
 [t^{\a_1} ...  t^{\a_{s+1}}]_{s-1}
  &= c_{s+1} \sum_{i<j} [t^{\a_i} t^{\a_j}] [t .... t]_{s-1}, \qquad c_{s+1}= \frac{3}{2s+1}
  \label{c-s-formula}
\end{align}
summing over all contractions,
where $[.]_{s-1}$ denotes the projection on $\cC^{s-1}$.
\end{lem}

\begin{proof}
The constants $c_{s+1}$ can be determined by
contracting with $\eta_{\a_1\a_2}$: 
\begin{align}
 [(t^\mu t_\mu) t^{\a_3} ...  t^{\a_{s+1}}]_{s-1} 
  &= c_{s+1}\Big([t^\mu t_\mu] [t^{\a_3} ...  t^{\a_{s+1}}]_{s-1}  
  + \sum_i [t^{\mu} t^{}] [t_\mu t .... t]_{s-1}  
 + \sum_j [t^{}t^{\mu} ] [t_\mu .... t ]_{s-1} \Big) \nn\\
 &= c_{s+1}\frac{\cosh^2(\eta)}{r^2}\Big( [t^{\a_3} ...  t^{\a_{s+1}}]_{s-1}
 +  2(s-1)\frac{1}{3} [t^{\a_3} .... t^{\a_{s+1}}] _{s-1}\Big)
\end{align}
noting that no contractions can occur in the last term, and 
using \eq{kappa-average}
\begin{align}
 [t^{\a}t^{\mu}]_0 t^\mu &= \frac{\cosh^2(\eta)}{3r^2}P_\perp^{\a\mu}t_\mu 
    = \frac{\cosh^2(\eta)}{3r^2}t_\a 
\end{align}
as well as $t^\mu t_\mu = \frac{\cosh^2(\eta)}{r^2}$.
Thus
\begin{align}
 [t^{\a_3} ...  t^{\a_{s+1}}]_{s-1} 
  &= c_{s+1}\Big( [t^{\a_3} ...  t^{\a_{s+1}}]_{s-1}
 +  \frac 23(s-1) [t^{\a_3} .... t^{\a_{s+1}}] _{s-1}\Big)  \nn\\
  &= c_{s+1}\big(1+ \frac 23(s-1)\big) [t^{\a_3} ...  t^{\a_{s+1}}]_{s-1}
\end{align}
which implies \eq{c-s-formula}.
\end{proof}

\subsection{Computation of $\a_s$}

We want to show the useful formula
\begin{lem}
\begin{align}
 -\{x^a,\{x_a, \phi^{(s)}\}_-\}_+ &= \a_s (\Box_H -2r^2(s+1)) \phi^{(s)} , \qquad \a_s = \frac{s}{2s+1} \ .
 \label{magic-formula-x+x-}
 \end{align}
 \label{lemma-alpha}
 \end{lem}
This formula was derived in  \cite{Sperling:2018xrm} using the representation \eq{phis-rep-deriv-tensorfields} for $s=1$,
and  for general $s$ based on an indirect argument; however $\a_s$ was not yet found for $s>2$.
The structure of the  formula is not surprising, since
 the lhs is a $SO(4,1)$-invariant 2nd order derivation on $\cC^s$, which can only be $\Box_H$ up to some constants.
%
%
Here we provide a  direct proof, using  the result for $s=1$. 
\begin{proof}

For $s=1$, the formula \eq{magic-formula-x+x-} was proved in \cite{Sperling:2018xrm} 
 for  $\phi = \{x^a,\phi_a\}$ for any tangential divergence-free vector field $\phi_a$, and 
 it is not hard to see that all  $\phi \in\cC^1$ can be written in this way\footnote{For example, 
 it suffices to show this for polynomial functions on $\C P^{1,2}$, for which the representation $\phi = \{x^a,\phi_a\}$ 
 can be shown using Young diagrams along the lines in \cite{Sperling:2017gmy}. 
 It is also easy to see that  $\{x^a,\phi_a\} = 0$ for $\phi_a = \eth_a\phi$. For more details
 we refer to \cite{Sperling:2018xrm}.}.
 Using this result for $\phi = f\theta^{ab}$
 as well as \cite{Sperling:2018xrm} 
 \begin{align}
 \Box_H &= -r^2 R^2\eth^d\eth_d  
 \label{BoxH-eth-formula}
\end{align}
where $\eth$  is defined in \eq{eth-def},
we obtain
\begin{align}
 \Box_H ( f\theta^{ab}) &=  -r^2 R^2  f\theta^{ab}\eth^d\eth_d f -2  r^2 f\theta^{ab} 
   -2  r^2 R^2 (\eth^d\theta^{ab})\eth_d f \ .
\label{BoxH-ftheta}
\end{align}
On the other hand, 
\begin{align}
 -\{x^c,\{x_c, f\theta^{ab}\}_-\}_+ 
  &= -2r^2 f\theta^{ab}
   - r^2R^2 (\eth^d\theta^{ab})\eth_d f   -  \{x^c,[\theta^{ab}\theta^{cd}]_0\eth_d f\}   \nn\\
  &\stackrel{!}{=} \frac 13 (\Box_H -4r^2)(f\theta^{ab})  \nn\\
  &=  -\frac 13 r^2 R^2  \theta^{ab}\eth^d\eth_d f - 2 r^2 f\theta^{ab}  - \frac 23 r^2 R^2 (\eth^d\theta^{ab})\eth_d f \ 
\end{align}
which gives the useful formula
\begin{align}    
    -  \{x^c,[\theta^{ab}\theta^{cd}]_0\eth_d f\}  
    &= \frac 13 r^2 R^2 \big(- \theta^{ab}\eth^d\eth_d f  + (\eth^d\theta^{ab})\eth_d f \big) \ .
    \label{s=1-formula-x+x-}
\end{align}
Now consider
the following constant  modes
\begin{align}
 \phi^{(s)} = \phi_{a_1...a_s;b_1..b_s} \theta^{a_1b_1} ... \theta^{a_sb_s} \qquad \in \cC^{s}
 \end{align}
where  $\phi_{a_1...a_s;b_1..b_s}\in\C$ are  traceless with 
the symmetry of a Young diagram
${\tiny \Yvcentermath1 \young(aaa,bbb) }$. Then 
\begin{align}
 -\{x^a, \phi^{(s)}\}_- &= -s \phi_{a_1...a_s;b_1..b_s} \theta^{a_1b_1} ... \{x^a, \theta^{a_sb_s}\} 
\end{align}
and
\begin{align}
 -\{x^a,\{x_a, \phi^{(s)}\}_-\}_+ 
  &= -s\phi_{a_1...a_s;b_1..b_s}  \theta^{a_1b_1} ... \{x^a,\{x_a, \theta^{a_sb_s}\}   \}_+   
  = -2 r^2 s \phi^{(s)}
\end{align}
since $-\{x^a,\{x_a, \theta^{bc}\}\} = -2 r^2 \theta^{bc}$.
It is easy to see that this coincides with 
\begin{align}
 \Box_H \phi^{(s)} &= - 2 r^2 s \phi^{(s)}
\end{align}
because $\phi_{a_1...a_s;b_1..b_s}$ is  traceless.
Therefore 
\begin{align}
 -\{x^a,\{x_a, \phi^{(s)}\}_-\}_+ &= \a_s (\Box_H - 2r^2(s+1)) \phi^{(s)} 
 =  -2r^2(2s+1)\a_s  \phi^{(s)} 
 \end{align}
 and we obtain 
 \begin{align}
 \a_s = \frac{s}{2s+1} \ . 
 \end{align}
Now consider general (non-constant) modes  in $\cC^s$ for  $s \geq 2$.
They are spanned by modes obtained by multiplying the above constant modes $\phi^{(s)}$ with 
some functions: 
\begin{align}
f(x) \phi^{(s)} = f(x)\phi_{a_1...a_s;b_1..b_s} \theta^{a_1b_1} ... \theta^{a_sb_s}  \qquad \in \cC^{s} \ .
 \end{align}
Then
\begin{align}
 -\{x^c, f\phi^{(s)}\}   &= -s \phi_{a_1...a_s;b_1..b_s} f\theta^{a_1b_1} ... \{x^c, \theta^{a_sb_s}\} 
     -  \phi_{a_1...a_s;b_1..b_s} \theta^{a_1b_1} ... \theta^{a_sb_s}\theta^{cd}\eth_d f
\end{align}
(note that there is no factor $s$ in the second term), 
and
\begin{align}
 -\{x^c, f\phi^{(s)}\}_- 
     &= -s \phi_{a_1...a_s;b_1..b_s} f\theta^{a_1b_1} ... \{x^c, \theta^{a_sb_s}\} 
     - sc_{s+1} \phi_{a_1...a_s;b_1..b_s} \theta^{a_1b_1} ... [\theta^{a_sb_s}\theta^{cd}]_0\eth_d f
\end{align}
using tracelessness, where \eq{c-s-formula}
\begin{align}
c_{2s+1}= \frac{3}{2s+1} \ .
\end{align}
Now consider first
\begin{align}
 -\{x_c,\{x^c, f\phi^{(s)}\}\}   
  &= -s \phi_{a_1...a_s;b_1..b_s} f\theta^{a_1b_1} ... \underbrace{\{x_c,\{x^c, \theta^{a_sb_s}\}\}}_{2 r^2 \theta^{a_sb_s}}
   - 2s \phi_{a_1...a_s;b_1..b_s} \theta^{a_1b_1} ... \{x_c, \theta^{a_sb_s}\}\theta^{cd}\eth_d f \nn\\
  &\quad  -s(s-1) \phi_{a_1...a_s;b_1..b_s} f \underbrace{\{x_c,\theta^{a_1b_1}\} ... \{x^c, \theta^{a_sb_s}\}}_0 
      - \phi_{a_1...a_s;b_1..b_s} \theta^{a_1b_1} ... \theta^{a_sb_s}\{x_c,\theta^{cd}\eth_d f\} \nn\\
  &=  -r^2 R^2  \phi^{(s)}  \eth^d\eth_d f -2 s r^2 f\phi^{(s)} 
  -2 s\phi_{a_1...a_s;b_1..b_s} \theta^{a_1b_1} ... \{x_c,\theta^{a_sb_s}\}\theta^{cd}\eth_d f \nn\\
  &=  -r^2 R^2  \phi^{(s)}  \eth^d\eth_d f -2 s r^2 f\phi^{(s)} 
  -2 s r^2 R^2 \phi_{a_1...a_s;b_1..b_s} \theta^{a_1b_1} ... (\eth^d\theta^{a_sb_s})\eth_d f \nn
\end{align}
using 
\begin{align}
 \{x_c,\theta^{cd}\eth_d f\} 
 &= \{x_c,\theta^{cd}\} \eth_d f +  \theta^{cd}\{x_c,\eth_d f\} 
  = r^2 R^2 \eth^d\eth_d f
\end{align}
since $x^d\eth_d=0$.
We observe that the last term is in $\cC^s$.
That formula could be obtained simply from \eq{BoxH-eth-formula}, but the intermediate steps are useful here.
The  first terms also arise in
\begin{align}
 -\{x_c,\{x^c, f\phi^{(s)}\}_-\}_+ 
  &= -s \phi_{a_1...a_s;b_1..b_s} f\theta^{a_1b_1} ... \underbrace{\{x_c,\{x^c, \theta^{a_sb_s}\}\}}_{2 r^2 \theta^{a_sb_s}}
   - s \phi_{a_1...a_s;b_1..b_s} \theta^{a_1b_1} ... \{x_c, \theta^{a_sb_s}\}\theta^{cd}\eth_d f  \nn\\
  &\quad 
     - s c_{s+1}\phi_{a_1...a_s;b_1..b_s} \theta^{a_1b_1} ... \{x_c,[\theta^{a_sb_s}\theta^{cd}]_0\eth_d f\}  \nn\\
  &\quad 
     - s(s-1) c_{s+1} \phi_{a_1...a_s;b_1..b_s} \{x_c,\theta^{a_1b_1}\} ... [\theta^{a_sb_s}\theta^{cd}]_0\eth_d f \nn\\
  &= -2sr^2 f\phi^{(s)}
   - s r^2R^2 \phi_{a_1...a_s;b_1..b_s} \theta^{a_1b_1} ... (\eth^d\theta^{a_sb_s})\eth_d f  \nn\\
  &\quad 
     - s c_{s+1} \phi_{a_1...a_s;b_1..b_s} \theta^{a_1b_1} ... \{x_c,[\theta^{a_sb_s}\theta^{cd}]_0\eth_d f\}  \nn\\
  &= -r^2 s  \phi_{a_1...a_s;b_1..b_s} \theta^{a_1b_1} ... \theta^{a_{s-1}b_{s-1}} \nn\\
 &\qquad \Big(2f\theta^{a_sb_s} + R^2 (\eth^d\theta^{a_sb_s})\eth_d f  
  + \frac 1{r^2} c_{s+1} \{x_c,[\theta^{a_sb_s}\theta^{cd}]_0\eth_d f\}\Big) \ .
  \label{x+xm-gen-1}
\end{align}
Here we observe that the term proportional to $s(s-1)$ vanishes since
\begin{align}
 \phi_{a_1...a_s;b_1..b_s} \{x^c,\theta^{a_1b_1}\} ... [\theta^{a_sb_s}\theta^{cd}]_0 
  &=  r^2\phi_{a_1...a_s;b_1..b_s} (\eta^{a_1c}x^{b_1} - \eta^{b_1c}x^{a_1}) ... [\theta^{a_sb_s}\theta^{cd}]_0 \nn\\
  &=  r^2\phi_{a_1...a_s;b_1..b_s} ([\theta^{a_sb_s}\theta^{a_1 d}]_0x^{b_1} 
  - [\theta^{a_sb_s}\theta^{b_1d}]_0x^{a_1}) ... = 0
\end{align}
as it involves a contraction or the irreducible tensors with $\eta_{ab}$ or $\varepsilon_{abcde}$ due to \eq{averaging-relns}.
Now we can reduce the term in brackets using the $s=1$ result  \eq{s=1-formula-x+x-}.
This gives  
{\small
\begin{align}
 -\{x_c,\{x^c, f\phi^{(s)}\}_-\}_+ 
   &= - \frac{r^2 s}{2s+1}  \phi_{a_1...a_s;b_1..b_s} \theta^{a_1b_1} ... \theta^{a_{s-1}b_{s-1}} 
  \Big( 2(2s+1)f\theta^{a_sb_s} + 2s R^2 (\eth^d\theta^{a_sb_s})\eth_d f  
  + R^2 \theta^{ab}\eth^d\eth_d f \Big) \nn\\
  &\stackrel{!}{=} \frac{s}{2s+1} (\Box_H -2r^2(s+1))(f\phi^{(s)}) \nn\\
   &=  \frac{r^2 s}{2s+1} \Big(-R^2  \phi^{(s)}  \eth^d\eth_d f -2(2s+1) f\phi^{(s)}
  -2 s R^2 \phi_{a_1...a_s;b_1..b_s} \theta^{a_1b_1} ... (\eth^d\theta^{a_sb_s})\eth_d f  \Big)  \nn
\end{align}
  }
using \eq{BoxH-ftheta}, which proves \eq{magic-formula-x+x-}.

\end{proof}


It is quite instructive to check the $s=1$ case explicitly for $\phi^{(1)} = x^p M^{ab}$:
\begin{align}
 \Box_H\phi^{(1)}
  &= -6r^2\phi^{(1)} + 2(\theta^{ap}x^b - \theta^{bp}x^a) \ .
\end{align}
Now
\begin{align}
 \{x^c,x^p M^{ab}\}_- &= - x^p  \{M^{ab},x^c\} + [\theta^{cp} M^{ab}]_0   \nn\\
  &= - x^p (\eta^{ac}x^b - \eta^{bc}x^a) 
  - \frac{R^2}{3} \big(P_\perp^{ca} P_\perp^{pb} - P_\perp^{cb} P_\perp^{pa} + \frac 1R \varepsilon^{capbe} x_e\big)  
\end{align}
hence 
\begin{align}
  \{x_c, \{x^c,x^p M^{ab}\}_- \}
   &= - \{x_a,x^p x^b\} + \{x_b,x^p x^a\}
  - \frac{R^2}{3}\big(\{x_a,P_\perp^{pb}\} - \{x_b,P_\perp^{pa}\} + \frac 1R \varepsilon^{capbe}  \theta_{ce}\big)     \nn\\
   &= - \frac{4}{3}\big(2x^p\theta^{ab} + x^b\theta^{ap} - x^a \theta^{bp}\big) -\frac 13 R \varepsilon^{abpce}  \theta_{ce}  \nn\\
   &= - \frac{4}{3}\big(2x^p\theta^{ab} + x^b\theta^{ap} - x^a \theta^{bp}\big)
     -\frac 23 (\theta^{ab}x^p + \theta^{bp}x^a + \theta^{pa}x^b)    \nn\\
   &=  \frac{10}{3} r^2\phi^{(1)} - \frac{2}{3} x^b\theta^{ap} + \frac{2}{3} x^a \theta^{bp} \nn\\
  &=  -\frac 13(\Box_H-4r^2)\phi^{(1)}  \nn
\end{align}
using the self-duality relations in  Lemma \ref{lem:selfdual-2}:

\begin{lem}
\label{lem:selfdual-2}
 $\theta^{ab}$ satisfies the following self-duality relations 
\begin{align}
 \varepsilon^{abpce} \theta_{ce} x^p &= 2R\theta^{ab}   \label{selfdual-1} \\
 \theta^{ab} &= \frac 1{2R} \varepsilon^{abcde}x_c\theta_{de}  \label{selfdual-2}  \\
 \varepsilon^{abpce}  \theta_{ce} &= \frac{2}{R}(\theta^{ab}x^p + \theta^{bp}x^a + \theta^{pa}x^b)  
 \label{selfdual-3}
\end{align}
\end{lem}
where the indices of $\theta_{ce} = \eta_{cc'}\eta_{ee'} \theta^{c'e'}$.
\begin{proof}
The first relation is already known \cite{Sperling:2018xrm}, and the second relation
reduces to the first at the reference point  $\xi=(R,0,0,0,0)$.
Now consider the third relation.
The rhs is  totally antisymmetric.  
At the reference point we can use $\theta^{0a} \sim \xi_b\theta^{ba} = 0$,
so that the lhs vanishes if all 3 indices $abp$ are tangential at $\xi$. If one is transversal, say $a=0$, 
this reduces to 
\begin{align}
  \varepsilon^{0bpce}  \theta_{ce} = \frac{2}{R}\theta^{bp}x^0 = 2\theta^{bp}
\end{align}
which is correct
using \eq{selfdual-1}.
As a check, contracting \eq{selfdual-3} with $\varepsilon_{abprs}$ gives 
\begin{align}
 \varepsilon_{abprs}\varepsilon^{abpce}  \theta_{ce} &= \frac{6}{R} \theta^{ab}x^p \varepsilon_{abprs} 
 = 12\, \theta_{rs} \ .
\end{align}

\end{proof}

As a corollary, we obtain
\begin{corollary}
\begin{align}
 -\{x^\mu,\{x_\mu, \phi^{(s,k)}\}_-\}_+ &= \Big(\a_s (\Box_H -2r^2(s+1))  + D^+D^-\Big)\phi^{(s,k)}  , \nn\\ 
       &= \a_s (\Box_H -2r^2(s+1)) \phi^{(s,0)} , \qquad k=0
 \label{magic-formula-xmu}
 \end{align}
 \label{cor-alpha-mu}
 \end{corollary}
\begin{proof}
 This  follows from the above noting that 
 $\{x^4,\{x_4, \phi\}_-\}_+  = D^+D^-\phi$ and  $D^-\phi^{(s,0)}=0$.
\end{proof}

\begin{corollary}
The totally symmetric tensor field 
$\phi_{\mu_1...\mu_s}(x)$
associated to $\phi^{(s,0)}$ via \eq{tensorfields-Am-formula} is square-integrable, space-like 
and divergence-free with positive inner product 
if $\phi^{(s,0)}$ is admissible. It is proportional to the tensor field in \eq{Cs-explicit}.
 \label{cor-pos}
 \end{corollary}
 For generic higher-spin modes such as  $\cA_\mu^{(-)}[\phi^{(s,k)}]$, positivity should not be expected, since they 
are in general not physical. 
\begin{proof}
To see this, we first note that $\{x_\mu,\phi^{(s,0)}\}_- \in \cC^{(s-1,0)}$, because
\begin{align}
 D^- \{x_\mu,\phi^{(s,0)}\}_- = 0 \ .
\end{align}
More generally,
\begin{align}
  \phi_{\mu_1...\mu_{l}} := \cA_{\mu_1}^{(-)}[... [\cA_{\mu_{l}}^{(-)}[\phi^{(s,0)}]...] \ \in \ \cC^{(s-l,0)}
 \end{align}
for any $l$,
so that 
\begin{align}
 \{x^\mu,\{x_\mu,\phi_{\mu_1...\mu_{l}} \}_-\}_+  = \{x^a,\{x_a,\phi_{\mu_1...\mu_{l}} \}_-\}_+
  = \a_s (\Box_H -2r^2(s-l+1)) \phi_{\mu_1...\mu_{l1}} \ .
\end{align}
Then 
in the computation of the inner product \eq{inner-tensorfield-BoxH} goes through with indices in $0,..,3$.
Space-like and divergence-free follows as in \eq{spacelike-A10}.
The relation with \eq{Cs-explicit} follows from irreducibility.

\end{proof}

%

\subsection{Algebraic relations for $D^\pm$, $\Box_H$ and $\cK$}
\label{sec:relations-box-H}

We can derive a  relation between $\Box$ and $\Box_H$ as follows:
consider 
\begin{align}
 -D^2\Box\phi &= \frac 2{R^2}\{x_\mu,\{x^\mu,\phi\}\} + \frac 2R(\{t_\mu,\{x^\mu,D\phi\}\} + \{x_\mu,\{t^\mu,D\phi\}\})
 + \{t_\mu,\{t^\mu,D^2\phi\}\} - 2 r^2 \Box\phi  \nn\\
 &= \frac 2{R^2}\{x_\mu,\{x^\mu,\phi\}\} + \frac 4R\{t_\mu,\{x^\mu,D\phi\}\} 
 - \frac{8}{R^2} D^2\phi - \Box (D^2\phi) - 2 r^2 \Box\phi  \nn\\
 &= -\frac 2{R^2} \Box_H\phi 
 + \Box D^2\phi - \frac 2{R^2} D^2\phi - 2D\Box D\phi  - 2 r^2 \Box\phi  
\end{align}
using $D^2 x^\mu = r^2 x^\mu$ and the identity
\begin{align}
 2R \{t^{\mu},\{x_\mu,\phi^{(s)}\}\} 
   = (R^2\Box+4) D\phi^{(s)} - R^2 D(\Box\phi^{(s)})
 \label{t-c-comm-id}
\end{align}
which is proved in (A.36) in \cite{Sperling:2019xar}.
Hence 
\begin{align}
\boxed{
 \ 2\Box_H = R^2(D^2\Box  + \Box D^2 - 2 D\Box D - 2 r^2 \Box) -2 D^2 \ . \
} 
\label{BoxH-Box-relation}
\end{align}
%
Writing $D = D^+ + D^-$
this can be written using \eq{Box-D-relation} as  
\begin{align}
 \boxed{
 \Box_H \phi^{(s)} = \Big(-R^2 r^2\Box + (2s-1) D^+D^- - (2s+3) D^- D^+\Big)\phi^{(s)} \ .
 \ }
 \label{BoxH-Box-relation-2}
\end{align}
This is a very useful relation, which can be checked easily e.g. for $\phi = x^a$.
It will allow to evaluate  the inner products of the fluctuation modes. 
It also allows to express $D^- D^+ $ in terms of $D^+D^-$
and the Box operators.
Since $D^- D^+$ commutes with both $\Box$ and  $\Box_H$,
we obtain
\begin{align}
 [D^- D^+,D^+ D^-] = 0 \ .
 \label{D-D+-comm}
\end{align}
In particular, this gives 
\begin{align}
  \Box_H \phi^{(s,0)} &= \Big(-R^2 r^2\Box - (2s+3) D^- D^+\Big)\phi^{(s,0)} \ .
\end{align}
Now consider
\begin{align}
 \Box_H D^+ \phi^{(s,0)} &= D^+ (\Box_H + 2r^2(s+1))\phi^{(s,0)} \nn\\
  &= D^+ \Big(-R^2 r^2\Box  - (2s+3) D^- D^++ 2r^2(s+1) \Big)\phi^{(s)}  \nn\\
  &=  \Big(-(2s+3) D^+ D^- - R^2 r^2  \Box +  4r^2(s+1) \Big)D^+\phi^{(s,0)} \ .
  \label{BoxH-D+}
\end{align}
Combining this with \eq{BoxH-Box-relation-2} for $D^+ \phi^{(s,0)}$ 
gives 
\begin{align}
  \Big(-(2s+3) D^+ D^- +  4r^2(s+1) \Big)D^+\phi^{(s,0)} 
   &= \Big((2s+1) D^+D^- - (2s+5) D^- D^+\Big)D^+ \phi^{(s,0)} 
\end{align}
hence
\begin{align}
 D^+ D^- (D^+ \phi^{(s,0)}) = \Big( \frac{2s+5}{4(s+1)} D^- D^+ + r^2 \Big)D^+ \phi^{(s,0)} \ .
  \label{D+D--CR-1}
\end{align}
These are effectively  commutation relations between $D^+$ and $D^-$ on $\cC^{(s+1,1)}$.
For the general case, we make the ansatz
\begin{align}
\boxed{
\ \  D^+ D^- \big((D^+)^k \phi^{(s,0)}\big) 
 = \big(\tilde a_k D^- D^+ + \tilde b_k \big)  \big((D^+)^k \phi^{(s,0)}\big) \ . \ 
 }
  \label{D+D--CR-general-1}
\end{align}
The constants are determined recursively by considering 
\begin{align}
 \Box_H(D^+)^k \phi^{(s,0)} &= D^+ \big(\Box_H + 2r^2(s+k) \big)(D^+)^{k-1} \phi^{(s,0)}   \nn\\
  &= D^+ \big(-R^2 r^2\Box + (2s+2k-3) D^+D^- - (2s+2k+1) D^- D^+ + 2r^2(s+k) \big)(D^+)^{k-1} \phi^{(s,0)}   \nn\\
  &= \Big(-R^2 r^2\Box + \big((2(s+k)-3)\tilde a_{k-1} - 2(s+k) -1 \big) D^+ D^-  \nn\\
   &\qquad    + (2(s+k)-3) \tilde b_{k-1} + 4r^2(s+k) \Big)(D^+)^{k} \phi^{(s,0)}  \ . 
\end{align}
On the other hand, the lhs can be written  using \eq{BoxH-Box-relation-2} as 
\begin{align}
  \Box_H(D^+)^k \phi^{(s,0)} &=  \Big(-R^2 r^2\Box + (2(s+k)-1) D^+D^- - (2(s+k)+3) D^- D^+\Big)(D^+)^k \phi^{(s,0)}
\end{align}
and combining these we obtain 
\begin{align}
 &\big((2(s+k)-3)\tilde a_{k-1} - 4(s+k) \big) D^+ D^-  + (2s+2k-3) \tilde b_{k-1} + 2r^2(2s+2k)
  \nn\\ &\qquad  =  - (2(s+k)+3) D^- D^+  \ 
\end{align}
acting on $(D^+)^{k} \phi^{(s,0)}$.
Comparing with \eq{D+D--CR-general-1}, we obtain  two recursion relations
\begin{align}
 \tilde a_k &= -\frac{2(s+k)+3}{(2(s+k)-3)\tilde a_{k-1} - 4(s+k) }  \nn\\
 \tilde b_k &= -\frac{(2(s+k)-3) \tilde b_{k-1} + 4r^2(s+k) }{(2(s+k)-3)\tilde a_{k-1} - 4(s+k) }
 \label{ak-bk-recursion}
\end{align}
with 
\begin{align}
 \tilde a_0 &= 0 = \tilde b_0  \ .
\end{align}
This  is solved by the remarkably simple general formula\footnote{A random 
change of the  recursion would lead to a complete mess here, which
strongly indicates that we are on the right track.}
\begin{align}
 \boxed{ \
\begin{aligned}
 \tilde b_{k} &= k r^2\frac{2s + k}{2s + 2k-1}    \nn\\
 \tilde a_{k} &=  \frac {k}{k+1} \ \frac{2s + 2k + 3}{2s + 2k-1} \ \frac{2s + k}{2s + k+1}  \ . \
 \end{aligned}
 \ }
\end{align}
Now we change notation as follows:
\begin{align}
\boxed{
\ \  D^+ D^- \phi^{(s,k)} = \big(a_{s,k} D^- D^+ + b_{s,k} \big) \phi^{(s,k)},
\qquad  \phi^{(s,k)} = (D^+)^k \phi^{(s-k,0)} \ . \ 
 }
  \label{D+D--CR-general}
\end{align}
Comparing with the above we see that 
 $a_{s,k} = \tilde a_k|_{s\to s-k}$ and $a_{s,k} = \tilde a_k|_{s\to s-k}$,
and therefore
\begin{align}
 \boxed{ \
\begin{aligned}
 b_{s,k} &=  r^2 k\, \frac{2s - k}{2s -1}    \nn\\
 a_{s,k} &=  \frac {k}{k+1} \ \frac{2s + 3}{2s-1} \ \frac{2s - k}{2s - k+1}  \ . \
 \end{aligned}
 \ }
\end{align}
We also note the inverse relation
\begin{align}
 D^- D^+ &= \frac{1}{a_{s,k}} D^+ D^-  - \frac{b_{s,k}}{a_{s,k}}   \nn\\
  &= \frac {k+1}{k} \ \frac{2s-1}{2s + 3} \ \frac{2s - k+1}{2s - k} D^+ D^-
    - r^2 (k+1) \frac{2s - k+1}{2s + 3}\ , \qquad k\geq1
     \label{D-D+-CR-general}
\end{align}
which however only makes sense for  $k\geq1$.

\paragraph{Relations for $\cK$.}

As a consequence, we obtain
\begin{align}
  D^+D^- \phi^{(s,k)}  &=  D^+ (D^- D^+) \phi^{(s-1,k-1)} \nn\\
  &= D^+   \Big(  \frac {k}{k-1} \ \frac{2s - k}{2s - k-1} \frac{2s-3}{2s + 1} \  D^+D^- 
   - r^2 k \frac{2s - k}{2s + 1}\Big) \phi^{(s-1,k-1)}  
\end{align}
for $k\geq 2$,
which gives 
\begin{align}
  \frac {(2s+1)(2s-1)}{k(2s-k)} D^+D^- \phi^{(s,k)}  
  &=  D^+   \Big( \frac {(2s-1)(2s-3)}{(k-1)(2s - k-1)} \  D^+D^- 
   - r^2(2s-1)\Big) \phi^{(s-1,k-1)} \ .  \nn
\end{align}
Comparing with the definition \eq{chi-def} of $\cK$
\begin{align}
 -r^2\cK &= r^2 s^2 + \frac{4s^2-1}{k (2s-k)}D^+D^- \ 
 = r^2(s+1)^2 + \frac{(2s+1)(2s+3)}{(k+1) (2s-k+1)}D^-D^+ \ 
 \label{cK-def-op}
\end{align}
we obtain 
\begin{align}
  [\cK, D^+] = 0 \ \quad \mbox{on} \ \cC^{(s,k)}, \ k\geq 1 \ .
 \label{chi-comm-0}
\end{align}
For $k=1$, we can write 
\begin{align}
 \cK\phi^{(s,1)}  &=  \cK D^+ \phi^{(s-1,0)}
  = - \big((2s+1)D^+D^- + s^2\big) D^+ \phi^{(s-1,0)} \nn\\
  &= - D^+\big((2s+1) D^-D^+  + s^2 \big) \phi^{(s-1,0)} 
  = D^+ \cK \phi^{(s-1,0)}
\end{align}
using the second form in \eq{chi-def} of $\cK$. It follows that 
\begin{align}
 \boxed{ \ [\cK, D^\pm] = 0 \ }
 \label{chi-comm}
\end{align}
without any restrictions.
In particular, diagonalizing the space-like Laplacian $D^+D^-$ on 
$\phi^{(s,k)} =(D^+)^k \phi^{(s-k,0)}$ is equivalent to diagonalizing it on $\phi^{(s-k,0)}$.

\paragraph{Evaluation of $\Box_H$ and positivity.}

We can use the above results to show 
\begin{align}
 \Box_H  \phi^{(s,k)} 
  &= r^2\Big(-R^2 \Box +\cK + (s+1)^2  +  k (2s - k) \Big)\phi^{(s,k)} 
   \label{BoxH-explicit-CR}
\end{align}
which is obtained from \eq{BoxH-Box-relation-2} using the relations \eq{D+D--CR-general}.
This provides an on-shell relation between 
the Laplacians on $H^3$ and $H^4$.
Moreover, we recall that $\Box_H$ is manifestly positive, and satisfies the bound 
$\Box_H > r^2 (s^2+s+2)$ using the  admissibility condition \eq{admissible}.
Then \eq{BoxH-explicit-CR} gives 
\begin{align}
 \Big(-R^2 \Box +\cK  +s-1 +  k (2s - k) \Big)\phi^{(s,k)} & >  0 \ .
\end{align}
This is useful to establish the signature $(+++-)$ of off-shell modes  in section \ref{sec:inner}.

\subsection{Positivity of the space-like Laplacian $\cK$}
\label{sec:positivity-chi}

Now we show  lemma \ref{chi-pos-lemma}, which states that  $\cK > 0$ for admissible $\phi$.
To get some insight,  recall  from \cite{Sperling:2019xar} that for scalar fields $\phi\in\cC^0$, 
\begin{align}
 -D^- D^+ \phi  &= \frac{r^2 R^2}3 \cosh^2(\eta)  \Delta^{(3)}\phi \ .
  \label{D-D+-explicit}
\end{align}
Here $\Delta^{(3)}=-\nabla^{(3)\a} \nabla^{(3)}_\a$ is the space-like Laplacian on $H^3$ w.r.t. the induced metric,
extended  to symmetric tensor fields $\phi_{\mu_1...\mu_s}(x)$.
The lhs is related to $\cK$ \eq{cK-def-op} by a factor and a shift.
Clearly $\Delta^{(3)}>0$ for square-integrable functions, but the 
required bound $\cK>0$ is slightly stronger.

\begin{proof}
 Using \eq{chi-comm}, it suffices to show $\cK>0$ for $k=0$, which is the statement
\begin{align}
 - D^-D^+ \ > \ r^2\frac{(s+1)^2}{2s+3}  \qquad \mbox{on}  \ \ \cC^{(s,0)} \ .
 \label{DpDm-bound}
\end{align}
%
Recall that the tenssor field encoded in $\phi^{(s)} = \phi_{\mu_1...\mu_s} t^{\mu_1} ... t^{\mu_s} \in \cC^{(s,0)}$ 
is divergence-free. Then
\begin{align}
  D^- D^+ \phi^{(s)} 
   &= r^2 R D^-\left(\nabla^{(3)}_\a \phi_{\mu_1...\mu_s} 
    t^{\mu_1} ... t^{\mu_s} t^\a\right)  \nn\\
 &= r^4 R^2 \left[\nabla^{(3)}_\b \ \nabla^{(3)}_\a \phi_{\mu_1...\mu_s} 
    t^{\mu_1} ... t^{\mu_s} t^\a t^\b\right]_s  \nn\\
   &= \frac{c_{s+2}}3  r^2 R^2 \cosh^2(\eta) \left(\nabla^{(3)\a} \nabla^{(3)}_\a 
 \phi_{\mu_1...\mu_s}  t^{\mu_1} ... t^{\mu_s}
   + s\nabla^{(3)\mu_1} \nabla^{(3)}_\a \phi_{\mu_1...\mu_s}   t^{\mu_2} ... t^{\mu_s}t^\a\right)  \nn\\
%
%
   &=\frac{r^2}{2s+3}  \left(-R^2 \cosh^2(\eta) \Delta^{(3)} \phi^{(s)} - s(s+1)\phi^{(s)}\right)
    \label{D+D--phi1}
\end{align}
using Lemma \ref{lemma-wick-2},  
noting that $\nabla^{(3)} P_\perp^{\mu\nu} = 0$ and 
\begin{align}
 [\nabla^{(3)}_{\mu_1}, \nabla^{(3)}_\a]\phi^{\mu_1...\mu_s}
 &= (R^{(3)}_{\mu_1\a})^{\mu_1}_{\ \nu}\phi^{\nu\mu_2...\mu_s} 
   + \sum_{j\neq 1} (R^{(3)}_{\mu_1\a})^{\mu_j}_{\ \nu}\phi^{\mu_1\nu...\mu_s'} \nn\\
 &= -\frac 1{R^2\cosh^2(\eta)}\Big(2P^\perp_{\a\nu}\phi^{\nu\mu_2...\mu_s} 
   + \sum_{j\neq 1}\big(\phi^{\mu_j\a...\mu_s'} 
             - P^\perp_{\mu_1\nu} \phi^{\mu_1\nu..\a..\mu_s'}\big)\Big) \nn\\
 &= -\frac 1{R^2\cosh^2(\eta)}\Big(2\phi^{\a\mu_2...\mu_s}  + (s-1)\phi^{\a\mu_2...\mu_s'} \Big)\nn\\
 &= -\frac 1{R^2\cosh^2(\eta)} (s+1)\phi^{\a\mu_2...\mu_s}  
\end{align}
using space-like gauge and tracelessness of $\phi^{\mu_1...\mu_s}$. Here 
\begin{align}
 R^{(3)}_{\mu\nu;\a\b} &= -\frac{1}{\r^2}\big(P^\perp_{\mu\a}P^\perp_{\nu\b} - P^\perp_{\mu\b}P^\perp_{\nu\a}\big),
 \qquad R^{(3)}_{\mu\a} = -\frac{2}{\r^2}P^\perp_{\mu\a}  
\end{align}
are the  Riemann and Ricci tensors on $H^3$ with radius $\rho = R\cosh(\eta)$, and
$P^\perp_{\mu\nu}$ is the tangential projector \eq{H3-projector} on $H^3$.
Now \eq{DpDm-bound} follows using results of Delay (remark 6.2 in \cite{delay2002essential}) and Lee 
(Proposition E in \cite{lee2006fredholm}), 
which essentially state that the spectrum\footnote{In \cite{lee2006fredholm}, the result is established only for the essential spectrum, but we 
assume that this is not a significant restriction. I am grateful for useful communications with Erwann Delay and 
Wilhelm Schlag.} 
of $\r^2\Delta^{(3)}$ on rank $s$ symmetric square-integrable\footnote{Since
$\phi_{\mu_1...\mu_s}$ is square-integrable on $H^4$ being a principal series irrep of $SO(4,1)$ (see also corollary \ref{cor-pos}),
it is also square-integrable on almost all $H^3$ by Fubini's theorem, and for sufficiently smooth
wavefunctions this should hold for all $H^3$.
However, there should be a better way to justify this.} 
tensor fields  on $H^3$ 
is given by $[s+1,\infty)$, i.e. $\r^2\Delta^{(3)}|_{\phi_{\mu_1...\mu_s}} > s+1$.


\end{proof}


\subsection{Proof of $v_{\rm null} = 0$}
\label{sec:v-null-vanish}

In this section we prove that the null vector \eq{v-null-vanish} which arises in the scalar sector actually vanishes,
\begin{align}
  v_{\rm null} = 0 \ .
  \label{vnull-vanish}
\end{align}
To see this, we need some identities.
\begin{prop}
The following identities hold
\begin{align}
   (2s+1) D^+ D^-\cA^{(+)}[\phi^{(s)}] 
  &= \cA^{(+)}\big[\big((2s-1)D^+D^- - (2s+3)D^-D^+  +r^2(2s+1)\big)\phi^{(s)}\big] \nn\\
 &\quad + (2s+5) D^- \cA^{(+)}[D^+\phi^{(s)}]   -2r^2 R \cA^{(g)}[D^+\phi^{(s)}] \ ,
 \label{D+D-A+-generic}  \\
 (2s+1)D^+ \cA^{(n)}[\phi^{(s)}]  
  &= 2\cA^{(+)}[\big( D^+D^- - D^-D^+ \big)\phi^{(s)}]
  - (2s+5) \cA^{(-)}[D^+D^+\phi^{(s)}]  \nn\\
 &\quad  + 2(2s+3)\cA^{(n)}[D^+\phi^{(s)}] 
  - 2r^2 R \cA^{(g)}[D^+\phi^{(s)}] \ .
  \label{D+An}
 \end{align}
\end{prop}

\begin{proof}
We can rewrite the lhs  of the first relation using \eq{BoxH-Box-relation-2} 
and use the intertwiner properties \eq{Box-A2m-eigenvalues} and \eq{BoxH-A-relation} to get
\begin{align}
  (2s+1)D^+ D^-\cA^{(+)}[\phi^{(s)}] 
  &= \Big(R^2 r^2\Box + \Box_H + (2s+5) D^- D^+\Big)\cA^{(+)}[\phi^{(s)}]  \nn\\
 &= \cA^{(+)}[(R^2 r^2\Box + \Box_H +r^2(2s+1))\phi^{(s)}]  \nn\\
 &\quad + (2s+5) D^- \cA^{(+)}[D^+\phi^{(s)}]  -2r^2 R \cA^{(g)}[D^+\phi^{(s)}] \nn\\
 &= \cA^{(+)}[\big((2s-1)D^+D^- - (2s+3)D^-D^+  +r^2(2s+1)\big)\phi^{(s)}]  \nn\\
 &\quad + (2s+5) D^- \cA^{(+)}[D^+\phi^{(s)}]  -2r^2 R \cA^{(g)}[D^+\phi^{(s)}] \ .
\end{align}
Then \eq{D+An} is obtained using  \eq{DApm-relation} twice.

\end{proof}

Now we can prove \eq{vnull-vanish}.
Consider first
\paragraph{$s=1$ Case.}
We start with the easy observation 
\begin{align}
  D^-\cA_\mu^{(+)}[\phi]  &= r^2 R\{t^\mu,\phi\} +\cA_\mu^{(-)}[D\phi] \ 
  \label{D+A+-s0-rel}
\end{align}
for $\phi\in\cC^0$.
Acting with $D^+$, this gives
\begin{align}
 D^+ D^-\cA_\mu^{(+)}[\phi] 
   &= r^2 \cA_\mu^{(+)}[\phi] + r^2 R \cA_\mu^{(g)}[D\phi]
   + \cA_\mu^{(n)}[D\phi] \ ,
\end{align}
and using  \eq{D+D-A+-generic} for the lhs leads to 
\begin{align}
  -3\cA^{(+)}[D^-D\phi] + 5 D^- \cA^{(+)}[D\phi]  
 &= 3 r^2 R \cA_\mu^{(g)}[D\phi] + \cA_\mu^{(n)}[D\phi]  \ .
\end{align}
Writing $D\phi = \phi^{(1,1)}$ and replacing $D^- \cA^{(+)}$ using \eq{DApm-relation}, we obtain
\begin{align}
 v_{\rm null}^{(s=1)} \equiv
  2 r^2 R \cA^{(g)}[\phi^{(1,1)}] + 5\cA^{(-)}[D^+\phi^{(1,1)}] 
  + 2 \cA^{(+)}[D^-\phi^{(1,1)}]  - 6\cA_\mu^{(n)}[\phi^{(1,1)}] &= 0 \ .
  \label{vnull-s1-vanish}
\end{align}
This is precisely the null vector in \eq{v-null-vanish} for $s=1$, 
which is thus shown to vanish identically.

\paragraph{Generic $s$.}

Acting with $D^+$ on $v_{\rm null}^{(s)}$ \eq{v-null-vanish} and assuming inductively that it vanishes,
we obtain with \eq{D+D-A+-generic} and \eq{D+An} after some straightforward calculations
\begin{align}
 0 &= s D^+ v_{\rm null}^{(s)} \nn\\
  &= \frac 1{s} \cA^{(+)}[D^+D^-\phi^{(s,s)}] 
  + \frac{s(2s+3)}{1 + s} \cA^{(n)}[D^+\phi^{(s,s)}]  
 - (2s+1) D^+\cA^{(n)}[\phi^{(s,s)}] 
  + s r^2 RD^+\cA^{(g)}[\phi^{(s,s)}] \nn\\
 &=  \cA^{(+)}\big[\big(\frac {1-2s}{s} D^+D^-   + 2D^-D^+ + r^2 s  \big)\phi^{(s,s)}\big]   
   + (2s+5) \cA^{(-)}[D^+D^+\phi^{(s,s)}]    \nn\\
 &\quad - (2s+3)\frac{s + 2}{1 + s} \cA^{(n)}[D^+\phi^{(s,s)}]  + (s+2) r^2 R\cA^{(g)}[D^+\phi^{(s,s)}] \ . 
\end{align}
Now we can use the commutation relations \eq{D+D--CR-general} in the form
\begin{align}
  \frac{2s -1}{s} D^+ D^- \phi^{(s,s)} = \Big( \frac {s(2s + 3) }{(s+1)^2}  D^- D^+ + r^2 s \Big) \phi^{(s,s)}
\end{align}
and $\phi^{(s+1,s+1)} = D^+ \phi^{(s,s)}$, which leads to 
\begin{align}
 0 
  &= (s+2) v_{\rm null}^{(s+1)} \ .
\end{align}

\subsection{Exceptional scalar modes}

It was shown in section \ref{sec:k=s-inner}
that the regular scalar modes $\tilde\cA_\mu^{(i)}[\phi^{(s,s)}]$ for $s\geq 1$
span only a 3-dimensional space, which implies that there is one missing scalar mode 
for each $s\geq 1$. 
Here we show how to determine this missing mode for $s=1$. The remaining
modes then arise as in \eq{exceptional-modes-properties}.

\paragraph{A relation for $s=0$.}

As a preparation, recall that $\cA_\mu^{(-)}[\phi^{(1)}]$ and $\cA_\mu^{(g)}[\phi]$
are complete in $C^0\otimes\R^4$. Therefore there must be a relation
\begin{align}
 0 &= x_\mu \phi  
 + \tilde\cA_\mu^{(g)}[\tilde\phi] +  \cA_\mu^{(-)}[D\phi'] \ .
 \label{relation-s0}
\end{align}
By acting with $\{t^\mu,.\}$ and $\{x^\mu,.\}$,
this implies
\begin{align}
 0 &=  \{t^\mu,x_\mu \phi\}  
 + \{t^\mu,\tilde\cA_\mu^{(g)}[\tilde\phi]\} + \{t^\mu,\cA_\mu^{(-)}[D\phi']\}  \nn\\
 &= \sinh(\eta) (\t + 4)\phi 
 - \Box\tilde\phi + \frac 1R D^-D\phi' \ .
 \label{eq-except0-2}
\end{align}
Similarly,
\begin{align}
 0 &=  \{x^\mu,x_\mu \phi\} 
 + \{x^\mu,\tilde\cA_\mu^{(g)}[\tilde\phi]\} + \{x^\mu,\tilde\cA_\mu^{(-)}[D\phi']\} \nn\\
 &= - R \sinh(\eta) D\phi 
 + \{\{x^\mu,t_\mu\},\tilde\phi\} + \{t_\mu,\{x^\mu,\tilde\phi\}\} 
 - \big(\a_1 (\Box_H -4 r^2) + D^+ D^-\big) D\phi' \nn\\
 &= D\Big(- R \sinh(\eta) \phi 
 + \frac{2}{R} \tilde\phi 
 - \big(\a_1 (\Box_H - 2 r^2) + D^- D \big)\phi'\Big)
\end{align}
implies
\begin{align}
 \frac{2}{R}\tilde\phi =  R \sinh(\eta) \phi + \big(\a_1 (\Box_H - 2 r^2) + D^- D \big)\phi' \ .
 \label{eq-except0-2}
\end{align}
These two equations can be solved  for $\phi'$ and $\tilde\phi$,
for given (admissible, generic) $\phi$.
 It follows that
\begin{align}
 0 &= D(x_\mu \phi)  
 + D\tilde\cA_\mu^{(g)}[\tilde\phi] +  D\cA_\mu^{(-)}[D\phi'] \nn\\
  &= r^2 R t_\mu \phi + x_\mu D\phi 
 + \tilde\cA_\mu^{(g)}[D\tilde\phi + r^2 RD\phi']
 + \frac 1R\tilde\cA_\mu^{(+)}[\tilde\phi]
 + \cA_\mu^{(-)}[D^+D\phi']  .
 \label{xDphi-tphi-relation}
\end{align}

%
%
%

\paragraph{Exceptional mode for $s=1$.}
\label{sec:except-modes}

Similar to \eq{relation-s0}, we make the ansatz\footnote{Using \eq{xDphi-tphi-relation}, this is equivalent to the ansatz $\cA_\mu^{(ex,1)} = x_\mu D\phi + \sum \tilde\cA^{(i)}$.}
\begin{align}
 \cA_\mu^{(ex,1)}[\phi]
 =  t_\mu D^+\phi 
 +   \tilde\cA_\mu^{(g)}[D\tilde\phi] + \tilde\cA_\mu^{(+)}[D\phi_+] 
 + \tilde\cA_\mu^{(-)}[D\phi_-]
\end{align}
for $\phi, \tilde\phi,\phi'_\pm \in \cC^0$.
 By definition,
these  are orthogonal to all $\cA_\mu^{(i)}$ modes, which amounts to 
\begin{align}
 0 = D^- \cA_\mu^{(ex,1)} = \{t^\mu,\cA_\mu^{(ex,1)}\} = \{x^\mu,\cA_\mu^{(ex,1)}\} \ .
\end{align}
These 3 equations can  be solved for $\tilde\phi,\phi'_\pm$,
for any given (admissible, generic) $\phi$. 
The same constraints   follow for all $\cA_\mu^{(ex,s)}$ by acting with $D^+$.
In particular, the $\cA_\mu^{(ex,s)}$ are physical.
We refrain from studying these in detail, since no nice general formula
was found.

\subsection{Inner products of $\cA$ modes.}
\label{sec:inner-prod-calc}

Here we derive the explicit formulas for the inner products \eq{Gij-matrix} for all $\tilde\cA^{(i)}[\phi]$ modes
for  $\phi = \phi^{(s,k)}$ and $\phi' = \phi'^{(s,k)}$.
First, it is clear (by invariance) that the modes are orthogonal unless the spin quantum numbers $s=s'$ and $k=k'$ 
coincide. Assuming this, we obtain
\begin{subequations}
\label{eq:inner_products}
\begin{align}
\int \cA_\mu^{(g)}[\phi'^{}] \cA^{(g)\mu}[\phi^{}] &=  \int \phi'^{} \Box \phi^{}  
\label{puregauge-inner}\\
 \int \cA_\mu^{(g)}[\phi'] \cA^{(+)\mu}[D^-\phi] 
 &= - \frac{s+2}R\int  \phi' D^+ D^-\phi
\\
\int \cA_\mu^{(g)}[\phi']  \cA^{(-)\mu}[D^+\phi] 
 &=   \frac{s-1}R \int\,  \phi' D^-D^+\phi  \quad 
 \\ 
\int \cA_\mu^{(-)}[D^+\phi'^{}] \cA^{(+)\mu}[D^-\phi^{}] 
&= -\int\, D^-D^+\phi'\,  D^+D^-\phi \
 = -\int\, \phi'\, D^-D^+ D^+D^-\phi \\
\int \cA_\mu^{(+)}[D^-\phi'] \cA^{(+)\mu}[D^-\phi^{(s)}]  
 &=   \int\,D^-\phi'\big((1-\a_{s-1})\Box_H + 2\a_{s-1} r^2 s + D^-D^+  \big) D^-\phi  \nn\\
  &=  - \int\,\phi'\big(\frac{s}{2s-1}(\Box_H - 2r^2) + D^+D^- \big)D^+ D^- \phi  \label{A+A+-inner}\\
%
%
\int \cA_\mu^{(-)}[D^+\phi'] \cA^{(-)\mu}[D^+\phi]  
 &=   \int\, D^+\phi'\big(\a_{s+1}(\Box_H - 2 r^2(s+2)) + D^+ D^-\big) D^+\phi  \\
 &= -  \int\, \phi'\big(\a_{s+1}(\Box_H - 2 r^2) + D^-D^+ \big)D^-D^+ \phi 
\end{align}
\end{subequations}
using \eq{BoxH-D-relation} in the last two relations, and 
$\a_s = \frac{s}{2s+1}$ \eq{magic-formula-x+x-}.
The inner products with the  $\cA^{(n)}$ modes is obtained as  follows:
\begin{align}
 \int \cA_\mu^{(n)}[\phi'^{}] \cA^{(-)\mu}[D^+\phi^{}]  
&= \int D^+\cA_\mu^{(-)}[\phi'^{}] \cA^{(-)\mu}[D^+\phi^{}]  
= - \int \cA_\mu^{(-)}[\phi'^{}] D^-\cA^{(-)\mu}[D^+\phi^{}]\nn\\
&= - \int \cA_\mu^{(-)}[\phi'^{}] \cA^{(-)\mu}[D^-D^+\phi^{}]  \nn\\
&=  -\int\, \phi'\big(\a_s(\Box_H - 2r^2(s+1)) + D^+D^- \big)D^-D^+\phi \ .
\end{align}
Here we used
\begin{align}
 D^-\cA^{(-)}[\phi^{}] &= \cA^{(-)}[D^-\phi^{}] \ ,  \qquad
 D^+\cA^{(+)}[\phi^{}] = \cA^{(+)}[D^+\phi^{}] \ .
\end{align}
Next,
\begin{align}
\int \cA_\mu^{(n)}[\phi'] \cA^{(+)\mu}[D^-\phi]  
  &= \int D^+\cA_\mu^{(-)}[\phi'^{}] \cA^{(+)\mu}[D^-\phi^{}]    \nn\\
  &= \int \Big(- D^-\cA_\mu^{(+)}[\phi'] + r^2 R\{t_\mu,\phi'\} + \cA_\mu^{(-)}[D^+\phi'] +  \cA_\mu^{(+)}[D^-\phi'] \Big) \cA^{(+)\mu}[D^-\phi^{}]      \nn\\
 &= \int \cA_\mu^{(+)}[\phi']\cA^{(+)\mu}[D^+D^-\phi^{}] \nn\\
  &\qquad
  + \int\Big(r^2 R\cA_\mu^{(g)}[\phi'] + \cA_\mu^{(-)}[D^+\phi'] 
    +  \cA_\mu^{(+)}[D^-\phi'] \Big) \cA^{(+)\mu}[D^-\phi^{}]  \nn\\
 &= \int \phi'\Big(\frac{1}{1-4s^2}\Box_H 
 + r^2 \,\frac{-2 s^2+s+2}{4 s^2-1} - D^+D^- \Big)D^+ D^- \phi 
 \end{align}
  using \eq{DApm-relation} along with the previous inner products, 
and 
\begin{align}
 \int \cA_\mu^{(+)}[\phi']\cA^{(+)\mu}[\phi^{(s)}]
  &=  \int\,\phi'\big(\frac{s+1}{2s+1}\Box_H + 2r^2 \frac{s(s+1)}{2s+1}  + D^-D^+  \big)\phi^{(s)} \ .
\end{align}
Next,
 \begin{align} 
 \int \cA_\mu^{(n)}[\phi'^{}] \cA^{(g)\mu}[\phi^{}] 
  &= - \int  \cA_\mu^{(-)}[\phi'^{}]D^-\{t^\mu,\phi\} \nn\\
  &= - \int \frac 1R \cA_\mu^{(-)}[\phi']\cA_\mu^{(-)}[\phi] 
   + \cA_\mu^{(-)}[\phi']\cA^{(g)\mu}[D^-\phi^{}] \nn\\
  &= \frac 1R \int \phi'\Big( \a_s( -\Box_H + r^2 2(s+1))  + (s-3) D^+ D^-\Big) \phi 
    \label{An-Ag-inner}
\end{align}
and finally
\begin{align}
 \int \cA_\mu^{(n)}[\phi'^{}]  \cA^{(n)\mu}[\phi^{}]  
  &= \int D^+\cA_\mu^{(-)}[\phi'^{}] D^+\cA_\mu^{(-)}[\phi^{}] \
  = -\int \cA_\mu^{(-)}[\phi'^{}] D^-D^+\cA_\mu^{(-)}[\phi^{}]\nn\\
  &= -\frac{1}{2s+1}\int \cA_\mu^{(-)}[\phi'^{}]\Big((2s-3) D^+D^- - \Box_H - R^2 r^2\Box\Big) \cA_\mu^{(-)}[\phi^{}] \nn\\
  &= \frac{2s-3}{2s+1}\int D^-\cA_\mu^{(-)}[\phi'^{}] D^-  \cA_\mu^{(-)}[\phi^{}] 
    +\frac{1}{2s+1}\int \cA_\mu^{(-)}[\phi'^{}](\Box_H + R^2 r^2\Box) \cA_\mu^{(-)}[\phi^{}] \nn\\
  &= \frac{2s-3}{2s+1}\int \cA_\mu^{(-)}[D^-\phi'^{}] \cA_\mu^{(-)}[D^-\phi^{}] \nn\\
  &\quad  +\frac{1}{2s+1}\int \cA_\mu^{(-)}[\phi'^{}]
  \Big(\cA_\mu^{(-)}[(\Box_H + r^2 R^2\Box - r^2 (2s+1))\phi^{}] 
   -  2Rr^2\cA^{(g)\mu}[D^-\phi^{(s)}]\Big)  \nn\\
  &=  \frac{2s-3}{2s+1}\int \phi' 
    \big(-\frac{s-1}{2s-1}(\Box_H  - 4r^2 s) D^+D^-
    -D^+D^+D^- D^-\big) \phi \nn\\
  &\quad  +\frac{1}{2s+1}\int \phi'
  \big(\a_s(\Box_H - 2r^2 (s+1)) + D^+D^-\big)(\Box_H + r^2 R^2\Box - r^2 (2s+1))\phi \nn\\
  &\quad  + r^2\frac{2(s -2)}{2s+1}\int  \phi' D^+ D^-\phi \ 
  \label{inner-n}
\end{align}
 using  \eq{BoxH-Box-relation-2},
\eq{Box-A2m-eigenvalues} ff, 
and the previous relations with \eq{A2-gaugefix} in the last step.
To rewrite $D^+D^+D^- D^-$ we need to specify $\phi = \phi^{(s,k)}$. 
Then we obtain using \eq{D+D--CR-general} 
\begin{align}
 \int \cA_\mu^{(n)}[\phi'^{}]  \cA^{(n)\mu}[\phi^{}]  
   &= -\frac{2s-3}{2s+1}\int \phi' 
    \Big(\big(\frac{s-1}{2s-1}(\Box_H  - 4r^2 s) + b_{s-1,k-1} \big) D^+ D^-
    + a_{s-1,k-1}D^+ D^-D^+D^-   \Big) \phi \nn\\
  &\quad  +\frac{1}{2s+1}\int \phi'
  \big( \frac{s}{2s+1}(\Box_H - 2r^2 (s+1)) + D^+D^-\big)(\Box_H + r^2 R^2\Box - r^2 (2s+1))\phi \nn\\
  &\quad  + r^2\frac{2(s -2)}{2s+1}\int  \phi' D^+ D^-\phi \ .
  \label{inner-n-3}
\end{align}
This can  be used to perform the computations in section \ref{sec:inner}, and a 
non-trivial consistency check is provided by \eq{I-G-consistency}.

\bibliographystyle{JHEP}
\bibliography{papers}

\end{document}